\documentclass[12pt]{article}
\usepackage{amssymb,latexsym,stmaryrd}
\usepackage{amsmath}
\usepackage{pstricks,pst-node,pst-coil,pst-plot}
\usepackage{hyperref}
\usepackage{amsthm}
\usepackage{cite}
\usepackage{graphicx}
\usepackage{mathtools}
\usepackage{cleveref}
\usepackage{refcount}
\usepackage{xcolor}
\usepackage{esint}
\usepackage{xfrac}
\usepackage{accents}

\usepackage[utf8]{inputenc}

\DeclareFontFamily{OT1}{pzc}{}
\DeclareFontShape{OT1}{pzc}{m}{it}{<-> s * [1.10] pzcmi7t}{}
\DeclareMathAlphabet{\mathpzc}{OT1}{pzc}{m}{it}
\usepackage{graphicx}
\usepackage{color}
\usepackage{enumitem}
\usepackage{comment}

\usepackage[normalem]{ulem}
\numberwithin{equation}{section}

\usepackage[title]{appendix}

\newtheorem{theorem}{Theorem}[section]
\newtheorem{lemma}[theorem]{Lemma}

\newtheorem{prop}[theorem]{Proposition}
\newtheorem{rmk}[theorem]{Remark}

\newtheorem{notation}{Notation}

\newcommand\R{{\mathbb R}}
\newcommand\Sphere{{\mathbb S}}

\newcommand{\photo}{P^{\:\!n}}
\renewcommand\d{\partial}

\newcommand{\definedas}{\mathrel{\raise.095ex\hbox{\rm :}\mkern-5.2mu=}}
\newcommand{\asdefined}{\mathrel{=\mkern-5.2mu}\raise.095ex\hbox{\rm :}\;}
\newcommand{\surf}{\Sigma^{n-1}}

\newcommand{\slice}{M^{n}}

\newcommand\beq{\begin{align}}
\newcommand\eeq{\end{align}}
\newcommand\ben{\begin{enumerate}}
\newcommand\een{\end{enumerate}}
\newcommand\bit{\begin{itemize}}
\newcommand\eit{\end{itemize}}

\newcommand{\Ric}{\operatorname{Ric}}

\newcommand{\go}{g_0}
\newcommand{\Deg}{\Delta_g}
\newcommand{\nana}{\nabla^2}
\newcommand{\DD}{\mathrm{D}^2}
\newcommand{\Ricg}{\operatorname{Ric}_g}
\newcommand{\HHH}{\operatorname{H}}
\newcommand{\Hg}{\HHH_g}
\newcommand{\ep}{\varepsilon}

\newcommand{\pra}{\rho} 
\newcommand{\prao}{\rho_0} 
\newcommand{\ffi}{\varphi} 
\newcommand{\Ele}{\Psi} 
\newcommand{\mm}{\mathpzc{m}} 
\newcommand{\qq}{\mathpzc{q}} 
\newcommand{\ADM}{\mu} 
\newcommand{\C}{\mathfrak{S}}

\definecolor{purple}{RGB}{100, 0, 100}

\definecolor{greenish}{HTML}{a5b555}
   \definecolor{blueish}{HTML}{6a99d3}
   \definecolor{redish}{HTML}{a91b26}
    \definecolor{yellowish}{HTML}{DCF763}

\title{Black Hole and Equipotential Photon Surface Uniqueness in $4$-dimensional Asymptotically Flat Electrostatic Electro-Vacuum Spacetimes}

\author{Stefano Borghini\thanks{stefano.borghini@unitn.it, Mathematics Department, University of Trento},
 Carla Cederbaum\thanks{cederbaum@math.uni-tuebingen.de, Mathematics Department, University of T\"ubingen}, 
 Albachiara Cogo\thanks{albachiara.cogo@math.uni-tuebingen.de, Mathematics Department, University of T\"ubingen}}

\begin{document}
\date{}
\maketitle

\begin{abstract}
We study $4$-dimensional asymptotically flat electrostatic electro-vacuum spacetimes with a connected black hole, photon sphere, or equipotential photon surface inner boundary. 
Our analysis, inspired by the potential theory approach by Agostiniani--Mazzieri, allows to give self-contained proofs of known uniqueness theorems of the sub-extremal, extremal and super-extremal Reissner--Nordstr\"om spacetimes. 
We also obtain new results for connected photon spheres and for connected photon surfaces in the extremal case.
Finally, we provide, up to a restriction on the range of their radii, the uniqueness result for connected (both non-degenerate and degenerate) photon surfaces  in the super-extremal case, not yet treated in the literature. 
\end{abstract} 

\section{Introduction}\label{sec:intro}
The celebrated static vacuum black hole uniqueness theorem first proved by Israel~\cite{Israel} states that the only asymptotically flat static vacuum $4$-dimensional spacetime containing a black hole is the Schwarzschild spacetime (of positive mass). It has since been established in various degrees of generality -- fewer technical assumptions, weaker asymptotic assumptions, allowing for multiple black holes, and for higher dimensions -- and using many different strategies of proof. More information and a (then) complete list of references can be found in the reviews \cite{He,RobinsonReview,LivRev}. More recently, new proofs have been given in \cite{Mazz,CCLP,ndimunique}, see also Simon's proof described in \cite[Appendix~A]{Raulot}.

The static vacuum black hole uniqueness theorem has also been extended to other matter models, most notably to electro-vacuum, first by Israel \cite{IsrEl} and then by others in various degrees of generality and using many different strategies of proof. In the electro-vacuum case, one needs to distinguish between non-degenerate and degenerate black holes. In the non-degenerate case or if there is a single degenerate black hole, one recovers a sub-extremal resp.\ extremal Reissner--Nordstr\"om spacetime. In the degenerate case with multiple black holes, one recovers a Majumdar--Papapetrou spacetime. We refer the interested reader to \cite{He,RobinsonReview,LivRev}, also for results for other matter models, and to \cite{bubble,Jahns,Lucietti} for more recent results in electro-vacuum. Related results have been established in the presence of  a cosmological constant, see \cite{Chrusciel_Simon,Lee_Neves,Wang,Ambrozio,BorChrMaz}.

More recently, the static (electro-)vacuum black hole uniqueness theorems have been adapted to other situations, replacing black holes by photon spheres and, more generally, by equipotential photon surfaces. Here, \emph{photon surfaces} are timelike umbilic hypersurfaces in any spacetime; in a static spacetime, a photon surface is  \emph{equipotential} if the static lapse function is (spatially) constant along each of its canonical time-slices, and it is a \emph{photon sphere} if the static lapse function is constant along it. They are \emph{non-degenerate} if $dN\neq0$ along them. In vacuum, uniqueness of photon spheres and non-degenerate equipotential photon surfaces was first shown by Cederbaum \cite{CederPhoto} and by Cederbaum and Galloway \cite{ndimunique,CedGalSurface}. It has since been generalized in the same directions as the black hole uniqueness results, including to electro-vacuum, to multiple components, and to higher dimensions by Cederbaum--Galloway, Jahns, Raulot, Yazadjiev--Lazov, see \cite{cedergal,Raulot,CedGalElec,YazLaz,Jahns,CCLP,CedCoFer,CedJaVi,cogo}. For related works, including results for other matter models, see \cite{CVE,Vol,Yazadjiev,YazaLazov2,Shoom,Rogatko,Yoshino,Yoshino2,VE1,VE2,GibbonsWarnick,Tomi, Tomi2, Koga:2020gqd,MS,PhysRevD.104.124016,Reiris,Beig_Simon}.

In this paper, we give a new proof of the static electro-vacuum black hole uniqueness theorem as well as the corresponding theorems for photon spheres and equipotential photon surfaces. We allow for both non-degenerate and degenerate black holes and for both non-degenerate and degenerate sub-extremal, extremal, and super-extremal photon surfaces (under a technical condition in the last case) but our results are restricted to a single black hole or a single photon sphere or a single equipotential  photon surface and to $4$ spacetime dimensions. We require very weak asymptotic decay assumptions. Our proof is based on a cylindrical ansatz inspired by the work by Agostiniani--Mazzieri \cite{Mazz}, see \Cref{sec:strategy}, and may be of independent interest.

\textbf{Comparison to previous uniqueness results in electro-vacuum:} For non-degenerate black holes, we do not make the technical assumption that the static lapse function regularly foliates the spacetime made in Israel's non-degenerate black hole uniqueness theorem \cite{IsrEl}. Thus, we fully recover the uniqueness results by M\"uller zum Hagen--Robinson--Seifert \cite{MRS} and by Simon \cite{SimEV}, allowing for weaker asymptotic decay. On the other hand, we do not fully recover Masood-ul-Alam's and Ruback's uniqueness results \cite{MuA,Ruback} for a priori multiple non-degenerate black holes as our strategy of proof a priori assumes that there is only a single black hole. Yet, in the single black hole case, we recover their results and weaken their asymptotic decay assumptions. We do not include magnetic charge, so we do not recover the result by Heusler \cite{Heus1}, but see \cite{SimEV} showing a reduction from the full to the purely electric case. 

In the case of degenerate black holes, we do not fully recover either of the uniqueness results by Chru\'sciel \cite{Chrusciel} who obtains uniqueness for multiple degenerate and non-degenerate black holes with charges of the same sign, nor of Heusler \cite{Heus2} who considers multiple degenerate black holes, nor those of Chru\'sciel--Tod \cite{CT} who consider the most general scenario of combinations of multiple degenerate and non-degenerate black holes, again because we only allow for a single black hole. Yet, in the single black hole case, we recover their results and weaken their asymptotic decay assumptions. 

Similarly, for photon spheres, we only recover the multiple non-degenerate black hole and photon sphere results by Cederbaum--Galloway \cite{CedGalElec} in the single photon sphere case, but allow for weaker asymptotic decay; moreover, we do not assume sub-extremality  of the photon sphere. We also recover the results, remove Israel's technical condition, and weaken the asymptotic decay assumptions of the uniqueness result by Yazadjiev--Lazov \cite{YazLaz}, and other than them do not exclude extremal photon spheres. 

Finally, for equipotential photon surfaces, we again only recover the results by Cederbaum--Jahns--Vi\v{c}\'anek-Mart\'inez \cite{CedJaVi} in the single equipotential photon surface case, but allow for weaker asymptotic decay; moreover, we do not assume sub-extremality nor non-degeneracy.  See \Cref{rem:asy} for a discussion of the asymptotic decay assumptions, \Cref{sec:extremalcomparison} for a comparison between the different definitions of sub-extremality, extremality, and super-extremality used in the above works, and \Cref{sub:degenerateequi} for a discussion on the special aspects of the degenerate case.\\[-1ex]

\textbf{This article is structured as follows:} In \Cref{sec:setup}, we will introduce the setup for the static electro-vacuum uniqueness theorems and give the relevant definitions. Our results are stated in \Cref{sec:results}. In \Cref{sec:strategy}, we explain the strategy of our proof. \Cref{sec:cyl} is dedicated to setting up the cylindrical ansatz formulation of the problem. \Cref{sec:BH} contains the proof of our black hole uniqueness theorems. \Cref{sec:photon_surface} is dedicated to proving our photon sphere and equipotential photon surface uniqueness results. We end the paper by a discussion in \Cref{sec:discussion}.\\[-1ex]

\textbf{Acknowledgements.} The work of Carla Cederbaum is supported by the focus program on Geometry at Infinity (Deutsche Forschungsgemeinschaft,  SPP 2026).
Stefano Borghini is a member of Gruppo Nazionale per l’Analisi Matematica, la Probabilit\`a e le loro Applicazioni (GNAMPA), which is part of the Istituto
Nazionale di Alta Matematica (INdAM).

\section{Setup}\label{sec:setup}
A \emph{(standard) static spacetime} is a (necessarily time-oriented) Lorentzian manifold of the form $(\R\times M^{n},-N^{2}dt^{2}+g_{0})$ for some smooth Riemannian manifold $(M^{n},g_{0})$ and some smooth positive \emph{(static) lapse function} or \emph{static potential} $N\colon M\to\R^{+}$. A (standard) static spacetime is called \emph{(standard) electrostatic} if it carries an electric potential $\Psi$ which is invariant under the action of the static Killing vector field $\d _{t}$. A (standard) electrostatic spacetime is called \emph{electro-vacuum} if it satisfies the \emph{source-free Einstein--Maxwell equations} 
\begin{align*}
\mathfrak{Ric}_{\alpha\beta}-\frac 1 2 \mathfrak R \, \mathfrak g_{\alpha \beta}&=2\left(F_{\alpha\gamma}F_\beta^{\phantom{\beta}\gamma}-\frac 1 {4} \mathfrak g_{\alpha\beta}F^{\mu\nu}F_{\mu\nu}\right),\\
\left(\operatorname{\mathfrak{div}}F\right)_{\alpha}&=0
\end{align*}
 on $\R\times M$ with respect to the Maxwell tensor\footnote{In higher dimensions, the assumed structure of the Maxwell tensor encapsulates the additional assumption that the magnetic field vanishes; however in the $4$ spacetime dimensions we are working in, one can assume this structure without loss of generality by the so-called ``electromagnetic duality'', see e.g.\cite{bubble,Jahns}.} $F=d\Ele\wedge dt$. Here $\mathfrak{Ric}$ and $\mathfrak{R}$ denote the Ricci and scalar curvature of, and $\operatorname{\mathfrak{div}}$ denotes the divergence with respect to the spacetime metric $\mathfrak{g}=-N^{2}dt^{2}+g_{0}$.
 
For a $4$-dimensional (standard) electrostatic spacetime $(\R\times M^{3},-N^{2}dt^{2}+g_{0})$, it can be seen by a straightforward computation that the source-free Einstein--Maxwell equations reduce to the \emph{electrostatic electro-vacuum equations}
\begin{align}\label{eq:pb}
\begin{dcases}
N\operatorname{Ric}&=\mathrm{D}^2 N-\frac{2}{N}\,d\Ele\otimes d\Ele+\frac{1}{N}\vert \mathrm{D}\Ele\vert^2 g_0,\\
\Delta N&=\frac{1}{N}\,|\mathrm{D}\Ele|^2,\\
\Delta\Ele&=\frac{1}{N}\,\langle\mathrm{D}N\,|\,\mathrm{D}\Ele\rangle,
\end{dcases}
\end{align}
on $M$, where $\mathrm{D}$, $\mathrm{D}^{2}$, $\Delta$, and $\operatorname{Ric}$ denote the covariant derivative, covariant Hessian, Laplacian, and Ricci curvature with respect to $g_{0}$, respectively. Slightly abusing notation by treating $\Ele$ as a function $\Ele\colon M\to\R$, we will call tuples $(M^{3},g_{0},N,\Ele)$ as described above \emph{electrostatic systems}, and \emph{electro-vacuum} provided they solve \eqref{eq:pb}. 

As we are interested in boundary value problems for static spacetimes solving the source-free Einstein--Maxwell equations, we need to briefly discuss how to include a boundary in the above considerations. In fact, we will treat two different cases of boundary conditions which each require somewhat different approaches: on the one hand, we want to allow for spacetime boundary which is a Killing horizon; on the other hand, we want to treat boundaries which arise as equipotential photon surfaces. Before we address how to include such boundaries in the above considerations, let us give the relevant definitions.

Let $(\mathfrak{L}^{n+1},\mathfrak{g})$ be a smooth Lorentzian manifold with boundary $\partial\mathfrak{L}$ such that the interior $(\mathfrak{L}\setminus\d \mathfrak{L},\mathfrak{g})$ is (isometric to) a standard static spacetime $(\R\times M^{n},-N^{2}dt^{2}+g_{0})$ and $\d \mathfrak{L}$ is a null hypersurface. Then $\d \mathfrak{L}$ is called a \emph{(static) Killing horizon} if the static Killing vector field $\d _{t}$ smoothly extends to $\d \mathfrak{L}$ and the positive lapse function $N$ of $\mathfrak{L}\setminus\d \mathfrak{L}$ smoothly extends to zero on $\d \mathfrak{L}$. Deviating from but equivalent to the standard definition, we say that a Killing horizon $\d \mathfrak{L}$ is \emph{non-degenerate} if $dN\neq0$ along $\d \mathfrak{L}$, otherwise it is \emph{degenerate}, see \cite{KW}. For this distinction it is relevant that $\vert dN\vert_{\mathfrak{g}}$ is constant on each connected component of a Killing horizon (see e.g.~\cite[Theorem~7.4]{Heusler_Blackholes}). We will slightly abuse notation and call a Lorentzian manifold $(\mathfrak{L}^{n+1},\mathfrak{g})$ with (degenerate or non-degenerate) Killing horizon boundary $\d \mathfrak{L}$ \emph{standard static} as well, and \emph{standard electrostatic} if the electric field $\Psi$ smoothly extends to $\d \mathfrak{L}$.

If $(\mathfrak{L}^{n+1},\mathfrak{g})$ is a standard (electro-)static spacetime with non-degenerate Killing horizon boundary, the \emph{spatial manifold} $M^{n}$ is well-known to possess a smooth boundary $\d  M$ such that $\d \mathfrak{L}$ is diffeomorphic to $\R\times\d  M$. Furthermore, the metric $g_{0}$, the lapse function $N$, and the electric potential $\Psi$ all smoothly extend to $\d  M$, with $N=0$ there, hence $(M^{n},g_{0},N,\Psi)$ represents an \emph{electrostatic system with boundary}. We will call $\d  M$ a \emph{time-slice} of the non-degenerate Killing horizon $\d \mathfrak{L}$ and a \emph{(non-degenerate) horizon}. From non-degeneracy, we also know that $dN\neq0$ holds on $\d  M$ and in fact that $\vert dN\vert>0$ is constant on connected components of $\partial M$. Moreover, the electro-vacuum equations \eqref{eq:pb} extend to $\d  M$ and hence by tracing the first equation in \eqref{eq:pb} one finds that $2\vert {\rm D}\Psi\vert^{2}=N^{2}\operatorname{R}=0$ holds on $\d  M$, where $\operatorname{R}$ denotes the scalar curvature of $g_{0}$. In particular, $\Psi$ is constant along connected components of non-degenerate Killing horizons. Moreover, ${\rm D}^{2}N=0$ must hold on $\d  M$ by \eqref{eq:pb}. This implies that $\d  M$ is totally geodesic in $(M,g_{0})$.

If, on the other hand, $(\mathfrak{L}^{n+1},\mathfrak{g})$ is a standard (electro-)static spacetime with degenerate Killing horizon boundary, one still obtains that the spatial manifold $M^{n}$ has a smooth boundary $\d  M$ to which $N$ and $\Psi$ extend smoothly, with $N=0$ and $dN=0$ on $\d  M$. However, the metric $g_0$ does not extend smoothly at the boundary: in fact, it is known from~\cite[Lemma~3]{KW} that geodesics do not reach the degenerate horizon in finite time. This means that a \emph{(degenerate) horizon} $\partial M$ should be thought of as a (cylindrical) end of the manifold and not as a proper boundary. The peculiar case of the degenerate horizon is treated more in details in \Cref{sub:degenerate}.

Having settled the case of standard electrostatic spacetimes with Killing horizon boundary, let us now move on to consider equipotential photon surface and photon sphere inner boundaries. Again, we begin by giving the relevant definitions: first, a smooth timelike (connected) hypersurface $\photo\hookrightarrow\mathfrak{L}^{n+1}$ in a smooth Lorentzian manifold $(\mathfrak{L}^{n+1},\mathfrak{g})$ is called a \emph{photon surface} if every null geodesic initially tangent to $\photo$ remains tangent to $\photo$ as long as it exists or in other words if $\photo$ is \emph{null totally geodesic}. This definition readily extends if $(\mathfrak{L}^{n+1},\mathfrak{g})$ is a smooth Lorentzian manifold with boundary and $\photo\subset\partial\mathfrak{L}$. It will be useful to know that being a null totally geodesic timelike hypersurface is equivalent to being a totally umbilic timelike hypersurface:
\begin{prop}[\!\!\!\!{\cite[Theorem II.1]{CVE}, \cite[Proposition 1]{Vol}}]\label{prop:umbilic}
Let $(\mathfrak{L}^{n+1},\mathfrak{g})$ be a Lorentzian manifold and $\photo\hookrightarrow\mathfrak{L}^{n+1}$ an embedded, timelike hypersurface. Then $\photo$ is a photon surface if and only if it is \emph{totally umbilic}, that is, if and only if its second fundamental form is pure trace.
\end{prop}

It is easy to see from its purely algebraic proof that \Cref{prop:umbilic} readily applies in case $\photo\subset\partial\mathfrak{L}$. In the context of standard electrostatic spacetimes, we will use the following definitions of photon spheres and equipotential photon surfaces from \cite{CedJaVi}, see also \cite{CDiss,CederPhoto,CVE,CedGalSurface,YazLaz,CedGalElec,Jahns}:  a photon surface $\photo\hookrightarrow\mathfrak{L}^{n+1}$ in an electrostatic spacetime $(\mathfrak{L}^{n+1},\mathfrak{g},\Psi)$  is called a \emph{photon sphere} if $\photo = \R \times\surf$ for some smooth hypersurface $\surf\hookrightarrow \slice$ and if the lapse function $N$, the electric potential $\Psi$ and the length of its derivative $\vert d\Psi\vert_{\mathfrak{g}}$ are constant\footnote{In fact, the defining assumption that $\vert d\Psi\vert_{\mathfrak{g}}$ must be constant on each canonical time-slice of an equipotential photon surface or fully constant on a photon sphere is automatically satisfied in the connected setting we are studying. This can be seen by realizing that it is neither used in the proof of \Cref{pro:simplification} nor in the proof of \cite[Proposition 5.5]{CedJaVi} by \cite[Remark 5.11]{CedJaVi} (which we state as \Cref{thm: ps_properties} below).\label{foot:redundant}} along $P^n$. More generally, a photon surface $\photo\hookrightarrow\mathfrak{L}^{n+1}$ in an electrostatic spacetime $(\mathfrak{L}^{n+1},\mathfrak{g},\Psi)$  is called \emph{equipotential} if the static lapse $N$, the electric potential $\Psi$ and the length of its derivative $\vert d\Psi\vert_{\mathfrak{g}}$ are constant\textsuperscript{\ref{foot:redundant}} on each \emph{(canonical) time-slice} $\Sigma^{n-1}(t_{0})\definedas P^n\,\cap\,\{t=t_{0}\}$. Just as for Killing horizons, an equipotential photon surface $\photo$ (and in particular a photon sphere) is called \emph{non-degenerate} if $dN\neq0$ on $\photo$, otherwise it is called \emph{degenerate}. Photon spheres are always non-degenerate, see \Cref{sub:degenerateequi}. 

Photon spheres can naturally arise as boundary components of standard electrostatic spacetimes as they are warped products; in order to allow for general equipotential photon surface inner boundary, following \cite[Remark 5.18]{CedJaVi} we will abuse standard terminology and also call a spacetime $(\mathfrak{L}^{n+1},\mathfrak{g})$ with timelike boundary $\partial\mathfrak{L}$ \emph{standard static} if it is the closure of an open subset of a standard static spacetime. In this case, the \emph{(canonical) time-slices} $\{t=t_{0}\}$ of a standard electrostatic spacetime $(\mathfrak{L},\mathfrak{g},\Psi)$ are electrostatic systems $(M(t_{0}),g_{0},N,\Psi)$ with boundary $\partial M(t_{0})$ on which $N_{0}\definedas N\vert_{\partial M(t_{0})}>0$ and $\Psi_{0}\definedas\Psi\vert_{\partial M(t_{0})}\in\R$.

The following fact\footnote{adjusted to our notation and reduced to what is relevant to us.} for non-degenerate equipotential photon surfaces from \cite[Proposition 5.5]{CedJaVi} will be relevant for our analysis.
\begin{prop}[\!\!{\cite[Proposition 5.5]{CedJaVi}}]\label{thm: ps_properties}
Let $(\mathfrak{L}^{n+1},\mathfrak{g},\Psi)$ be a standard electrostatic electro-vacuum spacetime, $\photo\hookrightarrow\mathfrak{L}^{n+1}$ a non-degenerate equipotential photon surface (with $\photo\subseteq\partial\mathfrak{L}$ permitted). Then each canonical time-slice $\surf(t_{0})$ of $\photo$ is totally umbilic in $\slice(t_{0})$ and has constant mean curvature $\HHH_{\nu}$ and constant normal derivative $\nu(N)$ with respect to any unit normal $\nu$ to $\surf(t_{0})$ in $(\slice(t_{0}),g_{0})$, and constant scalar curvature $\operatorname{R}^{\surf(t_{0})}$ satisfying the \emph{photon surface identity}
\begin{align}\label{foundtheconstant}
\operatorname{R}^{\surf(t_{0})}&= \frac {2\nu(\Psi)^2}{N^2} + \frac {2\HHH_{\nu}\nu(N)}{N} +\frac{(n-2)\HHH_{\nu}^{2}}{(n-1)}
\end{align}
on $\surf(t_{0})$.
\end{prop} 
Note that \eqref{foundtheconstant} is invariant under the choice of unit normal $\nu$. Note also that $\nu(\Psi)^{2}=\vert{d\Psi}\vert_{\mathfrak{g}}^{2}=\vert{\rm D}\Psi\vert^{2}$ and $\nu(N)^{2}=\vert dN\vert_{\mathfrak{g}}^{2}=\vert{\rm D}N\vert^{2}$ as $\Psi$ and $N$ are assumed to be invariant under the action of the static Killing vector field and because both $N$ and $\Psi$ are assumed to be constant along each canonical time-slice $\surf(t_{0})$. 

As in \cite[Remark 5.7]{CedJaVi}, in order to make our technique effective, it will be convenient to choose the unit normal $\nu$ pointing to the asymptotically flat end of $(M(t_{0}),g_{0},N,\Psi)$ (see below) and introduce the quantity
\begin{align}
 c &\definedas \frac{2\nu(N)}{N \HHH_{\nu}},
 \end{align}
provided that $\HHH_{\nu}\neq0$, as this allows us to rewrite \eqref{foundtheconstant} as the two identities
\begin{align}\label{scalarMeanCurv}
 \operatorname{R}^{\Sigma}&= \frac{2\,|\mathrm{D}\Ele|^2}{N^2} + \left(c + \frac{1}{2} \right)\HHH_{\nu}^2,\\\label{normalLapse}
    2\nu(N) &=  cN\HHH_{\nu}.
\end{align} 
It is shown in \cite[Proposition~2.4]{CedGalElec} and \cite[Corollary 5.6]{CedJaVi} that $c=1$ holds on time-slices of photon spheres. Furthermore, \cite[Lemma 5.9]{CedJaVi} tells us that if $c=1$ holds on one time-slice $\surf(t_{0})$ of a non-degenerate equipotential photon surface then $\R\times \surf(t_{0})$ is a photon sphere. Concerning the condition that $\HHH_{\nu}\neq0$, it follows from the proof of \cite[Theorem 5.22]{CedJaVi} applied to the case of a connected non-degenerate equipotential photon surface boundary $\partial M$ that $\HHH_{\nu}>0$ holds on $\partial M$. It follows that $c$ is well-defined and that it satisfies $c > 0$ if and only if $\nu(N) > 0$ (for which, in words, we say that the photon surface is \emph{outward directed}). In \Cref{sec:photon_surface}, we independently re-prove this result as part of our uniqueness proof.

Summarizing the above considerations, in all cases we are interested in, we will treat electrostatic, electro-vacuum systems $(M,g_{0},N,\Psi)$ or in other words Riemannian manifolds $(M,g_{0})$ carrying a positive lapse function $N$ and an electric potential~$\Psi$. The manifold $M$ will always have a boundary $\partial M$ to which $N$ and $\Psi$ smoothly extend as constants $N\vert_{\partial M}=N_{0}\geq0$, $\Psi\vert_{\partial M}=\Psi_{0}\in\R$. In the non-degenerate horizon as well as in the equipotential photon surface case, the metric $g_{0}$ will also smoothly extend to the boundary, making $(M,g_{0},N,\Psi)$ an \emph{electrostatic, electro-vacuum system with boundary}. In all our results, we will always assume that the boundary $\d M$ is nonempty and {\em connected}. This hypothesis allows to simplify the setup as will be discussed in detail in \Cref{sub:simplification}. In our analysis, we will focus on electrostatic systems $(M^{3},g_0,N,\Ele)$ solving \eqref{eq:pb} which are {\em asymptotically flat}: that is, we assume that there is a compact set $K\supset\d M$ such that $M\setminus K$ is diffeomorphic to $\R^3\setminus B$ for some closed ball $B$ and such that, with respect to the coordinates $(x^{i})$ induced by this chart, one has
\begin{align}
\begin{split}\label{eq:asymptotics}
N&= 1-\frac{\ADM}{|x|}+o_2(|x|^{-1}),\\
 \Ele&=o(1),\\
 {(\go)}_{ij}&= \left(\delta_{ij}+\eta_{ij}\right)dx^i\otimes dx^j,\\
 \eta_{ij}&=o_1(1)   
\end{split}
\end{align}
as $|x|\to\infty$, where $\ADM\in\mathbb{R}$ is a constant and called the \emph{mass (parameter)}. When $\mu=0$, we ask in addition that
\begin{align}\label{eq:asymmu0}
\Psi&=\frac{\kappa}{\vert x\vert}+o_{2}(\vert x\vert^{-1})
\end{align}
as $|x|\to\infty$ for some constant $\kappa\in\R$ called the \emph{charge (parameter)}. This extra condition will be important to get enough information on the asymptotic behavior of the level sets of the lapse function $N$ in \Cref{sub:neg_mass}.

Last but not least, we note that this definition of asymptotic flatness also applies when $\partial M$ is a time-slice of a degenerate Killing horizon as by the above discussion $\partial M$ then has a pre-compact tubular neighborhood.

\begin{rmk}[Asymptotic decay]\label{rem:asy}
In the standard definition of asymptotic flatness, one usually requires stronger asymptotic conditions, namely that $\eta_{ij}=O_2(|x|^{-\frac{1}{2}-\varepsilon})$ for some $\varepsilon>0$ and that the scalar curvature $\operatorname{R}$ of $\go$ is integrable. Under these additional assumptions, it can be seen by a standard computation that $\ADM$ coincides with the ADM-mass. However, for our purposes, the asymptotic assumptions in~\eqref{eq:asymptotics} (combined with \eqref{eq:asymmu0} when $\mu=0$) will be enough. Our decay assumptions are very weak when compared with the other electrostatic electro-vacuum uniqueness results discussed in \Cref{sec:intro}; these assume that the metric $g_{0}$ is asymptotically flat with $\varepsilon=\frac{1}{2}$ (\cite{YazLaz}) or asymptotic to the Reissner--Nordstr\"om metric of mass $\mu$ and charge $q$ (\cite{CedGalElec,Jahns,CedJaVi}) and that $\Psi$ is asymptotic to the Reissner--Nordstr\"om potential $\frac{q}{\vert x\vert}$, see \eqref{eq:RN}. The same is true for the static vacuum uniqueness results discussed in \Cref{sec:intro} with the notable exceptions of \cite{Mazz,CedCoFer,CCLP} who make the same assumption on the decay of $N$ as we make. Moreover, in \cite{Mazz} (and consequently in its application in \cite{CedCoFer}), for $n=3$ spatial dimensions, it is assumed that $\eta_{ij}=o_{2}(\vert x\vert^{-\frac{1}{2}})$. On the other hand, \cite{CCLP} makes the same asymptotic assumptions we make. Note however that it is conceivable that, when combined with \eqref{eq:pb}, our asymptotic decay can be boot-strapped to stronger decay assumptions as for example in~\cite{KM}.
\end{rmk}

The following remark will be useful for our strategy of proof, see \Cref{sec:strategy}.
\begin{rmk}[Geodesic and metric completeness]\label{rmk:complete}
For later convenience, we note that asymptotically flat electrostatic systems, with or without boundary, are necessarily metrically complete and geodesically complete (up to the boundary) with at most finitely many boundary components which are necessarily all closed, see for example~\cite[Appendix]{CGM} (where stronger asymptotic flatness assumptions are made but not used in the proof). This continues to apply in the degenerate horizon case as cylindrical ends are also geodesically complete up to the boundary by similar reasoning. 

Note however that while $\partial M$ has a pre-compact tubular neighborhood, the topology on this neighborhood does not coincide with the topology induced by the Riemannian distance function $d_{g_{0}}$ induced by $g_{0}$ in the degenerate horizon case as this tubular neighborhood/cylindrical end is necessarily unbounded with respect to $d_{g_{0}}$.
\end{rmk}

The most important asymptotically flat, electrostatic, electro-vacuum spacetimes in view of this paper are the \emph{Reissner--Nordstr\"om spacetimes} characterized by
\begin{align}
\begin{split}\label{eq:RN}
M&=(r_0,\infty)\times\mathbb{S}^2,\\
 g_0&=\frac{dr\otimes dr}{N^2}+r^2 g_{\mathbb{S}^2},\\
 N&=\sqrt{1-\frac{2m}{r}+\frac{q^2}{r^{2}}},\\
  \Ele&=\frac{q}{r}
\end{split}\end{align}
for fixed parameters $m,q\in\R$ and $r_{0}\geq0$ chosen such that $N$ is well-defined for all $r> r_0$. Here, $m$ is called the \emph{mass} and $q$ is called the \emph{charge} of the Reissner--Nordstr\"om solution. When $q=0$, the Reissner--Nordstr\"om solution of mass $m$ becomes the more widely known Schwarzschild spacetime of mass $m$, while $m=q=0$ corresponds to the Minkowski spacetime. The behaviour of the Reissner--Nordstr\"om solution and the maximal interval $(r_0,\infty)$ on which it is well-defined depend on the relation between $m$ and $q$. In particular, one distinguishes three cases:
\begin{itemize}
    \item The {\em sub-extremal case}, when $m>|q|$: In this case the lapse function $N$ vanishes for two positive radii whose maximum is $r_{m,q}\definedas m + \sqrt{m^2 - q^2}$. Hence the maximal interval of definition of our solution is $(r_{m,q},\infty)$. It can be seen by switching to Gaussian null or 
to Kruskal--Szekeres coordinates that the spacetime can be smoothly  extended beyond $\{r=r_{m,q}\}$ and possesses a non-degenerate Killing horizon at $r=r_{m,q}$. The sub-extremal Reissner--Nordstr\"om spacetime of mass $m$ and charge $q$ models the exterior region of an isolated electrically charged black hole in static equilibrium, surrounded by electro-vacuum.
\item The {\em extremal case}, when $m=|q|\neq0$: In this case the lapse function $N$ vanishes for the single value $r=r_{m,q}=m$, so the maximal interval of definition of our solution is $(r_{m,q}=m,\infty)$. It can be seen by switching to Gaussian null coordinates that the spacetime can be smoothly extended to $\{r=m\}$ and possesses a degenerate Killing horizon at $r=m$. The extremal Reissner--Nordstr\"om spacetime of mass $m$ and charge $q$ models the exterior region of an isolated extremal black hole in equilibrium surrounded by electro-vacuum.
 \item The {\em super-extremal case}, when $m<|q|$: In this case, the solution is well-defined on $(0,\infty)$. Notice that we cannot extend the spacetime to $r=0$. The lapse function $N$ never reaches the value zero, so there is no (static) Killing horizon or in other words no black hole, hence $r=0$ corresponds to a ``naked singularity''.
 \end{itemize}

\begin{rmk}[Inner boundary in Reissner--Nordstr\"om spacetimes]
In view of our asymptotic flatness definition, one should instead consider closed radius intervals in the definition of Reissner--Nordstr\"om solutions. In the sub-extremal case, one can rewrite the Reissner--Nordstr\"om electrostatic system in isotropic coordinates and suggestively consider it on $[r_{m,q},\infty)$. Switching instead to Gaussian null coordinates as described above, one can suggestively rewrite the extremal Reissner--Nordstr\"om electrostatic system on $[m,\infty)$, keeping in mind that the metric does not smoothly extend to the cylindrical end ``at'' $m$. For the super-extremal case however, one needs to restrict to an interval $[r_{0},\infty)$ for some $r_{0}>0$ to obtain an asymptotically flat system. This aligns with the fact that the super-extremal Reissner--Nordstr\"om spacetimes contain naked singularities.
\end{rmk}

It is well-known (see e.g. \cite[Corollary 3.2]{CedJaVi}) that the Reissner--Nordstr\"om spacetime of mass $m$ and charge $q$ \begin{itemize}
\itemsep0em
\item has a unique photon sphere at $r_{*} =  \frac{3m + \sqrt{9m^2 - 8q^2}}{2}$ when $m\geq\vert q\vert$, $m\neq0$,
\item has precisely two photon spheres at $r_{*,\pm}=\frac{3m \pm \sqrt{9m^2 - 8q^2}}{2}$ when $\vert q\vert>m>\frac{2\sqrt{2}}{3}\vert q\vert$,
\item has a unique photon sphere at $r_{*}=\frac{3m}{2}$ when $m=\frac{2\sqrt{2}}{3}\vert q\vert$,
\item and no photon spheres otherwise.
\end{itemize}
It is shown in \cite[Theorems 3.7, 3.9, 3.10]{CedJaVi} that, for each radius $r_{1}>r_{0}$, $r_{0}$ the radius corresponding to the maximal existence interval in \eqref{eq:RN}, all Reissner--Nordstr\"om spacetimes have spherically symmetric equipotential photon surfaces going through $\{r=r_{1}\}$. Moreover, \cite[Corollary 2.25]{CedJaVi} asserts that there are no other equipotential photon surfaces than spherically symmetric ones in Reissner--Nordstr\"om spacetimes except all timelike planes in the Minkowski spacetime. Finally, we recall that $dN\neq0$ holds in all Reissner--Nordstr\"om spacetimes except on $\{r=\frac{q^{2}}{m}\}$ when  $0<m<\vert q\vert$. Thus spherically symmetric/equipotential photon surfaces in Reissner--Nordstr\"om spacetimes are always non-degenerate unless they cross $\{r=\frac{q^{2}}{m}\}$ when $0<m<\vert q\vert$.

\section{Results}\label{sec:results}
For our first result, we will restrict our attention to the case $N_0=0$, that is, we will assume that the boundary is a time-slice of a Killing horizon, or a \emph{horizon} for short. Hence, our first result is a version of the well-known electrostatic electro-vacuum black hole uniqueness theorem for a single horizon, see \Cref{sec:intro} for a detailed comparison with other proofs of this result.
\begin{theorem}
\label{thm:BH}
Let $(M,\go,N,\Ele)$ be an asymptotically flat electrostatic system  with mass $\mu$ solving \eqref{eq:pb}. Suppose that $M$ has a connected boundary $\d M$ arising as a time-slice of a Killing horizon, i.e., in particular satisfies $N\vert_{\partial M}=N_0=0$. Then $(M,\go,N,\Ele)$ is isometric to a sub-extremal or extremal Reissner--Nordstr\"om solution of mass $\mu$ and some charge $q$ and interval of definition $(r_{\mu,q},\infty)$. In particular, if the horizon $\d M$ is non-degenerate, this Reissner--Nordstr\"om solution is sub-extremal, $\mu>\vert q\vert$, while it is extremal, $\mu=\vert q\vert\neq0$, if the horizon $\d M$ is degenerate.
\end{theorem}

Our second result shows \emph{(electro-)static equipotential photon surface uniqueness} in electro-vacuum. Exactly as in the first result, it shows that an asymptotically flat electrostatic electro-vacuum spacetime that admits an equipotential photon surface as its inner boundary must be isometric to a Reissner--Nordstr\"om spacetime of the same mass. As discussed in more detail in \Cref{sec:intro}, our theorem recovers previously known results but also extends beyond what was previously known.

\begin{theorem}\label{thm:photon_surface}
Let $(M,\go,N,\Ele)$ be an asymptotically flat electrostatic system with boundary $\partial M$ of mass $\mu$ (and of charge $\kappa$ if $\mu=0$) solving \eqref{eq:pb}. Suppose that $\d M$ is a connected time-slice of an equipotential photon surface, hence in particular satisfies that $N\vert_{\partial M}=N_0 > 0$, $\Psi\vert_{\partial M}=\Ele_0$ are constants. If $N_{0}=1$ and $\Psi_{0}=0$ both hold, then $(M,\go,N,\Ele)$ can be isometrically embedded into the Euclidean space. 
In all the other cases, assuming that $N_0^2 \geq |1-\Psi_0^2|$ holds if $N_0^2 > (1-|\Psi_0|)^2$,  the system $(M,\go,N,\Ele)$ can be isometrically embedded into a Reissner--Nordstr\"om solution of mass $\mu$ and some charge $q$, with $q=\kappa$ when $\mu=0$. In particular, this Reissner--Nordstr\"om spacetime is
\begin{itemize}
    \item \emph{sub-extremal} if $N_0^2 < (1-|\Psi_0|)^2$,
    \item \emph{extremal} if $N_0^2 = (1-|\Psi_0|)^2$, and
    \item \emph{super-extremal} if $N_0^2 > (1-|\Psi_0|)^2$. 
\end{itemize}
\end{theorem}

\begin{rmk}\label{rmk: reductionToVacuum}
Notice that when $\Psi_0 = 0$, then $\Psi\equiv0$ in view of \Cref{pro:simplification} and hence the problem reduces to the analysis of the vacuum static system $(M, g_0, N)$ provided in \cite{CedCoFer}, which gives the isometric embedding of $(M, g_0, N)$ into a Schwarzschild solution with positive mass $\mu > 0$ if $N_0 < 1$, a Schwarzschild solution with negative mass $\mu < 0$ if $N_0 > 1$, or into the Euclidean space if $N_0 = 1$. 
Therefore, Theorem~\ref{thm:photon_surface} can be interpreted as an extension of \cite{CedCoFer} to the electro-vacuum case. For a more extended discussion, see \Cref{sub:simplification}.
\end{rmk}

\begin{rmk}\label{rem:condition}
The assumption that $N_0^2 \geq |1-\Psi_0^2|$ holds if $N_0^2 > (1-|\Psi_0|)^2$ is a technical condition needed for our method of proof. In the Reissner--Nordstr\"om model case, this corresponds to excluding equipotential photon surfaces with $r<m$ in super-extremal Reissner--Nordstr\"om spacetimes of positive mass $m>0$. See \Cref{sec:condition} for a discussion of why this condition arises in our approach.
\end{rmk}

\begin{rmk}
As it can be seen from the proof of \Cref{thm:photon_surface}, instead of assuming that $\partial M$ arises as a time-slice of an equipotential photon surface, it suffices to assume that $N\vert_{\partial M}=N_{0}$, $\Psi\vert_{\partial M}=\Psi_{0}$, that $\partial M$ is umbilic and has constant mean curvature $\HHH_{\nu}$ and scalar curvature $\operatorname{R}^{\partial M}$, and constant normal derivatives $\nu(\Psi)$ and $\nu(N)$ such that the photon surface identity \eqref{foundtheconstant} holds on $\partial M$. Moreover, in case $dN$ does not vanish entirely on $\partial M$, it follows from the proof of \cite[Proposition 5.5]{CedJaVi} that assuming $N\vert_{\partial M}=N_{0}$, $\Psi\vert_{\partial M}=\Psi_{0}$, umbilicity of $\partial M$, constancy of $\nu(\Psi)$ and constancy of either $\nu(N)$ or $\HHH_{\nu}$ imply all the other assumptions, i.e., constancy of $\HHH_{\nu}$ respectively $\nu(N)$, constancy of $\operatorname{R}^{\partial M}$, and \eqref{foundtheconstant} (see also \Cref{sub:degenerateequi}).
\end{rmk}

\section{Strategy of proof and comments}\label{sec:strategy}
The proofs of our results are based on the cylindrical ansatz strategy introduced by Agostiniani and Mazzieri in~\cite{AgoMazz_first}. Since then, the same strategy has found several applications to potential theory~\cite{Ago_Fog_Maz-2,MazzAgo, Fog_Maz_Pin}, manifolds with non-negative Ricci curvature~\cite{Ago_Fog_Maz-1,Ben_Fog_Maz} and static spacetimes~\cite{Mazz,BorMaz-collection,BorMaz-I,BorMazII,CedCoFer}. Other related ideas and methods have been developed in~\cite{AgoBorMaz,Bor_staticnegative,BorMasMaz,
Oronzio-PMT,Oronzio-Penrose}. Here we review and comment on the method in some detail.

To set up the cylindrical ansatz, the first step is to write down the radial coordinate $r$ of the model solution (the Reissner--Nordstr\"om system~\eqref{eq:RN} in our case) as a function of the lapse function $N$. This is of course possible only as long as $N$ is a monotonic function of the radius $r$. For sub-extremal and extremal Reissner--Nordstr\"om solutions, this is indeed the case: in fact, it can be checked by a direct computation that $N$ is a strictly increasing function of the radius $r$ as long as $m\geq|q|$ unless $m=q=0$. Taking the inverse of the relation between $r$ and $N$ we can then write $r=\rho(N)$ for a suitable \emph{pseudo-radial} function $\rho$.

The proofs of \Cref{thm:BH} and of the sub-extremal, extremal, and $m\leq0$-super-extremal cases of \Cref{thm:photon_surface} then loosely follow the original strategy of~\cite{MazzAgo}. The equipotential photon surface case is then handled similarly to the application of the results of \cite{MazzAgo} in \cite{CedCoFer}: first, one exploits the Bochner formula to produce a vector field depending on $\rho$ whose divergence is non-negative and equal to zero if and only if the solution is isometric to the Reissner--Nordstr\"om solution, see \Cref{sub:Bochner}. Then, to conclude the proof, one proves that the equality is necessarily saturated using either the properties of a horizon or of an equipotential photon surface. This is the objective of \Cref{sec:BH} in the black hole case, and of \Cref{sub:proof_photon_standard,sub:neg_mass} in the sub-extremal and extremal and in the $m\leq0$-super-extremal equipotential photon surface cases, respectively. We emphasize however that there are a couple of technical but important points that need to be addressed in these proofs: 
\begin{enumerate}
\itemsep0em
\item There are two natural choices of pseudo-radial functions and we will need to make suitable choices in all cases, giving suitable asymptotic behavior and justify their well-definedness and smoothness in each case. 
\item In order to work with the weaker notion of asymptotic flatness we use, see \eqref{eq:asymptotics} and \eqref{eq:asymmu0}, we need a more careful analysis of the asymptotics in the cylindrical ansatz. This is performed in \Cref{sub:asymptotics} for $\mu\neq0$ and in \Cref{subsub:mm=0} for $\mu=0$.
\item In the proof of the extremal case of \Cref{thm:BH}, one needs some refined analysis near the degenerate horizon, see \Cref{sub:degenerate}.
\end{enumerate}

On the other hand, in the $m>0$-super-extremal case, the lapse function of the Reissner--Nordstr\"om solution is \emph{not} a monotonic function of the radial coordinate anymore: $N$ is strictly decreasing when $r<\frac{q^2}{m}$, has a global minimum at $r=\frac{q^2}{m}$, and is strictly increasing when $r>\frac{q^2}{m}$. Therefore, it is only possible to write the radial coordinate $r$ as a function of $N$ if we restrict our attention to $r\leq \frac{q^2}{m}$ or to $r\geq \frac{q^2}{m}$. In order to perform the conformal change prescribed by the cylindrical ansatz on a general manifold in the $m>0$-super-extremal case, we hence need a more involved procedure: the first step will be to show that, given an $m>0$-super-extremal solution $(M,g,N,\Ele)$ of~\eqref{eq:pb}, the set $\C\subset M$ consisting of the global minima of the function $N$ either coincides with $\partial M$ or divides the manifold into two pieces, a compact one and a non-compact one. The first case is treated in \Cref{subsec:superextremalnocrit} and still roughly follows the same strategy as the other cases. In the second case, we perform different conformal changes in the two regions, i.e., use different pseudo-radial functions, and analyze them separately, see \Cref{subsec:critnonemptyN0large}. As a result of our analysis, in both regions we produce quantities that are monotonic along the level sets of the lapse function~$N$.
 
Remarkably, these monotonic quantities coincide on $\C$ and the monotonicities combine (see \Cref{pro:integral_Hg_superext}, and in particular formulas~\eqref{eq:intinequality_Gamma1} and~\eqref{eq:intinequality_Gamma2}). This allows us to compare the value of our monotonic quantities at infinity and on the photon surface, obtaining valuable insights into the geometry of the photon surface and allowing us to conclude the rigidity of the super-extremal Reissner--Nordstr\"om solution, at least for photon surfaces of radius $r\geq m$, see also \Cref{rem:condition}.

The strategy outlined above resembles the one explored recently by Borghini--Mazzieri in~\cite{BorMazII} for the study of the vacuum case with positive cosmological constant. In that paper, the authors divide the manifold along the set of the global maxima and then apply the cylindrical ansatz strategy separately in every region. There are however some interesting and crucial differences. The first one is that, while in the present paper we prove that the set of local minima $\C$ separates the manifold in exactly two regions, in the setting of~\cite{BorMazII} it is not known how many regions are produced when cutting along the set of the global maxima: there are models where cutting along the set of maxima produces two regions (Schwarzschild--de Sitter and Nariai solutions), but it is also possible for the cut to produce just a single region (de Sitter solution) and there is currently no way to rule out the case of more than two regions. Furthermore, while in~\cite{BorMazII} the authors also manage to exploit the cylindrical ansatz strategy to  produce monotonic quantities (see~\cite[Corollary~3.5]{BorMazII}), these quantities do not combine together across the set of maxima. The authors thus had to resort to other methods in order to obtain rigidity statements. 

\subsection{Simplifications in our framework}\label{sub:simplification}

In this subsection, we will discuss some preliminary observations, leading to a helpful rewriting of the rigidity/uniqueness problem. 
First, we observe that it is not restrictive to limit our attention to the case where the boundary value $\Ele_0$ of $\Psi$ is positive. 
In fact, notice that if $(M,\go,N,\Ele)$ is a solution to~\eqref{eq:pb}, then $(M,\go,N,-\Ele)$ is a solution as well. As a consequence, up to changing the sign of $\Ele$, we can assume that $\Ele_0$ is non-negative. As already pointed out in \Cref{rmk: reductionToVacuum}, in the case $\Ele_0=0$ , the rigidity/uniqueness problem reduces to the vacuum case which has already been studied in depth in several papers, see \Cref{sec:intro}. In particular, when $N_{0}\neq1$, the results by Cederbaum--Cogo--Fehrenbach~\cite{CedCoFer} and in higher dimensions by Cederbaum--Cogo--Leandro--Paulo dos Santos~\cite{CCLP} prove the vacuum case of our \Cref{thm:photon_surface}, under the same decay assumptions. As mentioned in \Cref{rmk: reductionToVacuum}, when $N_0 = 1$, the result follows the same argument provided in \cite{CedCoFer}. The aim of this work is that of employing and expanding on the cylindrical ansatz technique presented in the vacuum case by Agostiniani and Mazzieri~\cite{Mazz} for a connected (and necessarily non-degenerate, see e.g. \cite{KW}) horizon and by Cederbaum--Cogo--Fehrenbach~\cite{CedCoFer} for a connected non-degenerate equipotential photon surface to deal with the case where $\Ele$ is non-trivial. 

In fact, our asymptotic analysis in \Cref{sub:asymptoticsimplification,sub:asymptotics} can be used to extend the vacuum uniqueness results by \cite{Mazz} and hence also those by \cite{CedCoFer} to the weaker asymptotic assumptions \eqref{eq:asymptotics}.
Notice in particular that the value $\mu=0$ (which in this paper requires a special treatment and the additional asymptotic assumption~\eqref{eq:asymmu0}) is immediately ruled out in the vacuum case by using the divergence theorem applied to the vacuum equation $\Delta N=0$ and exploiting the asymptotic assumption on $N$ to obtain
\[
\int_{\d M}\nu(N)\,d\sigma\,=\,\lim_{R\to+\infty}\int_{\{\vert x\vert =R\}}\nu(N)\,d\sigma\,=\,4\pi\mu\,.
\]
As vacuum black holes are necessarily non-degenerate (see e.g. \cite{KW}), and the equipotential photon surfaces are non-degenerate by assumption, we find $\nu(N)\neq0$ on $\d M$ and hence $\mu\neq0$, see also \Cref{sub:degenerateequi}.

To set up the cylindrical ansatz, we first employ the strategy presented in~\cite[Corollary~9.6]{Heusler_Blackholes} and also exploited in \cite{YazLaz} to show how our assumptions of asymptotic flatness and connected boundary help to simplify the problem significantly by relating the electric potential explicitly with the lapse function. As discussed, we will only deal with the case $\Psi_{0}\geq0$, recalling that one can treat the opposite sign by a global change of $\Psi$.
\begin{prop}[Reduction]\label{pro:simplification}
Let $(M,\go,N,\Ele)$ be an asymptotically flat electrostatic electro-vacuum system such that $M$ has a connected boundary $\partial M$ to which $N$ and $\Psi$ smoothly extend to constants $N\vert_{\partial M}=N_0\geq0$, $\Ele\vert_{\partial M}=\Ele_{0}\geq0$. Then if $\Psi_{0}>0$
\begin{align}\label{eq:identity}
N^2&=\Ele^2+1+\left(\frac{N_0^2}{\Ele_0}-\Ele_0-\frac{1}{\Ele_0}\right)\Ele
\end{align}
on $M$ and $0<\Psi<\Psi_{0}$ on $M\setminus\partial M$. If on the other hand, $\Psi_{0}=0$, then $\Psi\equiv0$ on $M$.
\end{prop}

\begin{proof}
Using the electro-vacuum equations~\eqref{eq:pb}, a simple computation shows that, for any $a,b\in\R$, the quantity 
\begin{align*}
X_{ab}\definedas\frac{a}{2} \left( \Ele^2+1-N^2 \right)+b\Ele
\end{align*}
satisfies the elliptic PDE
\begin{align}\label{eq:elliptic_X}
\Delta X_{ab}&=\frac{1}{N}\langle \mathrm{D} X_{ab}\,|\,\mathrm{D} N\rangle_{g_{0}}.
\end{align}
From the asymptotic flatness assumption~\eqref{eq:asymptotics}, we know that $N\to 1$ and $\Ele\to 0$ at infinity, from which it follows that $X_{ab}\to 0$ at infinity for any choice of $a,b$. Assume first that $\Psi_{0}>0$. Choosing $a=2$, $b=\tfrac{N_0^2}{\Ele_0}-\Ele_0-\tfrac{1}{\Ele_0}$, $X_{ab}$ specializes to
\begin{align*}
X\definedas X_{a=2\,\,b=\frac{N_0^2}{\Ele_0}-\Ele_0-\tfrac{1}{\Ele_0}}&=\Ele^2+1-N^2+\left(\frac{N_0^2}{\Ele_0}-\Ele_0-\frac{1}{\Ele_0}\right)\Ele
\end{align*}
which satisfies $X=0$ on $\partial M$. If instead $\Psi_{0}=0$, we choose $a=0$, $b=1$ and obtain
\begin{align*}
Y\definedas X_{a=0\,\,b=1}&=\Psi
\end{align*}
which also satisfies $Y=0$ on $\partial M$. We will now apply the maximum and minimum principle to show that $X\equiv 0$ and $Y\equiv0$ on $M$. We have to be a bit careful because the manifold $M$ is not compact and in case $N_0=0$ the coefficient $\frac{1}{N}$ in~\eqref{eq:elliptic_X} blows up on $\partial M$, whereas in the degenerate horizon case $g_{0}$ does not extend smoothly to $\partial M$. However, this is easy to work around: suppose that $\{X>0\}\neq\emptyset$. Since $X=0$ at the boundary and $X\to 0$ at infinity, then $X$ must have a positive maximum at a point $p_0\in M\setminus\partial M$. Now let $\Omega\subset M\setminus\partial M$ be a neighborhood of $p_{0}$ with smooth boundary, large enough to contain some $q\in\Omega$ with $X(q)<X(p_{0})$ (if $X=X(p_{0})$ on any such neighborhood then $X$ is necessarily constant, a contradiction). We can hence apply the strong maximum principle to $X\vert_{\Omega}$ where in particular $N\neq0$, obtaining a contradiction. The cases $\{X<0\}\neq\emptyset$ and $\{Y>0\}\neq\emptyset$, $\{Y<0\}\neq\emptyset$ can be treated analogously. Similarly, for $\Psi_{0}>0$, using $Y$, one finds $0<\Psi<\Psi_{0}$ on $M\setminus\partial M$.
\end{proof}

From now on, under the assumptions of \Cref{pro:simplification}, we restrict to the case when $\Psi_{0}>0$. In view of \eqref{eq:identity} and of the nature of $N$ and $\Psi$ in the Reissner--Nordstr\"om solutions~\eqref{eq:RN} with positive charge (corresponding to $\Psi_{0}>0$), we aim at finding parameters $\mm\in\mathbb{R}$, $\qq>0$ satisfying
\begin{align}\label{eq:m_over_q}
\frac{N_0^2}{\Ele_0}-\Ele_0-\frac{1}{\Ele_0}&=-\frac{2\mm}{\qq}.
\end{align} 
Note that \eqref{eq:m_over_q} fixes the ratio between $\mm$ and $\qq$, while the actual values of one of $\mm$ and $\qq$ still needs to be defined (with the exception of $\frac{\mm}{\qq}=0$ when $\mm=0$ and $\qq>0$ is free). From \Cref{pro:simplification}, it then follows that
\begin{align}\label{eq: idPsiN}
N^2&=\Ele^2-\frac{2\mm}{\qq}\Ele+1
\end{align}
holds on $M$. In particular, the two potentials $N$, $\Ele$ are functions one of the other and we can equivalently rewrite \eqref{eq:pb} in terms of just $N$ by
\begin{align}\label{eq:pb_N}
\begin{dcases}
N \operatorname{Ric}&=\mathrm{D}^2 N-\frac{2N}{N^2+k}\,d N\otimes d N+\frac{N}{N^2+k}\,\vert\mathrm{D} N\vert^2\,\go, \\
\Delta N&=\frac{N}{N^2+k}\,\vert\mathrm{D} N\vert^2
\end{dcases}
\end{align}
on $M$, where
\begin{align}\label{def:k}
 k \definedas \frac{\mm^2}{\qq^2} -1.
 \end{align}
Notice that, by~\eqref{eq: idPsiN}, we have
\begin{align}\label{eq: Nk_geq_0}
    N^2 + k = \left( \Psi -\frac{\mm}{\qq} \right)^2 \geq 0
\end{align}
on $M$. When $k\leq0$, it is possible for the quantity $N^2+k$ to reach the value zero, in which case the equations~\eqref{eq:pb_N} are ill-posed. Nevertheless, from the relation~\eqref{eq: idPsiN} between $N$ and $\Ele$ and \eqref{eq: Nk_geq_0}, we obtain the identity
\begin{align}\label{eq:relation_DN_DPsi}
N^2\,\vert\mathrm{D} N\vert^2&=(N^2+k)\,\vert\mathrm{D}\Ele\vert^2
\end{align}
which can be used to make sense of the equations in~\eqref{eq:pb_N} even when $N^2+k$ vanishes.

Before moving on, let us note that via \eqref{eq:m_over_q}, the ratio of $\mm$ and $\qq>0$ is fully determined by the boundary values $N_{0}$ and $\Psi_{0}$. Consequently, transferring the definitions of sub-extremal, extremal, and super-extremal Reissner--Nordstr\"om spacetimes in terms of their mass and charge parameters $m$ and $q$ directly to general electrostatic electro-vacuum systems $(M,g_{0},N,\Psi)$ with constant lapse function $N_{0}$ and constant electric potential $\Psi_{0}$ on the connected boundary $\partial M$, we say that the system $(M,g_{0},N,\Psi)$ and / or its boundary $(\partial M,N_{0},\Psi_{0})$ is \emph{sub-extremal} if $\mm>\qq$, \emph{extremal} if $\mm=\qq$, and \emph{super-extremal} if $\mm<\qq$. Note that the comparison between $\mm$ and $\qq$ does not depend on how we will later choose $\mm$ and $\qq$ but just on their ratio fixed by  \eqref{eq:m_over_q}. A direct computation then shows that sub-extremality corresponds to $N_{0}^{2}<(1-\Psi_{0})^{2}$, extremality corresponds to $N_{0}^{2}=(1-\Psi_{0})^{2}$, and super-extremality corresponds to $N_{0}^{2}>(1-\Psi_{0})^{2}$. Recalling that the case $\Psi_{0}<0$ can be addressed by a sign change on $\Psi$, these conditions exactly coincide with those in \Cref{thm:photon_surface}.

In particular, if $\partial M$ is a horizon, then $(M,g_{0},N,\Psi)$ is sub-extremal if and only if $\Psi_{0}\neq1$, extremal if $\Psi_{0}=1$, and never super-extremal, as to be expected. From \eqref{eq:relation_DN_DPsi}, we see that in the limit to a non-degenerate horizon $\partial M$, using $2\vert{\rm D}\Psi\vert^{2} =N^{2}\operatorname{R}$ from \Cref{sec:setup} and the fact that the horizon is totally geodesic as well as the twice contracted Gauss equation, we get $0<2\vert{\rm D}N\vert^{2}=k\operatorname{R}^{\partial M}$ on $\partial M$, where $\operatorname{R}^{\partial M}$ denotes the scalar curvature of $\partial M$. This gives $k>0$ or in other words sub-extremality of any asymptotically flat electrostatic, electro-vacuum system $(M,g_{0},N,\Psi)$ bounded by a non-degenerate horizon. It also re-establishes the well-known fact that non-degenerate horizons must be round spheres in our context. 
Concerning the case of a degenerate horizon, it will be shown in \Cref{sub:degenerate} that $k=0$, namely that the system $(M,g_{0},N,\Psi)$ is extremal.

On the other hand, if $\partial M$ is a canonical time-slice of an equipotential photon surface with $\Psi_{0}>0$, the Hopf lemma\footnote{similarly modified as the maximum principle in the proof of \Cref{pro:simplification}.\label{foot:Hopf}} applied to the last equation in \eqref{eq:pb} and the assumptions $\Psi_{0}>0$ and $\Psi\to0$ at infinity by \eqref{eq:asymptotics} give that $\vert{\rm D}\Psi\vert>0$ on $\partial M$. Thus, \eqref{eq:relation_DN_DPsi} tells us that this time-slice is non-degenerate, i.e., $dN\neq0$ on $\partial M$, unless $N_{0}^{2}+k=0$ or equivalently $\Psi_{0}=\frac{\mm}{\qq}>0$ by \eqref{eq: Nk_geq_0}.

\subsection{Consequences of the simplification for degenerate equipotential photon surfaces}\label{sub:degenerateequi}
We continue to work under the assumptions of \Cref{pro:simplification}, with $\Psi_{0}>0$, and assume that $N_{0}>0$ so that we investigate an equipotential photon surface inner boundary. Let us first note that photon spheres are always non-degenerate in this context because they have positive mean curvature $\HHH_{\nu}$ by \cite[Lemma 2.6]{CedGalElec} (see also \cite[Theorem 3.1]{GalMiao}) and $c=1$ or in other words satisfy $N\HHH_{\nu}=2\nu(N)$ (see \cite[Proposition 5.3]{CedJaVi} and note that the non-degeneracy condition is in fact not used in the proof of the photon sphere case). 

Next, let us point out that all photon surfaces in Schwarzschild spacetimes (excluding the Minkowski spacetime) are non-degenerate as $dN\neq0$ holds everywhere in the spacetime, hence it is reasonable to assume non-degeneracy in vacuum equipotential photon surface uniqueness theorems. The same holds for sub-extremal and extremal as well as negative mass Reissner--Nordstr\"om spacetimes and accordingly to electro-vacuum equipotential photon surface uniqueness theorems pertaining only to those cases (and in particular to those only pertaining to the sub-extremal case). However, the situation changes when one considers super-extremal Reissner--Nordstr\"om spacetimes with positive mass: here, $dN=0$ holds precisely on $\{r=\frac{q^{2}}{m}\}$, so all photon surfaces passing through this radius are necessarily degenerate. Hence we do not get around also treating degenerate equipotential photon surfaces in this paper if we want to include all equipotential photon surfaces (or at least all those staying in $\{r\geq m\}$, see \Cref{rem:condition}, noting that $\frac{q^{2}}{m}>m$ when $q>m>0$).

In view of our strategy of proof, it is relevant to observe that unless $N_{0}^{2}+k=0$ holds on $\partial M$, we can appeal to the Hopf lemma\textsuperscript{\ref{foot:Hopf}} applied to the second equation in \eqref{eq:pb_N} to see that $\nu(N)\neq0$: in case $N^{2}+k>0$ in $M$, this follows from the asymptotic assumption $N\to1$ at infinity (as $N_{0}=1$ is ruled out by the maximum principle\textsuperscript{\ref{foot:max}} applied to the same equation, noting that $N\equiv1$ contradicts $\Psi_{0}>0$ and $\Psi\to0$ at infinity by \eqref{eq: idPsiN}). On the other hand, if $\{N^{2}+k=0\}\cap(M\setminus\partial M)\neq\emptyset$, the same argument can be given using some regular value $N_{1}$ of $N$ with $N_{0}>N_{1}>\sqrt{-k}$ (which exists by Sard's lemma). Hence time-slices of equipotential photon surfaces are necessarily \emph{non-degenerate} (meaning $dN\neq0$ on the time-slice) unless they satisfy $N_{0}=\sqrt{-k}$ in which case they are necessarily \emph{degenerate} by \eqref{eq:relation_DN_DPsi}, and indeed satisfy $dN\equiv0$. Before we proceed, let us agree that an (equipotential) photon surface is \emph{locally non-degenerate near a given time-slice $\{t=t_{0}\}$} if $dN\neq0$ on $\{t=t_{0}\}$ and hence in an open neighborhood of time-slices by continuity. 

Next, note that \Cref{thm: ps_properties} applies to non-degenerate time-slices (or locally non-degenerate equipotential photon surfaces) as its proof (see \cite[Proposition 5.5]{CedJaVi}) is completely local. It also readily applies to photon spheres in our context as these are non-degenerate as discussed above. If an equipotential photon surface has a single degenerate slice (with all nearby slices non-degenerate), then \Cref{thm: ps_properties} continues to apply by continuity. On the other hand, we know from~\cite{Tod} that $N$ is a real analytic function on $M$ and thus $dN$ cannot vanish on an open subset of $M$ (or of our spacetime $\R\times M$), unless $N\equiv1$ which we have excluded above. Thus a degenerate slice of an equipotential photon surface must have a neighborhood in which it is the only degenerate slice (as photon spheres are non-degenerate by the above). This ensures that we can appeal to  \Cref{thm: ps_properties} also in the degenerate case.

\subsection{Asymptotic consequences of the simplification}\label{sub:asymptoticsimplification}
We continue to work under the assumptions of \Cref{pro:simplification}, with $\Psi_{0}>0$. The goal of this section is to establish asymptotic decay on relevant quantities derived from $N$ and $g_{0}$ and to bootstrap the assumed decay on $\Psi$. Recalling our asymptotic assumptions from~\eqref{eq:asymptotics}, we compute the following asymptotic expansions for the derivatives of the lapse function $N$
\begin{align}\label{eq:asydN}
\frac{\d N}{\d x^i}&=\frac{\ADM x_{i}}{|x|^3}+o_1\left(|x|^{-2}\right),\\
\frac{\d^2 N}{\d x^i\d x^j}&=\frac{\ADM}{|x|^3}\left(\delta_{ij}-\frac{3x_{i} x_{j}}{|x|^2}\right)+o\left(|x|^{-3}\right)
\end{align}
as $\vert x\vert\to\infty$. Moreover, the Christoffel symbols $\Gamma_{ij}^{k}$ of $\go$ with respect to the asymptotic coordinates $(x^{i})$ satisfy $\Gamma_{ij}^k=o(|x|^{-1})$ and the area element $d\sigma$ of $g_{0}$ on coordinate spheres satisfies 
\begin{align}\label{eq:areaasy}
d\sigma&=\left(\vert x\vert^{2}+o\left(\vert x\vert^{2}\right)\right)d\sigma_{g_{\mathbb{S}^{2}}}
\end{align}
as $\vert x\vert\to\infty$. It follows immediately that
\begin{align}\label{eq:DNlength}
|\mathrm{D} N| &= \frac{\vert\ADM\vert}{|x|^2} + o_1(|x|^{-2}),\\\label{eq:DDNlength}
\DD_{ij}N&=\frac{\ADM}{|x|^3}\left(\delta_{ij}-\frac{3x_{i} x_{j}}{|x|^2}\right)+o\left(|x|^{-3}\right)
\end{align}
as $\vert x\vert\to\infty$. Next, let $\nu$ denote the unit normal to a coordinate sphere pointing towards infinity. By standard computations, it follows that
\begin{align}\label{eq:nuasy}
\begin{split}
\nu&=\frac{\partial}{\partial\vert x\vert}+\alpha^{i}\partial_{x^{i}},\\
\alpha^{i}&=o(1)
\end{split}
\end{align}
for $i=1,2,3$ as $\vert x\vert\to\infty$. Combining \eqref{eq:asymptotics} with \eqref{eq: idPsiN}, we find
\begin{align}\label{eq:Psidecaym}
\Psi(\Psi-\frac{2\mm}{\qq})&=-\frac{2\mu}{\vert x\vert}+o_{2}(\vert x\vert^{-1})
\end{align}
as $\vert x\vert\to\infty$. If $\mm\neq0$, $\Psi=o(1)$ then leads to $\Psi=\frac{\qq\mu}{\mm\vert x\vert}+o(\vert x\vert^{-1})$ and thus, plugging this back into \eqref{eq:Psidecaym}, we find
\begin{align}\label{eq:Psiasy}
\Psi&=\frac{\qq\mu}{\mm\vert x\vert}+o_{2}(\vert x\vert^{-1})
\end{align}
as $\vert x\vert\to\infty$. Next, still assuming $\mm\neq0$, by the last equation in \eqref{eq:pb}, \eqref{eq:areaasy}, \eqref{eq:Psiasy}, \eqref{eq:nuasy}, \eqref{eq:asymptotics}, and the divergence theorem, we find that
\begin{align}\label{eq:divpsi}
\int_{\Sigma}\frac{\nu(\Psi)}{N}\,d\sigma&=\lim_{R\to\infty}\int_{\mathbb{S}^{2}_{R}(0)}\frac{\nu(\Psi)}{N}\,d\sigma=-\frac{4\pi\qq\mu}{\mm}
\end{align}
for any closed, orientable surface $\Sigma$ homologous to $\partial M$, where $\nu$ now also denotes the $g_{0}$-unit normal to $\Sigma$ pointing towards the asymptotic end. On the other hand, applying the Hopf lemma\textsuperscript{\ref{foot:Hopf}} to the last equation in \eqref{eq:pb} and recalling $\Psi_{0}>0$ and $\Psi\to0$ at infinity, we find that $\nu(\Psi)=-\vert{\rm D}\Psi\vert<0$ if $\Sigma$ is sufficiently close to $\partial M$ (say $\Sigma=\partial M$ in the non-degenerate horizon and in the equipotential photon surface case and $\Sigma=\{z=\varepsilon\}$ for sufficiently small $\varepsilon>0$ in the degenerate horizon case). Hence by $\qq>0$, we have in particular asserted that $\frac{\mu}{\mm}>0$ when $\mm\neq0$. Alternatively, one could argue via \eqref{eq:relation_DN_DPsi} and the Hopf lemma to get information about the decay of $\nu(\Psi)$, without first computing the decay of $\Psi$, as will be done below.

Let us now consider the case $\mm=0$ and suppose towards a contradiction that $\mu\neq0$. Then \eqref{eq:relation_DN_DPsi} and $\nu(N)^{2}=\vert{\rm D}N\vert^{2}$ combined with \eqref{eq:asymptotics}, \eqref{eq:DNlength}, and \eqref{def:k} giving $k=-1$ lead to
\begin{align}
\nu(\Psi)^{2}&=-\frac{\mu}{2\vert x\vert^{3}}+o(\vert x\vert^{-3})
\end{align}
as $\vert x\vert\to\infty$ and thus in particular $\mu<0$. Arguing again by the Hopf lemma\textsuperscript{\ref{foot:Hopf}}, one finds that $\nu(\Psi)<0$ asymptotically and hence 
\begin{align}
\nu(\Psi)&=-\sqrt{-\frac{\mu}{2\vert x\vert^{3}}}+o(\vert x\vert^{-\frac{3}{2}})
\end{align}
as $\vert x\vert\to\infty$. Applying the divergence theorem as in \eqref{eq:divpsi}, this leads to
\begin{align}
\int_{\Sigma}\frac{\nu(\Psi)}{N}\,d\sigma&=\int_{\mathbb{S}^{2}_{R}(0)}\frac{\nu(\Psi)}{N}\,d\sigma=-4\pi\sqrt{-\frac{\mu R}{2}}+o(\sqrt{R})
\end{align}
as $R\to\infty$ for any closed, orientable surface $\Sigma$ homologous to $\partial M$. Fixing $\Sigma$ as above, this leads to the desired contradiction. We have hence proved the following lemma.
\begin{lemma}[Signs of $\mm$ and $\mu$]\label{lem:sign}
Under the assumptions of \Cref{pro:simplification} with $\Psi_{0}>0$, either $\mu=\mm=0$ or $\mu\mm>0$.
\end{lemma}

Using the additional decay assumption \eqref{eq:asymmu0} we are making in case $\mu=0$ combined with \eqref{eq:nuasy}, we find the improved decay
\begin{align}\label{eq:nupsi}
\nu(\Psi)&=-\frac{\kappa}{\vert x\vert^{2}}+o(\vert x\vert^{2})
\end{align}
as $\vert x\vert\to\infty$. Hence the above divergence theorem argument gives
\begin{align}\label{eq:foundkappa}
\int_{\Sigma}\frac{\nu(\Psi)}{N}\,d\sigma&=-4\pi\kappa
\end{align}
for any closed, orientable surface $\Sigma$ homologous to $\partial M$. Again, by the Hopf lemma\textsuperscript{\ref{foot:Hopf}}, this proves the following lemma.
\begin{lemma}[Sign of $\kappa$]\label{lem:signkappa}
Under the assumptions of \Cref{pro:simplification} with $\Psi_{0}>0$, assuming $\mu=\mm=0$ and \eqref{eq:asymmu0}, we have $\kappa>0$.
\end{lemma}

\section{The cylindrical ansatz}\label{sec:cyl}
In this section, we will set up the cylindrical ansatz we will use in all the different parts of the proof of \Cref{thm:BH,thm:photon_surface}. We will continue to work under the assumptions of \Cref{pro:simplification}, with $\Psi_{0}>0$. The analysis of this section is performed for any pair of parameters $\mm\in\R$ and $\qq>0$ whose ratio satisfies relation~\eqref{eq:m_over_q}. We will later show how to fix the values of $\mm$, $\qq$ appropriately.

\subsection{Setting up the cylindrical ansatz}\label{sub:conformal_change}
In the spirit of~\cite{BorMazII} and inspired by the explicit formula for $N$ as a function of the radial coordinate for the Reissner--Nordstr\"om solution~\eqref{eq:RN}, we define the {\em pseudo-radial function $\pra\colon M\to\R$} as a function satisfying the following relation
\begin{align}\label{eq:relation_N_psi}
N^2&=1-\frac{2\mm}{\pra}+\frac{\qq^2}{\pra^2},
\end{align}
or equivalently
\begin{align}
N^{2}+k&=\left(\frac{\mm}{\qq}-\frac{\qq}{\rho}\right)^{2}.
\end{align}
One can solve equation~\eqref{eq:relation_N_psi} explicitly to find two choices for $\pra$, namely
\begin{align}\label{eq:rho+}
\pra_+&=\frac{\qq^2}{\mm + \qq\sqrt{N^2+k}},\\\label{eq:rho-}
\pra_-&=\frac{\qq^2}{\mm - \qq\sqrt{N^2+k}}.
\end{align}

Before analyzing the properties of these pseudo-radial functions, let us explain when we will use which one. The choice that we will use for most of this work is $\pra_-$. In fact it will be the one we will consider when dealing with horizons (see \Cref{sub:horizon_consequences} for the non-degenerate and \Cref{sub:degenerate} for the degenerate case) and with both the sub-extremal and extremal cases for equipotential photon surfaces (see \Cref{sub:proof_photon_standard}). The reason for this is that we expect the pseudo-radial function to mimic the behaviour of the radial coordinate $r$ in the Reissner--Nordstr\"om solution, so in particular we want the pseudo-radial function to diverge at infinity, see \Cref{lem:rho+,lem:rho-}.

On the other hand, when dealing with the super-extremal case for equipotential photon surfaces, there will be two important occasions in which the choice $\rho_+$ will actually be the natural one. The first case is when $\mm\leq0$: in this case the behaviour at infinity of $\rho_-$ and $\rho_+$ is reversed, and $\rho_+$ is this time the one diverging as expected. Another, more subtle, case is when the \emph{critical set} $\C=\{ N^2 + k = 0\}$ is non-empty: in this case, we will show that $\C$ divides the manifold into two pieces, an asymptotically flat one (where the correct choice for the pseudo-radial function is again $\rho_-$) and a compact one (where the right choice is now $\rho_+$). These special cases will be discussed in \Cref{sub:neg_mass} and \Cref{sub:super-extremal_Cnonempty}, respectively. In the super-extremal case with $\C=\emptyset$ and $\mm>0$, we still use $\rho_{-}$ (see \Cref{subsec:superextremalnocrit}).

Let us now summarize the properties of the two pseudo-radial functions $\rho_{\pm}$ for later convenience. As already visible in the reduced equations~\eqref{eq:pb_N}, the critical set $\C$ plays a special role also for the smoothness of the pseudo-radial functions. Note that $\C$ can a priori contain $\partial M$: this happens for a degenerate horizon boundary where $N_{0}=0$ and $k=0$ (see \Cref{sec:cyl}) and for a degenerate time-slice of an equipotential photon surface on which $\Psi_{0}=\frac{\mm}{\qq}$. For a non-degenerate horizon, $\C\cap\partial M=\emptyset$ because $k>0$ as discussed at the end of \Cref{sec:cyl}. Also, for a non-degenerate equipotential photon surface boundary $\partial M$, $\C\cap\partial M=\emptyset$ by \eqref{eq:relation_DN_DPsi}.

\begin{lemma}[Properties of $\rho_{+}$]\label{lem:rho+}
The pseudo-radial function $\rho_{+}$ is well-defined and continuous on $M$ and smooth away from $\C$. Moreover, $\rho_{+}>0$ on $M$. If $\mm>0$, one has $\rho_+\to \frac{\qq^2}{2\mm}$ in the asymptotically flat end of $M$. If $\mm\leq0$, one has $\rho_{+}\to\infty$  in the asymptotically flat end of $M$. Last but not least, $\rho_{+}(1+N)-\mm>0$ on $M$ if and only if $\mm\leq0$ or $\qq>\mm>0$.
\end{lemma}
\begin{proof}
To see that $\rho_{+}$ is well-defined and continuous on $M$, we observe that by \eqref{eq: Nk_geq_0} and $\qq>0$, a zero of its denominator can only occur for $\mm\leq0$ and hence when $\Psi=\frac{2\mm}{\qq}\leq0$, in contradiction to \Cref{pro:simplification}. By its definition in \eqref{eq:rho+}, $\rho_{+}$ is smooth away from $\C$. The asymptotic claims are immediate when recalling that $N\to1$ in the asymptotically flat end, taking into account the definition of $k$ in \eqref{def:k}. Positivity of $\rho_{+}$ is immediate from the asymptotic behavior upon noticing that $\rho_{+}\neq0$ and recalling \eqref{def:k}, $\qq>0$. Finally, $\rho_{+}(1+N)-\mm>0$ obviously holds by positivity of $\rho_{+}$ provided that $\mm\leq0$. Assuming $\mm>0$, one computes that $\rho_{+}(1+N)-\mm>0$ is equivalent to $\qq^{2}>\mm^{2}$ which by $\qq>0$ and $\mm>0$ proves the claim.
\end{proof}

\begin{lemma}[Properties of $\rho_{-}$]\label{lem:rho-}
The pseudo-radial function $\rho_{-}$ is well-defined and continuous on $M\setminus\{N=1\}$ and well-defined and continuous on $\{N=1\}$ if and only if $\mm<0$. It is smooth in its domain of definition away from $\C$. Moreover, if $\mm>0$, $\rho_{-}>0$ on $M\cap\{N<1\}$ while $\rho_{-}<0$ on $M\cap\{N>1\}$. On the other hand, $\rho_{-}<0$ on $M$ if $\mm\leq0$. If $\mm>0$, one has $\rho_-\to \infty$ in the asymptotically flat end of $M$. If $\mm\leq0$, one has $\rho_{-}\to\frac{\qq^2}{2\mm}$  in the asymptotically flat end of $M$. Last but not least, $\rho_{-}(1+N)-\mm>0$ on $M$ if $\mm\leq0$, $\rho_{-}(1+N)-\mm<0$ on $M\cap\{N>1\}$ if $\mm>0$, and  $\rho_{-}(1+N)-\mm>0$ on $M\cap\{N<1\}$ if $\mm>0$.
\end{lemma}
\begin{proof}
To see that $\rho_{-}$ is well-defined and continuous away from $\{N=1\}$ and well-defined and continuous on $\{N=1\}$ if and only if $\mm<0$, note that by \eqref{eq: Nk_geq_0} and $\qq>0$, a zero of its denominator will occur if and only if $\mm\geq0$ and $\Psi=\frac{2\mm}{\qq}$ by \Cref{pro:simplification}. The second condition is equivalent to $N=1$ by \eqref{eq: idPsiN}. The smoothness claim readily follows from \eqref{eq:rho-}. A simple computation using \eqref{eq: idPsiN} gives that $\rho_{-}>0$ only holds if $\mm>0$ and $N<1$ but $\rho_{-}<0$ holds if $\mm\leq0$ or if $\mm>0$ and $N>1$. The asymptotic claims follows as in the proof of \Cref{lem:rho+}. Last but not least, $\rho_{-}(1+N)-\mm>0$ obviously holds on $M$ by positivity of $\rho_{-}$ when $\mm\leq0$. Similarly, $\rho_{-}(1+N)-\mm<0$ obviously holds on $M\cap\{N>1\}$ by negativity of $\rho_{-}$ when $\mm>0$. Finally, if $\qq\geq\mm>0$, one estimates 
\begin{align*}
\rho_{-}(1+N)-\mm&\geq\frac{\mm^{2}}{\mm-\qq\sqrt{N^{2}+k}}-\mm=\frac{\qq\sqrt{N^{2}+k}}{\mm-\qq\sqrt{N^{2}+k}}=\frac{\sqrt{N^{2}+k}}{\qq}\rho_{-}>0
\end{align*}
 on $M\cap\{N<1\}$ as $N^{2}+k>0$ where $N<1$ by \eqref{def:k}, and using $N\geq0$. On the other hand, if $\mm>\qq>0$,  one estimates
 \begin{align*}
\rho_{-}(1+N)-\mm&\geq\frac{\qq^{2}-\mm^{2}+\mm\qq\sqrt{k}}{\mm-\qq\sqrt{N^{2}+k}}
\end{align*}
 on $M\cap\{N<1\}$, again using $N\geq0$ and the fact that $k>0$ when $\mm>\qq$ by \eqref{def:k}. As $\mm>\qq$, 
$ \qq^{2}-\mm^{2}+\mm\qq\sqrt{k}>0$ is equivalent to $0<\mm^{2}-\qq^{2}<\mm\qq\sqrt{k}$ and thus, by squaring, equivalent to $\qq^{2}>0$ which is clearly true. This implies the last claim.
\end{proof}

In particular, \Cref{lem:rho+,lem:rho-} establish that we can use the pseudo-radial functions $\rho_{\pm}$ as explained in the above discussion, once we establish that $N\neq1$ on $M\setminus\partial M$ when using $\rho_{-}$ so that $\rho_{-}$ is well-defined. Moreover, \Cref{lem:rho+,lem:rho-} assert that once the pseudo-radial function $\pra_\pm$ has been chosen as explained above (and is well-defined), we can then introduce the {\em pseudo-affine functions} $\ffi_{\pm}\colon M\to\R$, defined by
\begin{align}\label{eq:ffi}
\ffi_\pm&\definedas\log\left[\pra_\pm(1+N)-\mm\right]
\end{align}
and the {\em cylindrical ansatz metric}
\begin{align}\label{eq:g}
g_{\pm}\definedas\frac{\go}{\pra_{\pm}^2}.
\end{align}
These are defined such that if $\pra_\pm$ is chosen to be the radial coordinate $r$ in the Reissner--Nordstr\"om solution of mass $m$ and charge $q>0$ as in the above explanation, one can check that $g_\pm$ is a round cylindrical metric of radius $1$ and $\ffi_\pm$ is an affine function such that $\nabla \ffi_\pm$ is the splitting direction for $g_{\pm}$, where $\nabla$ denotes the covariant derivative with respect to $g_{\pm}$. The cylindrical metric is smooth on $M\setminus\partial M$ and the pseudo-affine functions $\varphi_{\pm}$ are smooth away from $\C$ where well-defined. 

Quite surprisingly, almost all formulas that we obtain through local computations (without using any global insights) give the same expressions independently of whether we work with $\rho_+$ or $\rho_-$, as long as $\varphi_{\pm}$ is well-defined and $\rho_{\pm}>0$ (see \Cref{lem:rho+,lem:rho-}). This prompts us to avoid the subscript whenever possible, namely we simplify the notation as follows.

\begin{notation}\label{not1}
In order to ease the notation, when there is no risk of confusion, we will avoid to write explicitly the subscript $\pm$: we will simply write $\pra$, without explicitly saying which among the pseudo-radial functions~\eqref{eq:rho+} and~\eqref{eq:rho-} we are choosing. Analogously, we will denote by $\ffi$, $g$ the pseudo-affine function and the cylindrical ansatz metric with respect to the chosen pseudo-radial function $\pra$ (whenever $\varphi$ is well-defined and $\rho>0$, see \Cref{lem:rho+,lem:rho-}.Moreover, it will be understood that we can only use $\rho_{-}$ where $N\neq1$ (unless $\mm<0$) and that we stay away from the critical set $\C=\{N^{2}+k=0\}$ and from $\partial M$ unless explicitly saying otherwise.
\end{notation}

Now let $\pra$ be a pseudo-radial function, defined as in~\eqref{eq:rho+} or~\eqref{eq:rho-}, and let $\ffi$, $g$ be defined by~\eqref{eq:ffi} and~\eqref{eq:g}, respectively. We will investigate how the equations in \eqref{eq:pb_N} transform under the conformal change \eqref{eq:g}, or in other words to the new variables $(g,\varphi)$. The computations for a conformal change of this particular form have been discussed with some care in~\cite[Section~3]{BorMazII}. Starting from~\cite[Formula~(3.3)]{BorMazII}, with some calculations we obtain the following relation between $\nana\ffi$ and $\DD N$, where $\nabla$ and $\mathrm{D}$ denote the Levi--Civita connections of $g$ and $g_0$, respectively, obtaining
\begin{align}\label{eq:nana_ffi}
\nana\ffi&=\frac{\pra^2}{\mm\pra-\qq^2}\,\DD N+\frac{N\pra^4(3\mm\pra-4\qq^2)}{(\mm\pra-\qq^2)^3}\,dN\otimes dN-\frac{N\pra^4}{(\mm\pra-\qq^2)^2}\,|\mathrm{D}N|^2\,\go.
\end{align}
Note that $\mm\rho-\qq\neq0$ where $N^{2}+k\neq0$. Taking the trace of~\eqref{eq:nana_ffi} and recalling the second equation in~\eqref{eq:pb_N}, we find out that $\ffi$ is $g$-harmonic,
\begin{align}\label{eq:Deg_ffi}
\Delta_{g}\ffi&=0.
\end{align}
From \eqref{eq:ffi}, one computes
\begin{align}\label{eq:dphidN}
d\varphi&=\mp\frac{\rho_{\pm}}{q\sqrt{N^{2}+k}}\,dN
\end{align}
from which one obtains a very useful relation between $|\nabla\ffi|_{g}$ and $|\mathrm{D} N|$, namely
\begin{align}\label{eq:naffi_DNk}
|\nabla\ffi|_{g}^2&=\frac{\pra^4\,|\mathrm{D}N|^2}{\qq^2(N^2+k)}
 \end{align}
or equivalently 
\begin{align}\label{eq:naffi_DN}
|\nabla\ffi|_{g}^2&=\frac{\pra^6\,|\mathrm{D} N|^2}{(\mm \pra-\qq^2)^2}.
\end{align}
Another important relation is the one between the Hessians of $\pra$ and $\ffi$, namely
\begin{align}\label{eq:hessians}
\nana\pra&=\pra N \,\nana\ffi+(\pra-\mm) \,d\ffi\otimes d\ffi.
\end{align}
In particular, tracing the second equation and recalling $\Delta_{g}\ffi=0$, we get
\begin{align}\label{eq: laplacianRho}
   \Delta_{g}\pra&=(\pra-\mm)\,|\nabla\ffi|_{g}^2 
\end{align}
Finally, recall from~\cite[Theorem~1.159]{Besse} the formula for the conformal transformation of the Ricci tensor
\begin{align}
\operatorname{Ric}&=\Ric_{g}-\frac{1}{\pra}\,\nana\pra+\frac{2}{\pra^2}\,d\pra\otimes d\pra-\frac{1}{\pra}\,\Delta_{g}\pra\,g.
\end{align}
Using the first equation in~\eqref{eq:pb_N} to write down $\operatorname{Ric}_{g_{0}}$ as a function of $\DD N$ and using the above formul\ae~\eqref{eq:hessians} and~\eqref{eq:nana_ffi}, after some calculations we get
\begin{align}\label{eq:Ricg}
\Ric_{g}&=\frac{1}{N}\left(1-\frac{\mm}{\pra}\right)\nana\ffi-d\ffi\otimes d\ffi+|\nabla\ffi|_{g}^2\,g.
\end{align}
The reformulation \eqref{eq:Ricg}, \eqref{eq:Deg_ffi} of \eqref{eq:pb_N} in terms of $\pra$, $g$, and $\ffi$ will play a crucial role in the proofs of \Cref{thm:BH,thm:photon_surface}. Our strategy will be to exploit the geometry of the level sets, studied in \Cref{sub:geometry level sets,sub:Bochner}, and the properties of black holes and equipotential photon surfaces to prove that the quantity $|\nabla^2 \varphi|_g$ vanishes identically. This will lead us to conclude that  the solution is isometric to a (half) round cylinder (see for instance the final part of the proof of~\cite[Proposition~4.2]{Mazz}), which in turn will imply that $g_0$ is the spatial Reissner--Nordstr\"om metric.

\subsection{Geometry of the level sets}\label{sub:geometry level sets}
In this section, we will carefully study the geometry of the level sets of $N$, appealing to \Cref{not1}. To this end, let $\pra$ be a pseudo-radial function, defined as in~\eqref{eq:rho+} or~\eqref{eq:rho-}, and let $\ffi$, $g$ be defined by~\eqref{eq:ffi} and~\eqref{eq:g}, respectively. In this subsection we study the extrinsic curvature of a level set $\Sigma=\{N=s\}$. We start by discussing the case when $s$ is a regular value of the function $N$ and hence also of the function $\varphi$ by \eqref{eq:naffi_DN}. Let ${\rm n}\definedas\frac{\mathrm{D} N}{|\mathrm{D} N|}$ and ${\rm n}_g\definedas\frac{\nabla\ffi}{|\nabla\ffi|_g}$ be our choices of $\go$-unit normal and of $g$-unit normal to $\Sigma$, respectively, and let $\HHH$ and $\Hg$ be the mean curvatures of $\Sigma$ with respect to $\go$ and $g$ respectively. 

Since $\Sigma$ is also a level set of $\ffi$ and since $\varphi$ is $g$-harmonic by \eqref{eq:Deg_ffi}, arguing as for example in~\cite[Section~3.2]{BorMazII}, we deduce that
\begin{align}\label{eq:Hg_nanaffi}
\Hg|\nabla\ffi|_g&=-\nana\ffi({\rm n}_g,{\rm n}_g),\\\label{eq:H_DDN}
\HHH|\mathrm{D} N| &=\Delta N-\DD N({\rm n},{\rm n}).
\end{align}
Next, writing $\nana\ffi$ in terms of $\DD N$ using~\eqref{eq:hessians} and then using the first equation in~\eqref{eq:pb_N} to write $\DD N$ in terms of $\operatorname{Ric}$, we compute
\begin{align}\label{eq:Hg_Ric}
\Hg\,|\nabla\ffi|_g&=-\frac{N\pra^4}{\mm\pra-\qq^2}\operatorname{Ric}({\rm n},{\rm n})-\frac{2N\pra^6}{(\mm\pra-\qq^2)^2}\,|\mathrm{D} N|^2.
\end{align}
Recalling the first equation in~\eqref{eq:pb_N} and using~\eqref{eq:H_DDN}, we can rewrite \eqref{eq:Hg_Ric} as
\begin{align}\label{eq:H_Hg}
\Hg|\nabla\ffi|_g&=\frac{\pra^4}{\mm\pra-\qq^2}\left[\HHH\,|\mathrm{D} N|-\frac{2N\pra^2}{\mm\pra-\qq^2}\,|\mathrm{D} N|^2\right].
\end{align}
This is a useful formula to compare the mean curvatures with respect to the two metrics on a given level set. An alternative useful modification of~\eqref{eq:Hg_Ric} is obtained by means of the Gauss--Codazzi equation
\begin{align}
2\Ric({\rm n},{\rm n})&=\operatorname{R}-\operatorname{R}^\Sigma-|\operatorname{h}|^2+\HHH^2
\end{align}
where $\operatorname{R}^{\Sigma}$ denotes the scalar curvature of $\Sigma$ with respect to $g_{0}$ and $\operatorname{h}$ denotes its second fundamental form with respect to $g_{0}$. Recalling 
\begin{align}
\operatorname{R}&=\frac{2\qq^2\pra^2}{(\mm\pra-\qq^2)^2}\,|\mathrm{D} N|^2
\end{align}
(this can be obtained by tracing the first equation in~\eqref{eq:pb_N} and then applying some algebraic manipulations), we find
\begin{align}\label{eq:Hg}
\Hg\,|\nabla\ffi|_g&=\frac{N\pra^4}{2(\mm\pra-\qq^2)}\left(\operatorname{R}^\Sigma+|\mathring{\operatorname{h}}|^2-\frac{\HHH^2}{2}\right)-\frac{(2\mm\pra-\qq^2)N\pra^6}{(\mm\pra-\qq^2)^3}\,|\mathrm{D} N|^2,
\end{align}
where $\mathring{\operatorname{h}}=\operatorname{h}-\frac{\HHH}{2} g_0$ denotes the traceless part of the second fundamental form $\operatorname{h}$ of $\Sigma$ with respect to $g_{0}$.

The normals ${\rm n}$ and ${\rm n}_{g}$ are related by ${\rm n}_{g}=\mp \pra_{\pm}\,{\rm n}$ as one sees from a direct computation using \eqref{eq:naffi_DNk} and \eqref{eq:dphidN}. From the usual law of conformal transformation of mean curvature, we find $\Hg=\mp\left(\pra_{\pm}\HHH-2\,{\rm n}(\rho_{\pm})\right)$.

Let us end this subsection by discussing the non-regular level sets of $N$. Let ${\rm Crit}(N)$ be the set of the critical points of $N$ and suppose that $\{N=s\}\,\cap\,{\rm Crit}(N)\neq\emptyset$. We recall from~\cite{Tod} that $N$ is a real analytic function on $M$ (with respect to $g_{0}$-harmonic coordinates). In particular, the set of critical values of $N$ is discrete. Furthermore, from the results in~\cite{Lojasiewicz} (see also~\cite[Theorem~6.3.3]{Krantz_Parks}) on the structure of the critical set of a real analytic function, it follows that there exists a smooth surface $S\subseteq{\rm Crit}(N)$ with Hausdorff measure $\mathcal{H}^2({\rm Crit}(N)\setminus S)=0$. As a consequence, the unit normal, the second fundamental form, and the mean curvature are well-defined $\mathcal{H}^2$-almost everywhere on the critical level set $\{N=s\}$ as well, with respect to $g_{0}$. Proceeding now exactly as at the end of~\cite[Section~3.2]{BorMazII}, one computes $\Hg=\mp\pra_{\pm}\HHH$ at the points of $S$. It follows then that \eqref{eq:Hg} holds at the points of $S$, which in particular implies that~\eqref{eq:Hg} holds $\mathcal{H}^2$-almost everywhere on the critical level set $\{N=s\}$ as well.

\subsection{Bochner formula}\label{sub:Bochner}
In this section, we will study some useful consequences of the Bochner formula in the cylindrical picture, appealing again to \Cref{not1}. Again, let $\pra$ be a pseudo-radial function, defined as in~\eqref{eq:rho+} or~\eqref{eq:rho-}, and let $\ffi$, $g$ be defined by~\eqref{eq:ffi} and~\eqref{eq:g}, respectively. The Bochner formula for the $g$-harmonic function $\varphi$ then reads
\begin{align}\notag
\Deg|\nabla\ffi|_g^2&=2\,|\nana\ffi|_g^2+2\Ricg(\nabla\ffi,\nabla\ffi)\\\label{eq:bochner_original}
&=2\,|\nana\ffi|_g^2+\frac{1}{N}\left(1-\frac{\mm}{\pra}\right)\langle \nabla|\nabla\ffi|_g^2\,|\,\nabla\ffi\rangle_g,
\end{align}
where we have used \eqref{eq:Ricg} in the second equality. More generally, for any $p\geq 2$, one has the following formula
\begin{align}\label{eq:bochner_originalp}
\begin{split}
\Deg|\nabla\ffi|_g^p&=p\,|\nabla\ffi|_g^{p-2}\,|\nana\ffi|_g^2+\frac{1}{N}\left(1-\frac{\mm}{\pra}\right)\langle \nabla|\nabla\ffi|_g^p\,|\,\nabla\ffi\rangle_g\\
&\quad+p(p-2)|\nabla\ffi|_g^{p-2}\,\left|\nabla|\nabla\ffi|_g\right|_g^2,
\end{split}
\end{align}
which in particular implies
\begin{align}\label{eq:bochner}
\Deg|\nabla\ffi|_g^p-\frac{1}{N}\left(1-\frac{\mm}{\pra}\right)\langle \nabla|\nabla\ffi|_g^p\,|\,\nabla\ffi\rangle_g\geq0.
\end{align}
This elliptic inequality for $|\nabla\ffi|_g^p$ will play an important role later.
An alternative useful way to rewrite~\eqref{eq:bochner} is the following
\begin{align}
\label{eq:positive_divergence}
{\rm div}_g\left(\frac{1}{\pra N}\nabla|\nabla\ffi|_g^p\right)&=\frac{1}{\pra N}\,\left[\Deg|\nabla\ffi|_g^p-\frac{1}{N}\left(1-\frac{\mm}{\pra}\right)\langle \nabla|\nabla\ffi|_g^p\,|\,\nabla\ffi\rangle_g\right]\geq0.
\end{align}
In other words, we have found a family of vector fields 
\begin{align}\label{eq:X}
X_{p}&\definedas \frac{1}{\pra N}\nabla|\nabla\ffi|_g^p
\end{align}
with non-negative $g$-divergence as desired, for each $p\geq2$.

\section{Black Hole Uniqueness in Electro-Vacuum}\label{sec:BH}
This section is dedicated to the proof of \Cref{thm:BH}, that is, we want to characterize sub-extremal and extremal Reissner--Nordstr\"om solutions as the only asymptotically flat, electrostatic electro-vacuum spacetimes with a connected horizon. To this end, we will exploit the cylindrical ansatz formalism discussed in \Cref{sec:cyl} and will continue to assume $\Psi_{0}>0$. As discussed above, the parameters $\mm$ and $\qq$ shall be chosen so that $\frac{\mm}{\qq}$ satisfies~\eqref{eq:m_over_q}. Since $N_0=0$ in the horizon case and $\Ele_0+\frac{1}{\Ele_0}\geq 2$ whatever the choice of $\Ele_0>0$, we immediately deduce that $\mm\geq\qq >0$. This is expected, since we want $\mm$ and $\qq$ to represent mass and charge of a Reissner--Nordstr\"om solution, and hence a horizon should only be present in the sub-extremal and extremal case\footnote{recall that we have excluded $\qq=0$ by assumption and $\qq<0$ by alluding to a sign change for $\Psi$.}. 

Notice in particular that $k\geq 0$ follows from $\mm\geq\qq>0$ by \eqref{def:k} and that the critical set $\C$ does not intersect the \emph{interior of the manifold} $M\setminus\partial M$, so that $\C=\emptyset$ holds in the non-degenerate and $\C=\partial M$ holds in the degenerate case (see the discussion after \eqref{eq:g}). This tells us that the equations in~\eqref{eq:pb_N} are well-defined everywhere in the interior. Moreover, by the maximum principle\footnote{adapted to this context as in the proof of \Cref{pro:simplification}\label{foot:max}} applied to the second equation in \eqref{eq:pb}, exploiting $N_{0}=0$ and the asymptotic assumption $N\to1$ at infinity, ensures that $0<N<1$ in $M\setminus\partial M$ and in particular $N\neq1$ in $M$\label{p:01}. Hence we can apply the conformal change discussed in \Cref{sec:cyl} on $M\setminus\partial M$ with the pseudo-radial function $\rho_{-}$ from~\eqref{eq:rho-}, see \Cref{lem:rho-}. We then choose the pseudo-affine function $\varphi_{-}$ and cylindrical ansatz metric $g_{-}$ accordingly via \eqref{eq:ffi}, \eqref{eq:g} and note that they are well-defined and smooth on $M\setminus\partial M$ by \Cref{lem:rho-}. In the non-degenerate horizon case, all these smoothness claims extend to $\partial M$.

\subsection{Asymptotic analysis of the conformal metric}\label{sub:asymptotics}
We start out by discussing the asymptotic behaviour of our solution since it will play an important role in what follows.
More precisely, we show that the gradient and Hessian of $\ffi$ have bounded norm at infinity and that the metric $g$ is {\em asymptotically cylindrical}, namely the level sets of $\ffi$ converge to round spheres of fixed radius with respect to the metric $g$. The arguments are similar to the ones presented in~\cite[Lemma~3.1]{Mazz}, but since we are using a weaker notion of asymptotic flatness, we prefer to present all the details.

First, from $\mm>0$, we learn that $\mu>0$ by \Cref{lem:sign}. Hence the $g_{0}$-unit normal ${\rm n}=\frac{{\rm D}N}{\vert{\rm D}N\vert}$ to the level sets of $N$ satisfies
\begin{align}\label{eq:DN_asymptotics}
\begin{split}
{\rm n}&=\left(\frac{x^i}{|x|}+\sigma^{i}\right)\frac{\d}{\d x^i}=\frac{\d}{\d |x|}+\sigma^{i}\frac{\d}{\d x^i},\\
\sigma^{i}&=o_{1}(1)
\end{split}
\end{align}
as $\vert x\vert\to\infty$, where we have used \eqref{eq:asydN} and \eqref{eq:DNlength}.
Exploiting that $N_{0}<1$, \eqref{eq:DNlength}, and $\mu>0$ , there exists $0<t_0<1$ such that $|{\rm D} N|\neq 0$ in $\{N\geq t_0\}$, with $\{N\geq t_{0}\}$ contained in the asymptotic end $M\setminus K$ of $M$, $\{N\geq t_{0}\}\subseteq M\setminus K$. It follows from Morse Theory~\cite{Milnor} that $\{N\geq t_0\}$ is diffeomorphic to $\Sigma\times [t_0,1)$, with the level set $\{N=t\}$ corresponding to $\Sigma\times\{t\}$ for all $t\geq t_0$. On the other hand, asymptotic flatness implies that the topology of the asymptotic end $M\setminus K$  must be that of $\R^3\setminus B$ for some closed ball $B$. Since $\R^3\setminus B$ retracts to $\mathbb{S}^2$, it follows that $\Sigma$ does as well. In particular, the fundamental group of $\Sigma$ is the same as the one of $\mathbb{S}^2$. From the classification of surfaces, it follows that $\Sigma$ is diffeomorphic to a sphere. In other words, $\{N\geq t_0\}$ is diffeomorphic to $\mathbb{S}^2\times[t_0,1)$. By means of this diffeomorphism, the function $N$ is just the projection onto the second factor and can be used as a coordinate. It is well known (see e.g. \cite{Israel}) that, choosing coordinates $(\theta^1,\theta^2)$ on the first factor $\mathbb{S}^2$, the metric $g_{0}$ can be rewritten as
\begin{align}\label{eq:Morse_go}
\go&=\frac{d N\otimes d N}{|{\rm D}N|^2}+\bar g_{ab}(N,\theta^1,\theta^2)\,d\theta^a\otimes d\theta^b,
\end{align}
with $a,b=1,2$. Notice that, for a fixed $t$, the coefficients of the metric induced on the necessarily regular level set $\{N=t\}$ are given by $\bar g_{ab}(t,\theta^1,\theta^2)$. 

Let us now study the asymptotic behaviour of the metric $\bar g$ as $t\to1$. Since the asymptotic end of $M$ is diffeomorphic to $\R^3\setminus B\supset\{N\geq t_{0}\}$, we can also consider standard spherical coordinates $(|x|,\vartheta^1,\vartheta^2)$. 
Recalling the asymptotic behaviour~\eqref{eq:asymptotics} of $\go$, we can write
\begin{align}\label{eq:zeta}
\begin{split}
\go&= (1+\omega)\,d|x|\otimes d|x|+|x|\zeta_{a}\left(d|x|\otimes d\vartheta^a+d\vartheta^a\otimes d|x|\right)\\
&\quad+|x|^2\left[ (g_{\mathbb{S}^2})_{ab}+\xi_{ab}\right]d\vartheta^a\otimes d\vartheta^b,\\
\omega, \xi_{ab}, \zeta_a&=o(1)
\end{split}
\end{align}
as $\vert x\vert\to\infty$. Exploiting the fact that $\frac{\d}{\d\theta^a}$ is $g_{0}$-orthogonal to ${\rm D} N$, we can rewrite~\eqref{eq:DN_asymptotics} in spherical coordinates, obtaining
\begin{align}\label{eq:ndecay}
\begin{split}
{\rm n}&=(1+\lambda)\frac{\d}{\d |x|}+ \frac{\lambda^a}{|x|}\frac{\d}{\d \vartheta^a},\\
\lambda,\lambda^a&=o(1)
\end{split}
\end{align}
as $\vert x\vert\to\infty$.
Since the normal ${\rm n}$ is converging to $\d/\d\vert x\vert$ at infinity, we expect the level sets of $N$ to become round as we approach the infinity. For the seek of completeness, we now show how to prove this in detail.
For $a=1,2$ and using \eqref{eq:zeta}, \eqref{eq:ndecay}, we 
compute 
\begin{align*}
0&=\go\left(\frac{\d}{\d\theta^a},{\rm n}\right)\\
&=\left[(1+\omega)(1+\lambda)+\zeta_c\lambda^c\right]\frac{\d |x|}{\d\theta^a}+|x|\left[\lambda^b\left[ (g_{\mathbb{S}^2})_{bc}+\xi_{bc}\right]+(1+\lambda)\zeta_c\right]\frac{\d\vartheta^c}{\d\theta^a}
\end{align*}
for $a=1,2$ as $\vert x\vert\to\infty$. It follows then that $\frac{\partial |x|}{\partial\theta^a}=\tau_b\frac{\partial\vartheta^b}{\partial\theta^a}$, where $\tau_b=o(|x|)$ as $\vert x\vert\to\infty$, $b=1,2$.
With this insight, recalling~\eqref{eq:zeta} we are finally able to compute the asymptotic behaviour of $\bar g_{ab}$, obtaining
\begin{align}\notag
\bar g_{ab}&=\go\left(\frac{\d}{\d\theta^a},\frac{\d}{\d\theta^b}\right)\\\notag
&=(1+\omega)\frac{\d |x|}{\d\theta^a}\frac{\d |x|}{\d\theta^b}+|x|\zeta_{c}\left(\frac{\d |x|}{\d\theta^a}\frac{\d \vartheta^c}{\d\theta^b}+\frac{\d \vartheta^c}{\d\theta^a}\frac{\d |x|}{\d\theta^b}\right)+|x|^2\left[ (g_{\mathbb{S}^2})_{cd}+\xi_{cd}\right]\frac{\d\vartheta^c}{\d\theta^a}\frac{\d\vartheta^d}{\d\theta^b}\\\notag
&=\vert x\vert^2\left[ (g_{\mathbb{S}^2})_{cd}+\tau_{cd}\right]\frac{\d\vartheta^c}{\d\theta^a}\frac{\d\vartheta^d}{\d\theta^b}
\\\label{eq:gbar}
&=\vert x\vert^2\left[ (g_{\mathbb{S}^2})_{cd}+\tau_{cd}\right]d\vartheta^c\otimes d\vartheta^d\left(\frac{\d}{\d\theta^a},\frac{\d}{\d\theta^b}\right)
\end{align}
for $a,b=1,2$ as $\vert x\vert\to\infty$, where $\tau_{cd}=o(1)$ for $c,d=1,2$. In other words, the metric $\bar g$ induced on a level set $\{N=t\}$ can be written as $\bar g=|x|^2\left[ (g_{\mathbb{S}^2})_{cd}+\tau_{cd}\right]d\tilde\vartheta^c\otimes d\tilde\vartheta^d$ where $\tilde\vartheta^a$ is the restriction of $\vartheta^a$ to $\{N=t\}$. 
We have thus shown that the metric $\bar g$ induced on the level sets of $N$ converges to $g_{\mathbb{S}^2}$ at infinity, as expected.

Next, using $\mu>0$, $\mm>0$ as well as \eqref{eq:rho-}, \eqref{eq:ffi}, we easily see that
\begin{align}\label{eq:rhoasy}
\pra &= \frac{\mm |x|}{\ADM} + o_2(|x|),\\\label{eq:phiasy}
\ffi&=\log \left(\frac{2\mm |x|}{\ADM}\right)+o_2(1)
\end{align}
 as $\vert x\vert\to\infty$. Substituting \eqref{eq:rhoasy}, \eqref{eq:phiasy} into \eqref{eq:naffi_DNk} and \eqref{eq:nana_ffi}, we deduce
\begin{align}\label{eq:limit_naffi}
|\nabla\ffi|_g&= \frac{\mm}{\ADM} + o_1(1),\\ 
\nana_{ij}\ffi&=o(|x|^{-2})
\end{align}
as $\vert x\vert\to\infty$. Furthermore, by \eqref{eq:g} and \eqref{eq:rhoasy}, we get
\begin{align}\label{eq:limit_nanaffi}
|\nana\ffi|_g&=o(1)
\end{align}
as $\vert x\vert\to\infty$. Let us now move our attention to the cylindrical ansatz metric $g=\frac{\go}{\pra^2}$. This time we use $\ffi$ as a coordinate on our end instead of $N$.
Dividing formula~\eqref{eq:Morse_go} by $\pra^2$ and using \eqref{eq:dphidN} and \eqref{eq:naffi_DNk} for $\rho_{-}$, we obtain
\begin{align}\label{eq:gphi}
g&=\frac{d\ffi\otimes d\ffi}{|\nabla\ffi|_g^2}+ \frac{\bar g_{ab}}{\pra^2}\,d\theta^a\otimes d\theta^b.
\end{align}
From \eqref{eq:rhoasy} and \eqref{eq:gbar}, we deduce that the metric induced by $g$ on the level sets of $\ffi$ is
\begin{align}
\frac{\bar g_{ab}}{\rho^2}&=\frac{\mu^{2}}{\mm^{2}}\, g_{\mathbb{S}^2}\left(\frac{\d}{\d\theta^a},\frac{\d}{\d\theta^b}\right)+o(1)
\end{align}
as $\vert x\vert\to\infty$. In other words, the metric induced by $g$ on the level sets of $N$ (and thus of $\varphi$) converges to  $\frac{\ADM^2}{\mm^2} g_{\mathbb{S}^2}$. Denoting by $d\sigma_{g}$ the area element induced by $g$, it follows that 
\begin{align}
\label{eq:limit_area}
\lim_{s\to\infty}\int_{\{\ffi=s\}}d\sigma_g &=\frac{\ADM^2}{\mm^2} \,|\mathbb{S}^2|.
\end{align}

\subsection{First consequences of the Bochner formula}\label{sub: standard_Conformal}
In this section, we will discuss a couple of consequences of the Bochner formula introduced in \Cref{sub:Bochner}. In particular, we will exploit inequalities~\eqref{eq:bochner} and~\eqref{eq:positive_divergence} to obtain some crucial information on the function $\ffi$. Further consequences of the Bochner formula will be discussed separately for non-degenerate and degenerate black holes in \Cref{sub:horizon_consequences,sub:degenerate}, respectively, and for equipotential photon surfaces in \Cref{sec:photon_surface}. For our first result, we start from the vector fields $X_{p}$ with non-positive divergence found in~\eqref{eq:X} and we prove an integral inequality on the mean curvature $\Hg$ of the level sets of $\ffi$.

\begin{prop}\label{pro:integral_Hg}
Let $(M,\go,N,\Ele)$ be an asymptotically flat electrostatic electro-vacuum system with mass parameter $\mu$ such that $M$ has a connected horizon boundary $\partial M$. Let $\pra=\rho_{-}$, $\ffi=\varphi_{-}$, $g=g_{-}$ be defined as in~\eqref{eq:rho-}, \eqref{eq:ffi}, and \eqref{eq:g}, with respect to parameters $\mm,\qq$ satisfying~\eqref{eq:m_over_q}. Then, for any $p\geq 3$ and any $s<\infty$ such that $\{\ffi=s\}\subset M\setminus\partial M$, one has
\begin{align}\label{eq:positive_int}
\int_{\{\ffi=s\}}|\nabla\ffi|_g^{p-1}\Hg\,d\sigma_g\geq 0,
\end{align}
where $\Hg$ is the mean curvature of the level set $\{\ffi=s\}$ with respect to the metric $g$. Furthermore, if equality holds in~\eqref{eq:positive_int} for some $s$ and $p$, then $(\{\varphi\geq s\},\go,N,\Ele)$ is isometric to the Reissner--Nordstr\"om solution~\eqref{eq:RN} with mass $\mm$ and charge $\qq$. In particular, if equality holds in~\eqref{eq:positive_int} in the limit $s\to \varphi_{0}$ where $\partial M=\{\varphi=\varphi_{0}\}$ then $(M\setminus\partial M,\go,N,\Ele)$ is isometric to the Reissner--Nordstr\"om solution~\eqref{eq:RN} with mass $\mu$ and charge $q>0$ such that $\frac{\mu}{q}=\frac{\mm}{\qq}$, including the boundary in the non-degenerate horizon case.
\end{prop}

\begin{rmk}
It is worth mentioning that one can refine the above proposition and show that the same result remains true for any $p>3/2$, at the cost of making the proof more complicated (a more careful analysis of the critical points is needed, in the spirit of~\cite[Section~3.2]{MazzAgo}). However, the above proposition is enough for our purposes.
\end{rmk}

\begin{proof}
From \eqref{eq:phiasy} or \Cref{lem:rho-}, we know that $\{\ffi\geq s\}$ contains the asymptotic end (or a suitable piece thereof for large $s$). On the other hand, $\{\ffi\geq s\}$ does not contain $\partial M$, as $N$ (and thus $\ffi$) achieves its minimum value on $\d M$. Furthermore, $\{\ffi\geq s\}$ is connected. In fact, if $\{\ffi\geq s\}$ had a second connected component, this would have to be compact with boundary contained in $\{\ffi=s\}$. Since $\ffi$ is $g$-harmonic, by the maximum and minimum principle we would get that $\ffi$ is constant on this component and thus everywhere by analyticity. Under our hypotheses, the functions $\pra$, $\ffi$ and the metric $g$ are well-defined and smooth on $M$. By continuity it follows that the level sets $\{\ffi=S\}$ are not empty for any $S>s$.

Let $s<S<\infty$ be two regular values for $\ffi$, and let us integrate the inequality~\eqref{eq:positive_divergence} over $\{s<\ffi<S\}$. As an application of the divergence theorem, we get
\begin{align}
-\int_{\{\ffi=s\}}\frac{1}{\pra N}\langle\nabla|\nabla\ffi|_g^p\,|\,{\rm n}_g\rangle+\int_{\{\ffi=S\}}\frac{1}{\pra N}\langle\nabla|\nabla\ffi|_g^p\,|\,{\rm n}_g\rangle\geq0.
\end{align}
Since $\langle\nabla|\nabla\ffi|_g^p\,|\,{\rm n}_g\rangle=p|\nabla\ffi|_g^{p-1}\nabla^2\ffi({\rm n}_g,{\rm n}_g)$, recalling~\eqref{eq:Hg_nanaffi}, we immediately obtain
\begin{align}\label{eq:auxiliary_integral_inequality}
   \int_{\{\ffi=s\}}\frac{1}{\pra N}|\nabla\ffi|_g^{p-1}\,\Hg\,d\sigma_g\geq \int_{\{\ffi=S\}}\frac{1}{\pra N}|\nabla\ffi|_g^{p-1}\,\Hg\,d\sigma_g. 
\end{align}
In case $s$ or $S$ are critical values for $\ffi$, the same inequality still holds. This can be proven exactly as in~\cite{Mazz}, by observing that the vector field $X_{p}$ from \eqref{eq:X} is bounded with bounded divergence for any $p\geq 3$, so that the divergence theorem still applies.

In order to conclude the proof of \eqref{eq:positive_int}, it is enough to show that the right-hand side of~\eqref{eq:auxiliary_integral_inequality} goes to zero as $S\to\infty$. To this end, let us recall again~\eqref{eq:Hg_nanaffi} to deduce that $|\Hg|=|\nana\ffi({\rm n}_g,{\rm n}_g)|\leq|\nana\ffi|_g$. Furthermore, from~\eqref{eq:limit_naffi},~\eqref{eq:limit_nanaffi}, and~\eqref{eq:limit_area}
we know that $|\nabla\ffi|_g$, $|\nana\ffi|_g$ and $\int_{\{\ffi=S\}}d\sigma_g$ are bounded near infinity. Since we also have $N\to 1$, $\pra\to\infty$ at infinity, it follows immediately that the right-hand side of \eqref{eq:auxiliary_integral_inequality} goes to zero as $S\to\infty$ as desired.

It only remains to prove the rigidity statement. Retracing our proof, we deduce that in the equality case, the right hand side of~\eqref{eq:positive_divergence} is equal to zero in $M\setminus\partial M$ or in other words, equality holds in \eqref{eq:bochner}. From~\eqref{eq:bochner_originalp} it then follows that $|\nana\ffi|_g\equiv 0$ and $\nabla\vert\nabla\varphi\vert_{g}\equiv0$ on $\{\ffi\geq s\}$. It is now straightforward to conclude by applying \cite[Theorem 4.1]{AgoMazz_first} that the solution is isometric to one half of a round cylinder (including the boundary). Here, we implicitly use \Cref{rmk:complete}. Returning back to the original metric and using \eqref{eq:g}, \eqref{eq:ffi},  \eqref{eq:relation_N_psi}, \eqref{eq: idPsiN}, and the asymptotic decay assumption on $\Psi$ from \eqref{eq:asymptotics}, we deduce that $g_{0}=\frac{1}{N^{2}}d\rho^{2}+\rho^{2}g_{\mathbb{S}^{2}}$ on $\{\varphi\geq s\}$, with $N=\sqrt{1-\frac{2\mm}{\rho}+\frac{\qq^{2}}{\rho^{2}}}$ and $\Psi=\frac{\qq}{\rho}$ (up to the isometry). By monotonicity of the left hand side, taking the limit as $s\to \varphi_{0}$ gives that $(M\setminus\partial M,\go,N,\Psi)$ is isometric to a Reissner--Nordstr\"om solution~\eqref{eq:RN} with mass $\mm$ and charge $\qq$, up to the boundary by continuity in the non-degenerate case.
\end{proof}

\begin{rmk}
The above proposition can be applied to obtain a monotonicity formula for the functions $F_p(s)=\int_{\{\ffi=s\}}|\nabla\ffi|_g^p \,d\sigma_g$. This is the main tool employed in~\cite{Mazz} and~\cite{CedCoFer} to prove the uniqueness in vacuum for black holes and non-degenerate equipotential photon surfaces, respectively. We could repeat the same arguments in our framework as well, with small modifications, but we prefer to show a faster strategy that requires far less technicalities. Nevertheless, the method involving the monotonicity of the functions $F_p$ will still be of use later on, see \Cref{sub:monotonicity} and \Cref{pro:monotonicity_Fp}.
\end{rmk}

\subsection{Proof of \Cref{thm:BH} for non-degenerate horizons}\label{sub:horizon_consequences}
In this section, we will give a proof of \Cref{thm:BH} for non-degenerate horizons. Before discussing our strategy of proof, let us recall that the quotient $\frac{\mm}{\qq}$ has already been fixed in \eqref{eq:m_over_q} depending on the values $N_0$ and $\Ele_0$ of the lapse function and electric potential on $\d  M$. However, we still have a degree of freedom in the choice of the two parameters. The second condition we would like to impose on $\mm$ and $\qq>0$ is that $|\nabla\ffi|_g=1$ shall hold on $\d  M$, to make $g$ a unit radius cylinder in the equality case in view of \eqref{eq:gphi}. In the non-degenerate horizon case, recalling \eqref{eq:rho-}, \eqref{def:k}, $k>0$ as discussed at the end of \Cref{sub:simplification}, \eqref{eq:naffi_DN}, and the fact that $\vert{\rm D}N\vert=\text{const.}>0$ on $\partial M$, this translates into
\begin{align}\label{eq: Phi_1_bound}
  \frac{\rho_0^6}{(\mm\rho_0-\qq^2)^2}\,|\mathrm{D} N|^2&= 1
\end{align}
on $\d  M$, where $\rho_0=\mm+\qq\sqrt{k}$ is the value of the pseudo-radial function $\rho_{-}$ at $\partial M$. The above formula can be rewritten as the following identity prescribing the value of the parameter $\qq$ (and consequently of $\mm$)
\begin{align}
\label{eq:second_condition_on_mq}
\qq&=\frac{\left(\frac{\mm}{\qq}-\sqrt{k}\right)^2\sqrt{k}}{|\mathrm{D}N|\vert_{\partial M}}>0.
\end{align}
Notice that the right hand side is a function of $\frac{\mm}{\qq}$ whose value has been fixed already, and of $|{\rm D} N|\vert_{\partial M}>0$ which is a given constant on the horizon $\partial M$. It follows from \eqref{def:k} that the parameters $\qq>0$ and $\mm$ are now well-defined for non-degenerate horizons.

Sticking with the non-degenerate horizon case, we are now ready to state the second result we will need in the proof, which consists in an application of a minimum principle argument in the spirit of~\cite[Proposition~3.3]{BorMazII}.

\begin{prop}\label{pro:grad_est}
Let $(M,\go,N,\Ele)$ be an asymptotically flat electrostatic electro-vacuum system with mass parameter $\mu$ and with connected non-degenerate horizon  boundary $\d M$. Let $\pra=\rho_{-}$, $\ffi=\varphi_{-}$, $g=g_{-}$ be defined as in~\eqref{eq:rho-}, \eqref{eq:ffi}, \eqref{eq:g}, respectively, with respect to the parameters $\mm>\qq>0$ determined by \eqref{eq:m_over_q} and~\eqref{eq:second_condition_on_mq}. Then
\begin{align}\label{eq:gradient_estimate}
|\nabla\ffi|_g\leq 1
\end{align}
holds on $M$. Furthermore, if equality holds in~\eqref{eq:gradient_estimate} at some point of $M\setminus\partial M$ then $(M,\go,N,\Ele)$ is isometric to the sub-extremal Reissner--Nordstr\"om solution~\eqref{eq:RN} with mass $\mu$ and charge $q>0$ such that $\frac{\mu}{q}=\frac{\mm}{\qq}$. 
\end{prop}

\begin{proof}
Consider the function $w\definedas\frac{1}{\pra}(1-|\nabla\ffi|_g^2)$. From \eqref{eq:relation_N_psi}, \eqref{eq:ffi}, and \eqref{def:k}, we compute
\begin{align*}
\nabla w&=-Nw\nabla\ffi-\frac{1}{\pra}\,\nabla|\nabla\ffi|_g^2
\end{align*}
on $M$, from which, using~\eqref{eq:bochner_original}, obtain
\begin{align*}
\Deg w- \left(1-\frac{3\mm}{\pra}+\frac{2\qq^2}{\pra^2}\right)\langle\nabla w|\nabla\ffi\rangle_g&=-\frac{2}{\pra}\,|\nana\ffi|_g^2\leq 0.
\end{align*}
Therefore the minimum principle\textsuperscript{\ref{foot:max}} applies, meaning that $w$ takes its minimum value on the boundary $\partial M$ or at infinity. On the boundary, we have $|\nabla\ffi|_g=1$ by our choice of $\mm$ and $\qq$, hence $w=0$ on $\partial M$. By \eqref{eq:limit_naffi} and \eqref{eq:rhoasy}, the quantity $|\nabla\ffi|_g$ is bounded at infinity, whereas $\frac{1}{\pra}\to 0$ and thus $w\to 0$ at infinity. It follows that we must have $w\geq 0$, or equivalently $|\nabla\ffi|_g\leq 1$, on the whole manifold $M$. 

Furthermore, if $|\nabla\ffi|_g=1$ (equivalently, $w=0$) at some interior point, then the strong minimum principle\textsuperscript{\ref{foot:max}} tells us that indeed $|\nabla\ffi|_g\equiv 1$ on $M$. Substituting this into~\eqref{eq:bochner_original}, we also get $|\nana\ffi|_{g}\equiv 0$ in $M$. Then the same argument given at the end of the proof of \Cref{pro:integral_Hg} allows us to conclude that the solution is isometric to the necessarily sub-extremal Reissner--Nordstr\"om solution of mass $\mm$ and charge~$\qq$.
\end{proof} 

Notice that, combining estimate~\eqref{eq:gradient_estimate} with the asymptotic expansion~\eqref{eq:limit_naffi} for $|\nabla\ffi|_g$, it immediately follows that
\begin{align}\label{eq:massestimate}
\mm \leq \ADM
\end{align}
for systems $(M,g_{o},N,\Psi)$ with non-degenerate black hole inner boundary. We will now proceed to prove the opposite inequality. This will trigger rigidity in \Cref{pro:grad_est}, leading to a proof of \Cref{thm:BH} for non-degenerate horizons.

To show $\mm\geq\mu$, we start by recalling formula~\eqref{eq:Hg} for the mean curvature of a level set of $\ffi$, replacing $\vert{\rm D}N\vert$ by $\vert\nabla\varphi\vert_{g}$ via \eqref{eq:naffi_DN}, and obtain
\begin{align}\label{eq:Hg_horizon}
\frac{\Hg}{N}\,|\nabla\ffi|_g&=\frac{\pra^4}{2(\mm\pra-\qq^2)}\left(\operatorname{R}^\Sigma+|\mathring{\operatorname{h}}|^2-\frac{\HHH^2}{2}\right)-\frac{2\mm\pra-\qq^2}{\mm\pra-\qq^2}\,|\nabla\ffi|_g^2
\end{align}
on any level set of $\varphi$. We now want to take the limit of this formula as $N\to 0$, that is, as we approach the horizon $\partial M$. Since $\mm$ and $\qq$ have been chosen so that $|\nabla\ffi|_g=1$ on $\partial M$, recalling that non-degenerate horizons are totally geodesic with respect to $g_{0}$ (see \Cref{sec:setup}), from~\eqref{eq:Hg_horizon} it is easy to compute 
\begin{align}
\lim_{N\to 0}\frac{\Hg}{N}&=\frac{\prao}{q\sqrt{k}}\left(\frac{\prao^2}{2}\operatorname{R}^{\d M}-1\right).
\end{align}
Applying  \Cref{pro:integral_Hg} by dividing~\eqref{eq:positive_int} by $N$ and taking the limit as we approach $\d M$, we find
\begin{align}
\int_{\d M}\left(\frac{\prao^2}{2}\,\operatorname{R}^{\d M}-1\right)d\sigma\geq 0,
\end{align}
exploiting that $k>0$, $\qq>0$, and $\rho_{0}>0$. From the Gauss--Bonnet formula, we then deduce the area bound
\begin{align}\label{eq:GB}
|\d M|\leq4\pi \prao^2=|\Sphere^2|\rho_{0}^{2},
\end{align}
where $\vert\partial M\vert$ denotes the area of $\partial M$ with respect to $g_{0}$ and $|\Sphere^2|$ denotes the area of $\mathbb{S}^{2}$ with respect to $g_{\mathbb{S}^{2}}$. Let us now exploit the $g$-harmonicity of $\ffi$. Integrating $\Deg\ffi=0$ between two level sets of $\ffi$, we find out immediately that the function 
\begin{align}\label{def:F}
F(s)\definedas\int_{\{\ffi=s\}}|\nabla\ffi|_g \,d\sigma_g
\end{align}
 is constant, which in particular implies
\begin{align}\label{eq: identity_int_naffi}
   \int_{\d M}|\nabla\ffi|_g\,d\sigma_g&=\lim_{s\to\infty}\int_{\{\ffi=s\}}|\nabla\ffi|_g\,d\sigma_g.
\end{align}
Recalling that $|\nabla\ffi|_g=1$ and $\pra=\prao$ on $\d M$, combining this identity with~\eqref{eq:GB}, we obtain
\begin{align}
\label{eq:formula10_GB}
|\Sphere^2|\,&\geq\,\lim_{s\to\infty}\int_{\{\ffi=s\}}|\nabla\ffi|_g\,d\sigma_g.
\end{align}
On the other hand, from the asymptotic behaviours~\eqref{eq:limit_naffi} and~\eqref{eq:limit_area}, it follows that
\begin{align} 
\label{eq: lim_int_grad_phi}
  \lim_{s\to\infty}\int_{\{\ffi=s\}}|\nabla\ffi|_g\,d\sigma_g&=\frac{\ADM}{\mm}\,|\Sphere^2|.
\end{align}
Substituting this into~\eqref{eq:formula10_GB}, we find $\mm\geq \ADM$.  It follows that equality must hold in all inequalities in this subsection. In particular, the rigidity statement of \Cref{pro:grad_est} allows us to conclude that  $(M,\go,N,\Psi)$ is isometric to the sub-extremal Reissner--Nordstr\"om system~\eqref{eq:RN} with mass $\mm$ and charge $\qq$. This proves the non-degenerate horizon case of \Cref{thm:BH}.

\subsection{Degenerate horizons}\label{sub:degenerate}
To conclude the proof of \Cref{thm:BH}, it remains to discuss the degenerate horizon case, which we will address in this Section.

As it is pointed out in \Cref{sec:setup}, the degenerate case is more delicate as the metric does not extend to $\partial M$.  Hence, when looking at the space-time picture, we need to be more precise and observe that the connected \emph{Killing horizon}
$\partial\mathfrak{L}$ defined in \Cref{sec:setup} 
corresponds to a \emph{Killing prehorizon}, according to \cite[Section~2.5]{Chrusciel-KillingHorizon}.  In these terms, a \emph{Killing horizon}  is a connected component of a \emph{Killing prehorizon}, which is a null embedded hypersurface such that the Killing vector field never vanishes on it and represents its null generator. The main difference between the non-degenerate and the degenerate case is that $\partial M$ is part of the \emph{Killing prehorizon} $\partial \mathfrak{L}$ for the former, while is not for the latter. This explains why for the degenerate case, we can not work directly on $\partial M$, but rather we first need  to establish some notation and results, as follows.

Throughout this Section, we will denote a  Killing Horizon by $\mathcal{H}$.

\subsubsection{Gaussian null coordinates}\label{sub: GNC}
In order to study degenerate horizons, we can make use of the 
Gaussian null coordinates, first introduced by \cite{MoncriefIsemberg} on a neighborhood of smooth null hypersurfaces. In particular, we consider such coordinates $(v, z, x^1, x^2)$ on a neighborhood $[0, 1) \times \mathcal{H}$ of a Killing horizon $\mathcal{H} \subset \partial\mathfrak{L}$. The coordinate $v$ is a parameter value along the integral curve of the null generator of $\mathcal{H}$, which on the Killing horizon coincides with the Killing vector field. The coordinates $(x^1, x^2)$ are local coordinates on a cross-section of $\mathcal{H}$, while $z$ is the projection onto the first factor, so that $\{z=0\} = \mathcal{H}$. 
 The spacetime metric reads 
\[ \mathfrak{g} = z^2 A \, dv \otimes dv + 2dv \otimes dz + zf_a \, dv \otimes dx^a + h_{ab}\,dx^a\otimes dx^b,\]
where $A = A(z, x^1, x^2), \;f_a = f_a(z, x^1, x^2), \;h_{ab} = h_{ab}(z, x^1, x^2)$, are components of a smooth function $A$, a one-form $f$ and a non-degenerate metric $h$, respectively. Notice that they are independent of $v$, due to the fact that  $\mathcal{H}$ is a Killing horizon. For more details on the construction of these coordinates see e.g.~\cite{CT, KW}. 
Following~\cite[Section~2]{KW}, the restriction of the Gaussian null coordinates to the time-slice $M$ gives coordinates $(z,x^1,x^2)$ on $\left((0,1)\times \mathcal{H} \right) \cap M$ and the metric on $M$ reads
\begin{align}\label{metricGaussCoordintro} 
g_0 = \frac{1}{z^2 A} dz\otimes dz -\frac{f_a}{z A} (dz\otimes dx^a+dx^a\otimes dz) +\left( h_{ab}+\frac{f_a f_b}{A}\right) dx^a \otimes dx^b.
\end{align}
In the electro-vacuum case, as proven in~\cite{CT}, it holds $A(0,x^1,x^2)\asdefined C^{2}$ for some constant $C>0$. Furthermore $N=z\,C + O(z^2)$, $f_a=O(z)$ as $z\to0$, and $\lim_{z \to 0}(\d _z \Psi)^2 = C^{2}$. 
 Finally, one has that $C^{2}h|_{z = 0}=g_{\mathbb{S}^2}$ is the standard round metric on the sphere.

 From \eqref{metricGaussCoordintro} and the above decay assertions, one computes the components of the inverse metric to behave as
\begin{align}
(g_0)^{zz}&={z^2A}+o(z^2),\\
(g_0)^{za}&=O(z^2),\\
(g_0)^{ab}&=h^{ab}+o(1)
\end{align}
for $a,b=1,2$ as $z\to0$, where $(h^{ab})$ is the inverse of $(h_{ab})$. Since $C^{2} h |_{z = 0}=g_{\mathbb{S}^{2}}$, we have $h^{ab}=O(1)$ as $z\to0$. It follows that
\begin{align}\label{eq:DN_degenerate}
|{\rm D} N|^2&=(g_0)^{zz}(\d _z N)^2+o(z^2)=z^2 C^2 \,(\d _z N)^2+o(z^2)
\end{align}
as $z\to0$. 

\subsubsection{Surface gravity}\label{sub: SurfGrav}
As in \cite[Section~12.5]{Wald}, the \emph{surface gravity} of a Killing horizon $\mathcal{H}$ associated to the Killing vector field $K$ is a quantity $\kappa$ defined on static Killing horizons by
\[ \nabla |K|_{\mathfrak{g}}^2 \, |_{\mathcal{H}}\asdefined -2 \, \kappa \, K \, |_{\mathcal{H}},\]
related to the fact that both $\nabla \vert K \vert_{\mathfrak{g}}^2$ and $K$ are null vector fields orthogonal to the horizon. In particular, in the static case $K = \partial_t$ as pointed out in \Cref{sec:setup}. In addition, as in \cite{Wald}, one can show that 
\begin{equation}\label{eq: SG1}
	\kappa^2 = -\frac{1}{2} \vert \nabla K\vert^2_{\mathfrak{g}} \, |_{\mathcal{H}}.
\end{equation}
The proof is a consequence of the Frobenius formula (see \cite[Appendix~B]{Wald}) for $K$, which is equivalent to $K$ being orthogonal to $\mathcal{H}$ and hence applies even when $\kappa$ vanishes. By a direct computation, in the static case, one has (see e.g. \cite[Section~2.1.4]{BorPhD}) that
\begin{equation}\label{eq: SG2}
\vert \nabla K\vert^2_{\mathfrak{g}} = -2 \, \vert{\rm D}N\vert ^2 = -2 \, \vert dN \vert^2_{\mathfrak{g}}.
\end{equation}
We point out that in the above formula we are using a slight abuse of notation: when we write $\mathrm{D}N$ we see the static potential $N$ as a function on $M$, whereas when we write the differential $dN$ we are seeing $N$ as a function on the spacetime $\mathfrak{L}$, defined in the obvious way, by setting $N(x,t)=N(x)$ for every $(x,t)\in M\times\mathbb{R}=\mathfrak{L}$.

In the literature (see e.g. \cite[Section~2.5]{Chrusciel-KillingHorizon}) the definition of \emph{non-degenerate} and \emph{degenerate} horizon is given according to $\kappa \neq 0$ in the former and $\kappa = 0$ in the latter, as it can be shown that $\kappa$ is constant on the Killing horizon. Combining \Cref{eq: SG1} and \Cref{eq: SG2}, we deduce the following: first, that the definition of \emph{non-degenerate} and \emph{degenerate} horizon provided in \Cref{sec:setup} (i.e. $ dN$ non vanishing or vanishing) is consistent with the usual one; second, that with respect to the Gaussian null coordinates, $\vert{\rm D}N\vert \to 0$ as $z \to 0$. This last observation will be crucial in the following discussion and can not be deduced directly from the expansion~\eqref{eq:DN_degenerate}.

\subsubsection{Proof of \Cref{thm:BH} for degenerate horizons}

We are finally ready to prove \Cref{thm:BH} in the case where $\d M$ is degenerate.

We first show that, in our setup, degenerate horizons are extremal, namely that $k=0$ or in other words $\mm=\qq$. In fact, assuming towards a contradiction that $k \neq 0$, by~\eqref{eq:naffi_DNk} and $\vert{\rm D}N\vert \to 0$ as $z \to 0$ (see \Cref{sub: SurfGrav}), one has $|\nabla \varphi|_g \to 0$ as $z \to 0$ . This easily leads to a contradiction: exploiting the harmonicity of the function $\varphi$ as in~\eqref{eq: identity_int_naffi} and using~\eqref{eq: lim_int_grad_phi}, for any $\varepsilon > 0$, integrating between the sets $S_\varepsilon \coloneqq \{ z = \varepsilon\} \cap M$ and any $\{ \varphi = s\}$ for sufficiently large $s$,  we get
\begin{align}\label{eq:intnablaphi}
   \int_{S_\varepsilon}|\nabla\ffi|_g\,d\sigma_g&=\lim_{s\to\infty}\int_{\{\ffi=s\}}|\nabla\ffi|_g\,d\sigma_g\, = \,\frac{\ADM}{\mm}\,|\Sphere^2|.  
\end{align}
Using the fact that the metric on ${S_\varepsilon}$ approaches the round metric of radius $1/C$ for a positive constant $C > 0$ (see \Cref{sub: GNC}), we have $\lim_{\ep\to 0}|S_\ep|=\frac{|\mathbb{S}^2|}{C^{2}}=\frac{4\pi}{C^{2}}$, the slices $S_\ep$ have finite area, hence the left-hand side of \eqref{eq:intnablaphi} goes to zero and we get a contradiction. 

As a consequence of the previous discussion, if $\d M$ is degenerate we have
 $\mm = \qq$, and thus $\rho\to\mm \asdefined \rho_0$ as $z \to 0$. Moreover, we have $N^2=(\Ele-1)^2$ and thus in particular $\Psi_{0}=1$ and $|\d_z N|=|\d_z \Psi|$. Recalling that $\lim_{z \to 0}|\d _z \Psi|=C$ from \Cref{sub: GNC}, this gives $\lim_{z \to 0}|\d _z N|=C$. Now from~\eqref{eq:DN_degenerate}, we compute $\frac{|{\rm D}N|^2}{N^2} \to C^{2}$. By~\eqref{eq:naffi_DNk} we then have
\begin{align}\label{eq: lim_DNN}
\lim_{z\to0}\vert\nabla\varphi\vert_{g}=\mm^2\lim_{z\to0}\frac{|\mathrm{D}N|^2}{N^2}=\mm^2C^{2}.
\end{align}
As in the non-degenerate case, we can hence fix $\mm$ so that $|\nabla \varphi|_g \to 1$ as $z \to 0$ by setting $\mm=\frac{1}{C}$. From this, we find
\begin{align}
\lim_{\ep\to 0}|S_\ep|&=\frac{4\pi}{C^{2}}= 4\pi\mm^{2}=4\pi\prao^2=\vert\mathbb{S}^{2}\vert \rho_{0}^{2}.
\end{align}

We can now follow the same arguments in \Cref{sub:horizon_consequences} to show that Proposition~\ref{pro:grad_est} holds in the degenerate case as well, and that the area bound~\eqref{eq:GB} is in force (with equality), replacing $|\d M|$ by $\lim_{\ep\to 0}|S_\ep|$. Arguing as at the very end of \Cref{sub:horizon_consequences} gives $\mm=\mu$ thus triggering rigidity in Proposition~\ref{pro:grad_est} and concluding the proof of \Cref{thm:BH}.


\section{Equipotential Photon Surface Uniqueness in Electro-Vacuum}\label{sec:photon_surface}
This section is dedicated to the proof of \Cref{thm:photon_surface}, split into several different cases, closely related to if and where $dN=0$ in the different Reissner--Nordstr\"om spacetimes discussed at the end of \Cref{sec:setup}. We will continue to assume $\Psi_{0}>0$ and will appeal to not yet fixed parameters $\mm\in\R$, $\qq>0$ chosen so that their ratio $\frac{\mm}{\qq}$ satisfies~\eqref{eq:m_over_q}. We will first discuss the case where the critical set $\C=\{ N^2+k =0\}$ is empty and $N_0\leq1$ in \Cref{sub:proof_photon_standard}. This part covers all previously known equipotential photon surface uniqueness results, see \Cref{sec:intro}. In particular, our analysis will characterize all --- necessarily non-degenerate (see \Cref{sub:degenerateequi}) --- equipotential photon surfaces in the sub-extremal ($\mm>\qq>0$) and the extremal case ($\mm=\qq>0$) as well as all super-extremal ones with $\qq>\mm>0$ and $\rho_{0}>\frac{\qq^{2}}{\mm}$, where $\rho_{0}$ denotes the value of the pseudo-radial function $\rho$ on $\partial M$. Its proof resembles the one given for horizons. 

Using slight modifications of the same strategy, we will then treat the cases $\C=\emptyset$ and $N_{0}>1$ in \Cref{sub:neg_mass} and $N_{0}=\sqrt{-k}$ (and thus $\C\neq\emptyset$) in \Cref{subsec:superextremalnocrit}. This will characterize all super-extremal equipotential photon surfaces with $\mm\leq0<\qq$ and with $\qq>\mm>0$ and $\rho_{0}=\frac{\qq^{2}}{\mm}$, respectively. As discussed in \Cref{sub:degenerateequi}, the first ones are necessarily (locally) non-degenerate while the second one is necessarily degenerate.

Finally, after a careful analysis of the critical set $\C$ and its complement when $\C\subset (M\setminus\partial M)$ in \Cref{sub:super-extremal_Cnonempty}, we will treat the case $\C\neq\emptyset$ and $N_{0}>\sqrt{-k}$ in \Cref{subsec:critnonemptyN0large}. This proof deviates the most from the black hole case and uses an approach via monotonicity formulas derived in \Cref{sub:monotonicity}. This will characterize all --- again necessarily (locally) non-degenerate (see \Cref{sub:degenerateequi}) --- super-extremal equipotential photon surfaces with $\qq>\mm>0$ and $\rho_{0}<\frac{\qq^{2}}{\mm}$.

Note that $\C\neq\emptyset$ implies $k<0$ by positivity of $N$ in the equipotential photon surface case, thus the condition $N_{0}\geq\sqrt{-k}$ is well-defined; moreover, by \eqref{def:k} and \eqref{eq:m_over_q}, $N_{0}<\sqrt{-k}$ is not possible. Hence our analysis covers all cases of \Cref{thm:photon_surface}.

As announced in \Cref{sec:setup}, we will denote by $\nu$ the unit normal to the boundary $\d M$ pointing towards the asymptotically flat end.\label{notation: normal} Notice that $\nu$ may coincide with the normal ${\rm n}=\frac{{\rm D}N}{|{\rm D}N|}$ introduced in \Cref{sub:geometry level sets} or may point in the opposite direction.

\subsection{Proof of \Cref{thm:photon_surface} for $\C=\emptyset$ and $N_0\leq 1$}\label{sub:proof_photon_standard}
Recalling $\Psi_{0}>0$ and using $N_{0}\leq1$, this immediately asserts that $\mm>0$ and hence $\mu>0$ by \Cref{lem:sign}. Next, the assumption $\C=\emptyset$ ensures that problem~\eqref{eq:pb_N} is well-posed in $M$. In particular, the second equation in~\eqref{eq:pb_N} gives that $N$ satisfies the maximum and minimum principle\textsuperscript{\ref{foot:max}} in $M$. Since $N\to1$ at infinity, it follows immediately that $N_0\neq1$, otherwise we would have $N\equiv 1$ in $M$ so that, since $\Ele$ is a function of $N$ and $\Ele\to 0$ at infinity, this would imply that $\Ele$ is trivial, which is against our hypothesis $\Psi_{0}>0$.
We conclude that $N_0<1$ and, by the maximum and minimum principle\textsuperscript{\ref{foot:max}} and the Hopf lemma\textsuperscript{\ref{foot:Hopf}}, we have $N_{0}<N<1$ in $M$ and $\nu(N)>0$ on $\d M$ or equivalently $\nu={\rm n}$ on $\partial M$. Moreover, $N\neq1$ on $M$ so that $\rho_{-}$ is well-defined and smooth on $M$ by \Cref{lem:rho-}, using again our assumption $\C=\emptyset$. This allows us to choose the pseudo-radial function $\rho_{-}$ from \eqref{eq:rho-} as in the black hole case treated in \Cref{sec:BH}, with $\varphi=\varphi_{-}$ and $g=g_{-}$ as in \eqref{eq:ffi}, \eqref{eq:g}, respectively. Finally, $\C=\emptyset$ implies that $\partial M$ is non-degenerate by the discussion in \Cref{sub:degenerateequi} so that \Cref{thm: ps_properties} applies.

It is now easy to see that the arguments given in \Cref{sub:asymptotics,sub: standard_Conformal} work in this framework as well. Proceeding exactly in the same way as in the proof of \Cref{pro:integral_Hg}, we integrate the divergence formula~\eqref{eq:bochner} and use the consequences of asymptotic flatness established in \Cref{sub:asymptoticsimplification,sub:asymptotics} to establish that 
\begin{align}
\label{eq:positive_int_phonon_standard}
\int_{\{\ffi=s\}}|\nabla\ffi|_g^{p-1}\Hg\,d\sigma_g\geq  0
\end{align}
for any level set $\{\ffi=s\}$ and any $p\geq 2$, with equality if and only if $(M,\go,N,\Ele)$ is isometric to a suitable piece of the Reissner--Nordstr\"om solution~\eqref{eq:RN} with mass $\mu$ and charge $q>0$ such that $\frac{\mu}{q}=\frac{\mm}{\qq}$. As in the black hole case, it is convenient to use the additional degree of freedom that we have in the choice of $\mm$ and $\qq>0$ to fix the value of $|\nabla\ffi|_g$ at the boundary $\partial M$. More precisely, we ask that
\begin{align}
\label{eq:naffi_fixed_photon_standard}
|\nabla\ffi|_g^2&=\frac{\prao^6}{(\mm \prao-\qq^2)^2}\,|\mathrm{D} N|^2=1
\end{align}
on $\d M$, noting that $\vert{\rm D}N\vert=\nu(N)>0$ and that $\vert{\rm D}N\vert$ is constant on $\partial M$ by \Cref{thm: ps_properties}. This is equivalent to asking that
\begin{align}\label{eq:qqphotolessthan1}
\qq&=\frac{\left(\frac{\mm}{\qq}-\sqrt{N_0^2+k}\right)^2\sqrt{N_0^2+k}}{|\mathrm{D}N|\vert_{\partial M}}.
\end{align}
As $\frac{\mm}{\qq}$ and $k$ have already been fixed via~\eqref{eq:m_over_q} and \eqref{def:k}, \eqref{eq:qqphotolessthan1} uniquely determines $\qq>0$ (and consequently $\mm$). Next, note that we actually did not use directly the fact that $\d  M = \{N = 0\}$ in the proof of \Cref{pro:grad_est}. Instead, we have used that $|\nabla \varphi|_g = 1$ on $\d  M$. As this continues to hold here, we can apply (the proof of) \Cref{pro:grad_est} to find that $|\nabla\ffi|_g\leq1$ in $M$, with equality at some interior point if and only if $(M,\go,N,\Ele)$ is isometric to a suitable piece of the Reissner--Nordstr\"om solution~\eqref{eq:RN} with mass $\mm$ and charge $\qq$. Combining this with \eqref{eq:limit_naffi}, we find also in this case that
\begin{align}\label{eq:grad_est_photon_standard}
\mm \leq \ADM 
\end{align}
with equality if and only if $(M,\go,N,\Ele)$ is isometric to a suitable piece of the Reissner--Nordstr\"om solution~\eqref{eq:RN} with mass $\mu$ and charge $q>0$ such that $\frac{\mu}{q}=\frac{\mm}{\qq}$.

Since we know that the mean curvature $\HHH$ of a canonical time-slice of a non-degenerate equipotential photon surface is constant by \Cref{thm: ps_properties}, we deduce from~\eqref{eq:naffi_DNk} and~\eqref{eq:H_Hg} that $\Hg$ is also constant on $\partial M$. Thus, \eqref{eq:positive_int_phonon_standard} gives $\Hg\geq0$ on $\d M$. This inequality gives us a couple of important consequences. First of all, from \eqref{eq:qqphotolessthan1} and the relation~\eqref{eq:H_Hg} by $\nu={\rm n}$ between $\HHH=\HHH_{\nu}$ and $\Hg$, observing that 
\begin{align}\label{eq:restriction}
\mm \prao-\qq^2&=\frac{\qq\sqrt{N_0^2+k}\,(\mm+\qq\sqrt{N_0^2+k})}{1-N_0^2}\,>\,0,
\end{align}
where $\prao$ is the value of $\pra$ on $\d M$, we deduce that
\begin{align}\label{eq:Hestimate}
\HHH_{\nu}&\geq\frac{2N_0\prao^2}{\mm\prao-\qq^2}\,|{\rm D}N|=\frac{2N_0}{\prao}
\end{align}
so that in particular $\HHH_{\nu}>0$ (compare the discussion after \eqref{normalLapse}, recalling $\nu={\rm n}$ in this case). On the other hand, recalling umbilicity of $\partial M$ from \Cref{thm: ps_properties}, combining $\Hg\geq 0$ with \eqref{eq:Hg}, \eqref{eq:rho-}, \eqref{eq:qqphotolessthan1}, and \eqref{eq:Hestimate}, we get
\begin{align}
\operatorname{R}^{\d M}\geq\frac{\HHH_{\nu}^2}{2}+\frac{2\prao^2(2\mm\prao-\qq^2)}{(\mm\prao-\qq^2)^2}\,|{\rm D}N|^2\geq
\frac{2N_0^2}{\prao^2}+\frac{2(2\mm\prao-\qq^2)}{\prao^4}&=\frac{2}{\prao^2}.
\end{align}
Integrating this over $\d M$ and using the Gauss--Bonnet formula, we obtain
\begin{align}\label{eq:applyGB}
4\pi\chi(\d M)&=\int_{\d M}\operatorname{R}^{\d M}d\sigma\geq\frac{2|\d M|}{\prao^2}.
\end{align}
It follows that the Euler characteristic of $\d M$ is positive, meaning that $\d M$ is diffeomorphic to a sphere. Hence $\chi(\d M)=2$ and $|\d M|\leq 4\pi\prao^2$. In other words, we have recovered the area bound~\eqref{eq:GB} and we can now conclude the proof exactly as in \Cref{sub:horizon_consequences}; namely we combine~\eqref{eq: identity_int_naffi} with~\eqref{eq: lim_int_grad_phi} to get $|\d M|=\frac{\ADM}{\mm}\,|\mathbb{S}^2|$, from which it follows that $\ADM\leq\mm$ as desired. This allows us to conclude via \eqref{eq:grad_est_photon_standard} as in the black hole case treated in \Cref{sec:BH}. Moreover, we learned from \eqref{eq:restriction} that $\rho_{0}>\frac{\qq^{2}}{\mm}$ which imposes the anticipated restriction on $\rho_{0}$ when $\qq>\mm>0$ and no restriction on $\rho_{0}$ when $\mm\geq\qq>0$ as $\rho_{0}>r_{\mm,\qq}>\frac{\qq^{2}}{\mm}$ in that case.

\subsection{Proof of \Cref{thm:photon_surface} for $\C = \emptyset$ and $N_0 > 1$}\label{sub:neg_mass}
Again, we start by observing that, since the critical set $\C$ is empty, the second equation in~\eqref{eq:pb_N} is a well-posed elliptic equation for $N$ on the whole manifold $M$. Since $N \to 1$ at infinity, we conclude by the maximum and minimum principle\textsuperscript{\ref{foot:max}} and the Hopf lemma\textsuperscript{\ref{foot:Hopf}} that $1<N<N_0$ holds on $M$ and $\nu(N)<0$ or in other words $\nu=-{\rm n}$ holds on $\d M$. Next, as usual, let $\mm\in\R$, $\qq>0$ be chosen so that their ratio $\frac{\mm}{\qq}$ satisfies~\eqref{eq:m_over_q}. Since $N > 1$ by the above and $\Ele>0$  in $M$ by \Cref{pro:simplification}, we deduce from\eqref{eq: idPsiN} that $\Ele>\frac{2\mm}{\qq}$ in $M$. Since $\Ele\to 0$ at infinity, we conclude that $\mm\leq0$ as expected. Since $\mm \leq 0$, \Cref{lem:rho+,lem:rho-} suggest to use the pseudo-radial function $\rho_{+}$ from~\eqref{eq:rho+}, with $\varphi=\varphi_{+}$ and $g=g_{+}$ from \eqref{eq:ffi} and \eqref{eq:g}, respectively. Moreover, $\C=\emptyset$ implies that $\partial M$ is non-degenerate by the discussion in \Cref{sub:degenerateequi} so that \Cref{thm: ps_properties} applies.

Planning to proceed as before, we first need to analyze the asymptotic behavior of $\rho$, $\varphi$, and $g$. We will first discuss the argument for $\mm<0$ and then treat the case $\mm=0$ in \Cref{subsub:mm=0}; note that the latter case needs to be handled separately because of its different asymptotic behavior.

\subsubsection{The sub-case $\mm<0$}\label{subsub:mmnot0}
In this case, recalling that $\mu<0$ follows from $\mm<0$ by \Cref{lem:sign}, the asymptotic considerations in \Cref{sub:asymptotics} all remain valid upon replacing $\rho_{-}$, $\varphi_{-}$, and $g_{-}$ with $\rho_{+}$, $\varphi_{+}$, and $g_{+}$, respectively, and $N_{0}<t_{0}<1$ by $N_{0}>t_{0}>1$ --- except that the sign on the asymptotics of ${\rm n}$ in \eqref{eq:DN_asymptotics} and \eqref{eq:ndecay} is reversed (with corresponding sign changes in the computation directly after \eqref{eq:ndecay}). Consequently, we have $\nu=-{\rm n}={\rm n}_{g}$ on $\partial M$.

It is now easy to see that the arguments given in \Cref{sub: standard_Conformal} work in this framework as well. Proceeding exactly as in the proof of \Cref{pro:integral_Hg}, we integrate the divergence formula~\eqref{eq:bochner} and use the consequences of asymptotic flatness established in \Cref{sub:asymptoticsimplification,sub:asymptotics} to establish that \eqref{eq:positive_int_phonon_standard} also holds here for any level set $\{\ffi=s\}$ and any $p\geq 2$, with equality if and only if $(M,\go,N,\Ele)$ is isometric to a suitable piece of the Reissner--Nordstr\"om solution~\eqref{eq:RN} with mass $\mu$ and charge $q>0$ such that $\frac{\mu}{q}=\frac{\mm}{\qq}$. Next, we fix the choice of $\mm$ and $\qq>0$ so that $\vert\nabla\varphi\vert_{g}=1$ on $\partial M$. As in \Cref{sub:proof_photon_standard}, this amounts to asking that~\eqref{eq:naffi_fixed_photon_standard} is in place, but note that this time $\nu=-{\rm n}={\rm n}_{g}$ or
\begin{align}\label{eq:choiceq}
1&=|\nabla\ffi|_g=-\frac{\prao^3}{\mm \prao-\qq^2}\,|\mathrm{D} N|
\end{align}
on $\partial M$ which is equivalent to
\begin{align}\label{eq:qqphotobiggerthan1}
\qq&=\frac{\left(\frac{\mm}{\qq}+\sqrt{N_0^2+k}\right)^2\sqrt{N_0^2+k}}{|\mathrm{D}N|\vert_{\partial M}}>0.
\end{align}
 As $\frac{\mm}{\qq}$ and $k$ have already been fixed via~\eqref{eq:m_over_q} and \eqref{def:k}, \eqref{eq:qqphotobiggerthan1} uniquely determines $\qq>0$ (and consequently $\mm$). Proceeding as in \Cref{sub:proof_photon_standard}, we can now apply (the proof of) \Cref{pro:grad_est} to find that $|\nabla\ffi|_g\leq1$ in $M$, with equality at some interior point if and only if $(M,\go,N,\Ele)$ is isometric to a suitable piece of the Reissner--Nordstr\"om solution~\eqref{eq:RN} with mass $\mu$ and charge $q>0$ such that $\frac{\mu}{q}=\frac{\mm}{\qq}$. As before, we also deduce $\mm\leq\ADM$ with equality if and only if $(M,\go,N,\Ele)$ is isometric to a suitable piece of the Reissner--Nordstr\"om solution~\eqref{eq:RN} with mass $\mu$ and charge $q>0$ such that $\frac{\mu}{q}=\frac{\mm}{\qq}$.
 
Also, since we know that the mean curvature $\HHH$ of a canonical time-slice of a non-degenerate equipotential photon surface is constant by \Cref{thm: ps_properties}, we deduce from~\eqref{eq:naffi_DNk} and~\eqref{eq:H_Hg} that $\Hg$ is also constant on $\partial M$. Thus, \eqref{eq:positive_int_phonon_standard} gives $\Hg\geq0$ on $\d M$ also here. Since clearly $\mm \pra_0-\qq^2<0$ in this case, we deduce from~\eqref{eq:H_Hg} that $\HHH<0$ on $\d M$. As ${\rm n}=-\nu$, this shows that the mean curvature $\HHH_{\nu}=-\HHH$ with respect to the normal $\nu$ pointing towards the asymptotically flat end is again positive (compare the discussion after \eqref{normalLapse}). More precisely, using \eqref{eq:choiceq}, \eqref{eq:H_Hg} gives 
\begin{align}\label{eq:Hestimate_negativemass}
\HHH_{\nu}&\geq-\frac{2N_0\prao^2}{\mm\prao-\qq^2}|{\rm D}N|=\frac{2N_0}{\prao},
\end{align}
whereas from~\eqref{eq:Hg} and \eqref{eq:choiceq}, we get
\begin{align}\label{eq:Restimate_negativemass}
\operatorname{R}^{\d M}&\leq\frac{\HHH_{\nu}^2}{2}+\frac{2(2\mm\prao-\qq^2)}{\prao^4},
\end{align}
using again \Cref{thm: ps_properties}. This time, combining the scalar curvature estimate \eqref{eq:Restimate_negativemass} with the Gauss--Bonnet formula is inconclusive. Instead, we will exploit the photon surface constraint \eqref{foundtheconstant} in the form of \eqref{scalarMeanCurv} and~\eqref{normalLapse}, recalling that we have already established $\HHH_{\nu}>0$. Rewriting these conditions using \eqref{eq:qqphotobiggerthan1}, \eqref{eq:choiceq}, $\nu(N)=-\vert{\rm D}N\vert$, and \eqref{eq:relation_DN_DPsi}, we find
\begin{align}\label{eq:scalarMeanCurv_negativemass}
\operatorname{R}^{\d M}&= \frac{2\qq^2}{\prao^4}+\left(c + \frac{1}{2} \right)\HHH_{\nu}^2,\\\label{eq:normalLapse_negativemass}
c\HHH_{\nu} &= \frac{2(\mm\prao-\qq^2)}{N_{0}\prao^3} .
\end{align}
From~\eqref{eq:Hestimate_negativemass}, \eqref{eq:normalLapse_negativemass}, and $\nu(N)=-\vert{\rm D}N\vert<0$, we then deduce
\begin{align}\label{eq:cestimate}
\frac{\mm\prao-\qq^2}{N_{0}^2\prao^2}\leq c<0.
\end{align}
Computing $\HHH_{\nu}$ from \eqref{eq:normalLapse_negativemass} and substituting it into \eqref{eq:scalarMeanCurv_negativemass}, we get by \eqref{eq:relation_N_psi} that
\begin{align*}
\frac{c^2 N_{0}^2\prao^6}{2}\operatorname{R}^{\d M}&=c^2 N_{0}^2 \prao^2 \qq^2+(2c+1)(\mm\prao-\qq^2)^2\\
&=\qq^2\left[\frac{N_{0}^2\prao^2 c-(m\prao-\qq^2)}{N_{0}\prao}\right]^2\!+ \prao(\mm\prao-\qq^2)^2\left[\frac{2\mm c}{\mm\prao-\qq^2}+\frac{\prao-2\mm}{N_{0}^2\prao^2}\right]\!.
\end{align*}
Recalling that $\mm$ and $\mm\prao-\qq^2$ are both negative, we can then apply \eqref{eq:cestimate} in the term in the last square bracket above to deduce 
\begin{align}
\operatorname{R}^{\d M}&\geq\frac{2(\mm\prao-\qq^2)^2}{N_{0}^4 \prao^6 c^2}.
\end{align}
Applying \eqref{eq:cestimate} again gives
\begin{align}\label{eq:scalbelow}
\operatorname{R}^{\d M}&\geq\frac{2}{\prao^2}.
\end{align}
We now conclude exactly as in \Cref{sub:proof_photon_standard}, by integrating this bound on the scalar curvature and using the Gauss--Bonnet formula. This concludes the proof for $\mm<0$.

\subsubsection{The sub-case $\mm=0$}\label{subsub:mm=0}
Let us now discuss the case when $\mm=0$ and hence $\mu=0$ by \Cref{lem:sign}. We begin by analyzing the asymptotic decay of $\rho$, $\mu$, and $g$ which follows from the additional decay assumption \eqref{eq:asymmu0} we are making in case $\mu=0$. From \Cref{lem:signkappa}, we know that $\kappa>0$. Using $N^{2}=1+\Psi^{2}$ from \eqref{eq: idPsiN} and $\mm=0$ and \eqref{eq:nupsi} with $\kappa>0$ so that $\nu(\Psi)<0$ on the leaves of $N$, we find $\nu(N)=-\vert{\rm D}N\vert$ on the leaves of $N$ and thus ${\rm n}=-\nu$ asymptotically. Moreover, we see that
\begin{align}\label{eq:asyNm0}
N^{2}&=1+\frac{\kappa^{2}}{\vert x\vert^{2}}+o_{2}(\vert x\vert^{-2}),\\\label{eq:asyDNm0}
\vert {\rm D}N\vert&=\frac{\kappa^{2}}{\vert x\vert^{3}}+o_{1}(\vert x\vert^{-3})
\end{align}
as $\vert x\vert\to\infty$. As in \Cref{subsub:mmnot0}, the asymptotic considerations in \Cref{sub:asymptotics} up to and including \eqref{eq:gbar} all remain valid upon replacing $N_{0}<t_{0}<1$ by $N_{0}>t_{0}>1$ --- except that the sign on the asymptotics of ${\rm n}$ in \eqref{eq:DN_asymptotics} and \eqref{eq:ndecay} is reversed (with corresponding sign changes in the computation directly after \eqref{eq:ndecay}). Moreover, we again have $\nu=-{\rm n}={\rm n}_{g}$ on $\partial M$.

Furthermore, replacing $\rho_{-}$, $\varphi_{-}$, and $g_{-}$ with $\rho_{+}$, $\varphi_{+}$, and $g_{+}$, respectively, as above, we immediately find $\rho=\frac{\qq}{\Psi}$ and hence $\varphi=\log(\frac{\qq(N+1)}{\Psi})$. Exploiting \eqref{eq:asymmu0} with $\kappa>0$ as well $\qq>0$, we thus find from \eqref{eq:rho+}, \eqref{eq:ffi}, and \eqref{eq:asyNm0} that
\begin{align}\label{eq:rhoasym0}
\pra &= \frac{\qq |x|}{\kappa} + o_{2}(|x|),\\\label{eq:phiasym0}
\ffi&=\log \left(\frac{2\qq |x|}{\kappa}\right)+o_{2}(1)
\end{align}
 as $\vert x\vert\to\infty$. Inserting \eqref{eq:rhoasym0}, \eqref{eq:phiasym0}, \eqref{eq:asyNm0}, and \eqref{eq:asyDNm0} into \eqref{eq:naffi_DNk} and \eqref{eq:nana_ffi} gives
\begin{align}\label{eq:limit_naffim0}
|\nabla\ffi|_g&= \frac{\qq}{\kappa} + o_1(1),\\ 
\nana_{ij}\ffi&=o(|x|^{-3})
\end{align}
as $\vert x\vert\to\infty$. Furthermore, by \eqref{eq:g} and \eqref{eq:rhoasym0}, we get
\begin{align}\label{eq:limit_nanaffim0}
|\nana\ffi|_g&=o(\vert x\vert^{-1})
\end{align}
as $\vert x\vert\to\infty$. Arguing as in \Cref{sub:asymptotics}, we obtain
\begin{align}\label{eq:limit_aream0}
\lim_{s\to\infty}\int_{\{\ffi=s\}}d\sigma_g &=\frac{\kappa^2}{\qq^2} \,|\mathbb{S}^2|.
\end{align}
With these asymptotic assertions, it is straightforward to see that \eqref{eq:positive_int_phonon_standard} also holds here for any level set $\{\ffi=s\}$ and any $p\geq 2$, with equality if and only if $(M,\go,N,\Ele)$ is isometric to a suitable piece of the Reissner--Nordstr\"om solution~\eqref{eq:RN} with mass $\mm=\mu=0$ and charge $\kappa$. Next, as usual, we fix the choice of $\qq>0$ so that $\vert\nabla\varphi\vert_{g}=1$ on $\partial M$. As above, this amounts to asking that~\eqref{eq:choiceq} holds on $\partial M$ which is equivalent to
\begin{align}\label{eq:qqphotobiggerthan1m0}
\qq&=\frac{\sqrt{N_0^2-1}^{\,3}}{|\mathrm{D}N|}>0.
\end{align}
 This clearly uniquely determines $\qq>0$. Proceeding as in \Cref{sub:proof_photon_standard}, we can now apply (the proof of) \Cref{pro:grad_est} to find that $|\nabla\ffi|_g\leq1$ in $M$, with equality at some interior point if and only if $(M,\go,N,\Ele)$ is isometric to a suitable piece of the Reissner--Nordstr\"om solution~\eqref{eq:RN} with mass $\mm=0$ and charge $\qq$. Furthermore, from \eqref{eq:limit_naffim0}, we deduce $\qq\leq\kappa$ with equality if and only if $(M,\go,N,\Ele)$ is isometric to a suitable piece of the Reissner--Nordstr\"om solution~\eqref{eq:RN} with mass $\mm=\mu=0$ and charge $\kappa$. Arguing as in \Cref{subsub:mmnot0}, we again find \eqref{eq:scalbelow} on $\partial M$ and can conclude as at the end of \Cref{sub:proof_photon_standard} to find \eqref{eq:applyGB} and hence $\vert\partial M\vert\leq 4\pi\rho_{0}^{2}$. On the other hand, from \eqref{eq: identity_int_naffi}, \eqref{eq:limit_naffim0}, and \eqref{eq:limit_aream0}, we find that $\kappa\leq\qq$ which concludes the proof.

\subsection{Proof of \Cref{thm:photon_surface} for $N_0 =\sqrt{-k}$}\label{subsec:superextremalnocrit}
From $\C\neq\emptyset$, it follows via \eqref{eq:m_over_q} that $0<\mm<\qq$. Next, from $N_{0}=\sqrt{-k}$, we find $\partial M\subseteq\C$ and
\begin{align}\label{eq: radius}
\prao &= \frac{\qq^2}{\mm}
\end{align}
as anticipated. In particular, $\partial M$ is degenerate by \eqref{eq:relation_DN_DPsi}. We claim that $\C=\partial M$. To see this, let us exploit \eqref{eq: Nk_geq_0} to find $\C=\{\Psi=\frac{\mm}{\qq}\}$. Moreover, the last equation in \eqref{eq:pb} is invariant under addition of constants to $\Psi$, hence it also applies to $\overline{\Psi}\definedas\Psi-\frac{\mm}{\qq}$. Note that by $N_{0}=\sqrt{-k}$ and \eqref{eq: Nk_geq_0}, we have $\overline{\Psi}\vert_{\partial M}=0$, while the asymptotic decay assumption on $\Psi$ gives $\overline{\Psi}\to-\frac{\mm}{\qq}<0$ at infinity. Applying the strong maximum principle\textsuperscript{\ref{foot:max}} to $\overline{\Psi}$ thus gives $\overline{\Psi}<0$ or in other words $\Psi<\frac{\mm}{\qq}$ in $M\setminus\partial M$ and hence $\C=\partial M$ as claimed.

As in the degenerate horizon case, $\C=\partial M$ tells us that the equations in~\eqref{eq:pb_N} are well-defined everywhere in the interior of $M$. Moreover, by the maximum principle\textsuperscript{\ref{foot:max}} applied to the second equation in \eqref{eq:pb}, exploiting $N_{0}=\sqrt{-k}<1$ and the asymptotic assumption $N\to1$ at infinity, ensures that $\sqrt{-k}<N<1$ in $M\setminus\partial M$ and in particular $N\neq1$ in $M$. Hence we can use the pseudo-radial function $\rho_{-}$ from~\eqref{eq:rho-} also in this case, we well as $\varphi=\varphi_{-}$ and $g=g_{-}$ from \eqref{eq:ffi}, \eqref{eq:g}, respectively, and note that they are well-defined and smooth on $M\setminus\partial M$ by \Cref{lem:rho-}. Consequently, using $\mm>0$ and thus $\mu>0$ by \Cref{lem:sign}, the asymptotic analysis from \Cref{sub:asymptotics} also applies here, for some $\sqrt{-k}<t_{0}<1$. Moreover, the proof of \Cref{pro:grad_est} can be carried out upon adjusting the minimum principle argument as in the proof of \Cref{pro:simplification}. Thus, we have $\mm \leq \ADM$, with the usual rigidity statement, and we only need to prove the reverse inequality in order to conclude.

To see that also $\mu\geq\mm$, we first recall that by the Hopf lemma\textsuperscript{\ref{foot:Hopf}}, our assumption $\Psi_{0}>0$, and our asymptotic condition \eqref{eq:asymptotics} on $\Psi$, we have $\nu(\Ele)<0$ and thus $\vert{\rm D}\Psi\vert>0$ on $\partial M$. Combining \eqref{eq:naffi_DNk} with \eqref{eq:relation_DN_DPsi} gives $|\nabla\ffi|_g^2=\frac{\pra^4 |\mathrm{D}\Ele|^2}{\qq^2 N^2}$ on $M\setminus\partial M$ which by continuity extends to $\partial M$, leading to
\begin{align}\label{eq:Cdvarphi}
|\nabla\ffi|_g^2&=\frac{\pra_{0}^4 |\mathrm{D}\Ele|^2}{\qq^2 N_{0}^2}
\end{align}
and thus $\vert\nabla\varphi\vert_{g}>0$ on $\partial M$. Exploiting \eqref{eq: radius}, it is now easily seen that we can assert $|\nabla\ffi|_g=1$ on $\d M$ by the unique choice 
\begin{align}\label{eq:defqdegenerate}
\qq&=\frac{\left(\frac{\mm}{\qq}\right)^2\!N_{0}}{|{\rm D}\Psi|\vert_{\partial M}}.
\end{align}
Now, recall that $g_{0}$ smoothly extends to $\partial M$, but $\rho$ is only continuous there. On the other hand, by \eqref{eq:relation_DN_DPsi} and because $N>\sqrt{-k}$ in $M\setminus\partial M$, $N\to1$ at infinity, and $\vert{\rm D}\Psi\vert>0$ on $\partial M$, level sets $\{N=N_{\varepsilon}\}$ of $N$ with $\vert N_{\varepsilon}-N_{0}\vert<\varepsilon$ for $\varepsilon>0$ sufficiently small are connected, close to $\partial M$, and have $\vert{\rm D}N\vert>0$. Hence we can apply \Cref{pro:integral_Hg} on $\{N=N_{\varepsilon}\}$, obtaining \eqref{eq:positive_int} on $\{N=N_{\varepsilon}\}$. On the other hand, applying \eqref{eq:H_Hg} on $\{N=N_{\varepsilon}\}$ and exploiting \eqref{eq:relation_DN_DPsi}, one has
\begin{align}\label{eq:Hgremainswelldef}
\HHH_{g}\vert\nabla\varphi\vert_{g}&=\frac{\rho^{3}}{\qq}\left[\frac{\HHH\vert{\rm D}\Psi\vert}{N}-\frac{2\rho\,\vert{\rm D}\Psi\vert^{2}}{\qq N}\right]
\end{align}
on $\{N=N_{\varepsilon}\}$ for $\varepsilon>0$ sufficiently small. Using \eqref{eq: radius} and \eqref{eq:defqdegenerate}, the right hand side of \eqref{eq:Hgremainswelldef} converges to a well-defined constant as $\varepsilon\to0^+$. Combining this with the limit of \eqref{eq:positive_int} gives 
\begin{align}
 \HHH\geq\frac{2\mm \sqrt{-k}}{\qq^{2}}>0
\end{align}
on $\partial M=\C$, and in particular $\HHH=\HHH_{\nu}>0$ as $N>\sqrt{-k}=N_{0}$ and hence $\nu={\rm n}$ (compare the discussion after \eqref{normalLapse}). Now, recall from \Cref{sub:degenerateequi} that \Cref{thm: ps_properties} and in particular \eqref{foundtheconstant}--\eqref{normalLapse} also apply in the degenerate case. Thus, from \eqref{normalLapse} and $\nu(N) = 0$ we find $c=0$ and thus 
\begin{align} 
\operatorname{R}^{\d  M} &\geq \frac{2\mm^{2}}{\qq^{4}}=\frac{2}{\rho_{0}^{2}}
\end{align}
by \eqref{scalarMeanCurv}, \eqref{eq:defqdegenerate}, and \eqref{eq: radius}. We can now follow the same procedure applied before, namely we integrate over the surface $\d  M = \C$, apply the Gauss--Bonnet formula with respect to $d\sigma$, and use \eqref{eq: identity_int_naffi} (interpreted as a limit to $\partial M=\C$ as above and using continuity of $\rho$) and~\eqref{eq: lim_int_grad_phi} to obtain $\mu\leq\mm$ and thus isometry to a suitable piece of the Reissner--Nordstr\"om solution of mass $\mm$ and charge $\qq$ as desired, with the isometry extending to $\partial M$ by continuity.

\subsection{The structure of $\C$}\label{sub:super-extremal_Cnonempty}
In order to conclude the proof of \Cref{thm:photon_surface}, it remains to consider the case where $\C\neq\emptyset$ and $N_0>\sqrt{-k}$, so that $\partial M\cap\C=\emptyset$. As discussed in \Cref{subsec:superextremalnocrit}, it follows from $\C\neq\emptyset$ that $\qq>\mm>0$. Moreover, it follows from \eqref{eq:relation_DN_DPsi} and the definition of the critical set $\C$ that the gradient of $N$ must vanish on $\C$ and that $\sqrt{-k}$ is in fact necessarily the minimum value attained by $N$ in $M$. As a consequence, we have to be more careful when applying the maximum and minimum principle to the second equation in \eqref{eq:pb_N}. In order to prove \Cref{thm:photon_surface} in this setting, we first need to gain a better understanding of the structure of the set $\C=\{N=\sqrt{-k}\}$. This is the purpose of this subsection.

First of all, we recall from~\cite{Tod} that $N$ is a real analytic function on $M$. In particular, since $\C$ is clearly compact, 
the {\L}ojasiewicz Structure Theorem on the level sets of a real analytic function (established in~\cite{Lojasiewicz}, see also~\cite[Theorem~6.3.3]{Krantz_Parks}) guarantees us that $\C$ admits a decomposition
\begin{align}
\C=\C_0\sqcup \C_1\sqcup \C_2,
\end{align}
where $\C_0$ is a finite collection of points, $\C_1$ is a finite collection of real analytic curves and $\C_2$ is a finite collection of real analytic surfaces. A point $p\in\C$ belongs to $\C_i$, $i=0,1,2$, if $i$ is the largest integer such that there exists a neighborhood $\Omega$ of $p$ in $M$ and a real analytic diffeomorphism $\phi:\Omega\to\R^3$ such that $\phi(\Omega\cap\C_i)=L\cap\phi(\Omega)$ for some $i$-dimensional linear space $L$.

This structure of $\C$ is crucial in the following result, which tells us that the set $\C$ is connected and disconnects the manifold. 

\begin{lemma}\label{lem:structure_C}
Let $(M,g_0,N,\Ele)$ be an asymptotically flat electrostatic electro-vacuum system with connected boundary $\partial M$. Let  $N=N_{0}$ and $\Psi=\Psi_{0}$ hold on $\partial M$ for constants $N_{0}>\sqrt{-k}$, $\Psi_{0}>0$ and assume that $\C=\{N^2+k=0\}\neq\emptyset$. 

Then $\C$ disconnects the manifold $M$ into two pieces, one bounded \emph{inner piece $M_i$} containing $\partial M$ and one  \emph{outer piece $M_o$} containing the asymptotically flat end. Furthermore, every connected component of $\C$ intersects $\Gamma \definedas  \overline{M_i}\cap\overline{M_o}\subset \C$ and, at any point $p\in \C$, $|\mathrm{D}\Ele|(p)\neq 0$ if and only if $p\in \C_2\cap \Gamma$.
\end{lemma}

\begin{proof}
First, recall that $\qq>\mm>0$ follows from $\C\neq\emptyset$. Then for any choice of pseudo-radial function $\rho_{\pm}$, we compute 
\begin{align}\label{def:varphi0}
\varphi\vert_{\C}&=\log\left(\frac{\qq^{2}(1+\sqrt{-k})}{\mm}-\mm\right)\asdefined\varphi_{0}.
\end{align}
Next, we note that $\varphi_{0}$ is in fact a locally minimal value of $\varphi$ near $\C$: if $\rho=\rho_{-}$, using $\mm>0$ and $N>\sqrt{-k}$ away from $\C$ readily gives $\varphi_{-}>\varphi_{0}$ on $M\setminus(\C\cup\{N=1\})$ (recalling \Cref{lem:rho-}). As $N=\sqrt{-k}\neq1$ on $\C$ by $\mm>0$, $N\neq1$ holds near $\C$. If $\rho=\rho_{+}$, one finds from $\mm>0$ that $\varphi_{+}>\varphi_{0}$ in $\{\sqrt{-k}+\delta>N>\sqrt{-k}\,\}$ for suitably small $\delta>0$. By compactness of $\C$, there hence exists a bounded neighborhood $\Omega$ of $\C$ with smooth boundary such that $\varphi_{\pm}>\varphi_{0}$ or in other words $\varphi>\varphi_{0}$ on $\overline{\Omega}\setminus\C$.

Now suppose towards a contradiction that $\C\neq\emptyset$ and $M\setminus \C$ is connected. From \eqref{eq: Nk_geq_0}, it then follows that $\Ele-\frac{\mm}{\qq}$ is either positive or negative on $M\setminus \C$ (as it can only vanish on $\C$) and thus its value $\frac{\mm}{\qq}$ on $\C$ is either a maximum or a minimum value of $\Ele$ in $M$. As a consequence, we have $|\mathrm{D}\Ele|=0$ on $\C$. Next, by \eqref{eq:Cdvarphi}, we get $\vert\nabla\varphi\vert_{g}^{2}\to0$ as we approach $\C$, regardless of which choice of pseudo-radial function $\rho_{\pm}$ we use. We can now use the $g$-harmonicity of $\ffi$ to get a contradiction as follows, again without specifying which choice of pseudo-radial function we make: consider a regular value $s$ of $\varphi$ satisfying $\varphi_{0}<s<\min_{\d  \Omega}(\ffi)$. Then the set $\{\ffi\leq s\}\cap \overline{\Omega}$ is compact with smooth boundary $\{\ffi=s\}\cap\overline{\Omega}$. Given any small neighborhood $U\subset\Omega$ of $\C$ with smooth boundary $\partial U$, integrating the identity $\Deg\ffi=0$ in $(\{\ffi\leq s\}\setminus U)\cap \Omega$ and using the divergence theorem, we obtain
\begin{align}
\int_{\d  U}\langle\nabla\ffi\,|\,{\rm n}_g\rangle_g\, d\sigma_g&=\int_{\{\ffi=s\}\cap\Omega}|\nabla\ffi|_g\, d\sigma_g.
\end{align}
Since we have shown that $|\nabla\ffi|_g$ goes to zero as we approach $\C$, the left hand side of the above identity can be made as small as we want by choosing a small enough open set $U$ (and appealing to the Cauchy--Schwarz inequality). It follows that the right hand side must be equal to zero. Since this is true for any regular value $s$ of $\ffi$ close enough to $\ffi_0$ and as regular values of $\varphi$ are dense by Sard's lemma, it follows that $|\nabla\ffi|_g$ vanishes identically in $\Omega\setminus\C$, contradicting that $\varphi>\varphi_{0}$ in $\Omega\setminus\C$, $\varphi=\varphi_{0}$ on $\C$, and continuity of $\varphi$ on $\Omega$. This gives the desired contradiction, so it follows that $\C$ must disconnect $M$ into more than one connected component.

We now proceed to showing that there are exactly two connected components $M_i$ and $M_o$ of $M\setminus\C$  as in the claim of \Cref{lem:structure_C}. As $N\to1$ at infinity and $\sqrt{-k}<1$, it is clear that the asymptotically flat end of $M$ is contained in one connected component $M_{o}$ which is then clearly non-compact. Next, as $\partial M$ is connected and $N_{0}>\sqrt{-k}$, $\partial M$ is also entirely contained in one connected component of $M\setminus\C$; this could a priori also be $M_{o}$ or a separate component $M_{i}$. Now suppose towards a contradiction that there is a connected component $M_{e}$ of $M\setminus\C$ such that $M_{e}\cap\partial M=\emptyset$ and $M_{e}\neq M_{o}$. It follows that $M_{e}$ is bounded with $\partial M_{e}\subseteq\C$ which leads to a contradiction via the maximum principle\textsuperscript{\ref{foot:max}} applied to the second equation in \eqref{eq:pb_N} as $N>\sqrt{-k}$ in $M_{e}$ and $N=\sqrt{-k}$ on $\C$. Thus $M_{e}$ cannot exist and the connected component containing $\partial M$ must be a component $M_{i}\neq M_{o}$ and hence bounded.

We will now establish that every connected component of $\C$ must intersect $\Gamma$, with $\Gamma=\overline{M_i}\cap\overline{M_o}\subset \C$. In fact, if there were some other connected component $C$ of $\C$ inside $\overline{M_i}$ or $\overline{M_o}$ disconnected from $\Gamma$, i.e., satisfying $C\cap\Gamma=\emptyset$, we would get a contradiction by applying the arguments explained at the beginning of the proof. Namely, we show that $\vert{\rm D}\Psi\vert=0$ on $C$ as $\Psi-\frac{\mm}{\qq}$ has a sign on $M_{i}$ and $M_{o}$ and vanishes on $C$ and hence $|\nabla\ffi|_g=0$ on $C$; we then apply the divergence theorem to the integral of $\Deg\ffi$ in a small neighborhood $U$ of $C$ to prove that $|\nabla\ffi|_g$ vanishes identically in $U$, contradicting that $\varphi_{0}$ is a locally minimal value near $\C$.

Finally, let $p\in\Gamma$. If $p\in\C_2\cap\Gamma$, then $p$ belongs to the smooth portion of the boundary of $\overline{M_{i}}$. The boundary of $\overline{M_{i}}$ is the union of $\d M$, where $\Psi=\Psi_0\neq\frac{\mm}{\qq}$ by \eqref{eq: Nk_geq_0}, and of $\Gamma$, where $\Psi=\frac{\mm}{\qq}$. Then the Hopf lemma, applied in $M_i$ to the third equation in~\eqref{eq:pb} (which is strongly elliptic as $N\geq N_{0}>\sqrt{-k}$ in $	\overline{M_i}$), tells us that $|\mathrm{D}\Ele|(p)\neq 0$. It remains to show that $|\mathrm{D}\Ele|(p)=0$ if $p\not\in \C_2\cap\Gamma$. Assuming this were not the case, we would get that $\Delta N\vert_{p}\neq 0$ from the second equation in~\eqref{eq:pb} which gives the desired contradiction by~\cite[Theorem~3.3]{BorChrMaz}.
\end{proof}

Thanks to \Cref{lem:structure_C}, applying the Hopf lemma\textsuperscript{\ref{foot:Hopf}} to the second equation in \eqref{eq:pb_N} on $M_i$, we immediately deduce that $\nu(N)<0$ and hence $\nu=-{\rm n}$ on $\partial M$.

\subsection{A monotonicity argument}\label{sub:monotonicity}
We have seen in the previous subsection that if $\C\neq\emptyset$ and $N_0>\sqrt{-k}$, $\C$ disconnects the manifold $M$ into two pieces, an inner one $M_i$ containing $\d  M$ and an outer one $M_o$ containing the asymptotically flat end. Since $N\neq\sqrt{-k}$ away from $\C$, the second equation in~\eqref{eq:pb_N} is well-posed in both $M_i$ and $M_o$. It follows that the maximum and minimum principle\textsuperscript{\ref{foot:max}} apply to $N$ in $M_o\setminus U$ and $M_i\setminus U$, where $U$ is any small neighborhood of $\C$. Since the neighborhood $U$ can be taken to be arbitrarily small and $N$ is continuous on $M$, recalling that $N\to 1$ at infinity, we deduce that
\begin{align}\label{eq:rangeN}
\begin{split}
\sqrt{-k}&<N<N_0\ \text{inside}\, M_i,\\
\sqrt{-k}&<N<1\ \ \ \text{inside}\, M_o.
\end{split}
\end{align}
Let us now perform our conformal change in both $M_i$ and $M_o$. In $M_o$, we use the pseudo-radial function $\pra_-$ defined by~\eqref{eq:rho-} which is well-defined there as $N\neq1$ by \eqref{eq:rangeN}, whereas in $M_i$, we use the pseudo-radial function $\pra_+$ defined by~\eqref{eq:rho+}. By definition, $\pra_+=\pra_-=\frac{\qq^2}{\mm}$ holds on $\C$, hence the combined function 
\begin{align}
\pra&\definedas
\begin{dcases}
\pra_+ & \hbox{in }\overline{M_i},\\
\pra_- & \hbox{in }\overline{M_o},
\end{dcases}
\end{align} 
is well-defined and continuous in $M$ and smooth in $M\setminus \C$ by \Cref{lem:rho-,lem:rho+}. The same naturally applies to the combined function $\varphi$ and metric $g$. However, a priori there is no reason to expect higher regularity. In particular, it is hard to exploit the maximum principle argument described in \Cref{pro:grad_est} in the present situation. Instead, we will rely on an alternative argument, more in line with~\cite{Mazz}, finding functions that are monotonic on the level sets of $\ffi$ in both the inner region $M_i$ and the outer region $M_o$, and then showing that the monotonicity formulas in the two regions combine appropriately.

First of all, let us observe that \Cref{pro:integral_Hg} is in force in this setting as well.
\begin{prop}\label{pro:integral_Hg_superext}
In the setting of \Cref{lem:structure_C}, let $\pra$ be defined by~\eqref{eq:rho-} in $M_o$ and by~\eqref{eq:rho+} in $M_i$, with respect to parameters $\mm,\qq$ satisfying~\eqref{eq:m_over_q} and with asymptotic mass parameter $\mu$. Finally, let $\ffi$, $g$ be defined by~\eqref{eq:ffi} and~\eqref{eq:g} respectively. Then, for any $p\geq 3$ and any $s<\infty$ such that $\{\ffi=s\}\subset M$, we have
\begin{align}\label{eq:positive_int_superext}
\int_{\{\ffi=s\}}|\nabla\ffi|_g^{p-1}\Hg\,d\sigma_g\geq0.
\end{align}
Furthermore, if equality holds in~\eqref{eq:positive_int} for some $s$ and $p$, then $(M,\go,N,\Ele)$ is isometric to a suitable piece of the Reissner--Nordstr\"om solution~\eqref{eq:RN} with mass $\mu$ and charge $q$ such that $\frac{\mu}{q}=\frac{\mm}{\qq}$.
\end{prop}

\begin{proof}
Giving the same arguments used at the beginning of the proof of \Cref{pro:integral_Hg}, we see that $\{\ffi\geq s\}$ is connected and contains the infinity. If $\{\ffi\geq s\}$ does not contain $\C$ then $\{\ffi\geq s\}\subseteq M_o$. We can then proceed as in \Cref{pro:integral_Hg} to obtain immediately the desired inequality, recalling that we use $\rho=\rho_{-}$ in $M_{i}$. This concludes the proof for those values $s$ for which $\{\ffi\geq s\}\subseteq M_o$. In particular,
\begin{align*}
\lim_{s\to\ffi_0^+}\int_{\{\ffi=s\}}|\nabla\ffi|_g^{p-1}\Hg \,d\sigma_g\geq0,
\end{align*}
where $\ffi_0$ is given by \eqref{def:varphi0}, provided this limit exists. We know from \Cref{lem:structure_C} that $\vert{\rm D}\Psi\vert=0$ on $\C\setminus(\C_{2}\cap\Gamma)$ and $\vert{\rm D}\Psi\vert\neq0$ on $\C_{2}\cap\Gamma$. By \eqref{eq:Cdvarphi}, we conclude that $|\nabla\ffi|_g=0$ on $\C\setminus(\C_{2}\cap\Gamma)$ and $|\nabla\ffi|_g\neq0$  on $\C_{2}\cap\Gamma$. As $\C_{2}\cap\Gamma$ is dense in $\Gamma$ by definition of $\C_{2}$, we deduce
\begin{align}\label{eq:intinequality_Gamma1}
\lim_{s\to\ffi_0^+}\int_{\{\ffi=s\}}|\nabla\ffi|_g^{p-1}\Hg\,d\sigma_g&=\int_{\Gamma}|\nabla\ffi|_g^{p-1}\Hg\,d\sigma_g\geq0
\end{align}
converges as $\Hg$ is defined almost everywhere on a critical level set (see the end of \Cref{sub:geometry level sets}) and because $g$ is continuous across $\C$.

It remains to consider the case where $\{\ffi=s\}\subset M_i$, in which case $\{\ffi\geq s\}\supset\C$. In $M_i$, we can still proceed as in the proof of \Cref{pro:integral_Hg}, integrating \eqref{eq:positive_divergence} and using the divergence theorem, obtaining \eqref{eq:auxiliary_integral_inequality} for any $s<S<\ffi_0$. Taking the limit as $\{\ffi=S\}$ approaches the level set $\Gamma$, this time we obtain
\begin{align}\label{eq:intinequality_Gamma2}
\int_{\{\ffi=s\}}\frac{1}{\pra N}|\nabla\ffi|_g^{p-1}\Hg\,d\sigma_g\geq \int_{\Gamma}\frac{1}{\pra N}|\nabla\ffi|_g^{p-1}\Hg\,d\sigma_g.
\end{align}

Now, recall that inequalities~\eqref{eq:intinequality_Gamma1} and~\eqref{eq:intinequality_Gamma2} are computed with respect to different choices of the pseudo-radial function $\rho$. Since the combined $\pra$, $g$, $\ffi$ are continuous along and across $\C$, in order to combine the two inequalities we only need to check that the integrand $|\nabla\ffi|_g^{p-1}\Hg$ has the same value on $\Gamma$ from both sides. By \eqref{eq:Cdvarphi}, we find that $|\nabla\ffi|_g$ coincides on $\Gamma$ from both sides. Moreover, arguing as around \eqref{eq:Hgremainswelldef}, we find that $\Hg|\nabla\ffi|_g$ and thus $\Hg$ coincides on $\Gamma$ from both sides, keeping in mind that $\HHH$ is computed with respect to $\nu={\rm n}$. Hence the estimates in \eqref{eq:intinequality_Gamma1} and \eqref{eq:intinequality_Gamma2} can be combined, leading to \eqref{eq:positive_int_superext}. The rigidity statement is proven as usual, first in $M_{o}$ and subsequently extending to $M_{i}$ if necessary.
\end{proof}

Let us now discuss a first consequence of \Cref{pro:integral_Hg_superext}. Consider the function 
\begin{align}\label{eq:Fp}
F_p(s)&\definedas\int_{\{\ffi=s\}}|\nabla\ffi|_g^p\,d\sigma_g
\end{align}
for all $s$ in the range of $\varphi$. As before, integrating $\Deg\ffi=0$ between two level sets of $\ffi$, we find out immediately that $F_1$ is constant, which in particular implies
\begin{align}
\label{eq:formula0}
\int_{\d M}|\nabla\ffi|_g\,d\sigma_g&=\lim_{s\to\infty}\int_{\{\ffi=s\}}|\nabla\ffi|_g\,d\sigma_g.
\end{align}
In the same spirit, from \Cref{pro:integral_Hg_superext}, we obtain the following:
\begin{prop}\label{pro:monotonicity_Fp}
In the setting of \Cref{lem:structure_C}, let $\pra$ be defined by~\eqref{eq:rho-} in $M_o$ and by~\eqref{eq:rho+} in $M_i$, with respect to parameters $\mm$, $\qq$ satisfying~\eqref{eq:m_over_q} and asymptotic mass parameter $\mu$. Finally, let $\ffi$, $g$ be defined by~\eqref{eq:ffi} and~\eqref{eq:g} respectively. Then, for any $p\geq 3$, the function $F_p$ is locally absolutely continuous, with weak derivative
\begin{align}
F'_p(s)&=-(p-1)\int_{\{\ffi=s\}}|\nabla\ffi|_g^{p-1}\,\Hg\,d\sigma_g\leq 0
\end{align}
for all $s$ in the range of $\varphi$. Furthermore, if equality holds in~\eqref{eq:positive_int} for some $s$ and $p$, then $(M,\go,N,\Ele)$ is isometric to a suitable piece of the Reissner--Nordstr\"om solution~\eqref{eq:RN} with mass $\mu$ and charge $q$ with $\frac{\mu}{q}=\frac{\mm}{\qq}$.
\end{prop}

\begin{rmk}
With some more work, following the steps in~\cite[Theorem~3.2]{Mazz}, one can prove more, namely that $F_p$ is actually continuous and differentiable. However, \Cref{pro:monotonicity_Fp} is enough for our purposes.
\end{rmk}

\begin{proof}
The proof follows the strategy from~\cite[Proposition~3.5]{Ben_Fog_Maz}. We start by looking at $|\nabla\ffi|_g$. We know that $|\nabla\ffi|_g$ is bounded at infinity because of the asymptotic flatness assumption. Furthermore, $|\nabla\ffi|_g$ is a smooth function away from $\C$ and it is bounded on $\C$ thanks to \eqref{eq:Cdvarphi}. It follows easily from these observations that $|\nabla\ffi|_g$ is uniformly bounded, namely, there exists $C>0$ such that $|\nabla\ffi|_g\leq C$ on $M$.
As a consequence, for any regular value $s$ of $\varphi$, we have
\begin{align}
F_p(s)&=\int_{\{\ffi=s\}}|\nabla\ffi|_g^p\, d\sigma_g\leq C^{p-1}\int_{\{\ffi=s\}}|\nabla\ffi|_g\,d\sigma_g=C^{p-1}F_1(s).
\end{align}
Since $F_1$ is constant, it follows that, for every $p\geq 1$, the function $F_p(s)$ is essentially uniformly bounded, namely there exists a constant $K$ such that $F_p(s)\leq K$ for almost every value $s$ (in fact we have proven that this holds for any regular value $s$, and these are dense because of Sard's lemma). This enables us to employ the coarea formula, a tool that will be crucial in the proof, also thanks to~\cite[Remark~2.2]{Ben_Fog_Maz}, using the fact that $\nabla|\nabla\ffi|_g^{p-1}$ vanishes at the critical points of $\varphi$ for any $p\geq 3$.

Next, let $a,b\in\mathbb{R}$ be constants with $a<b<\ffi_0$ or $\ffi_0<a<b$, where $\ffi_0$ again denotes the value of $\ffi$ on $\C$. Let also $\eta\in {\mathcal{C}}_c^{\infty} (a,b)$. By means of the coarea formula, we then compute 
\begin{align*}
\int_a^b \eta'(s)F_p(s)\,ds&=\int_a^b\eta'(s)\int_{\{\ffi=s\}}|\nabla\ffi|_g^p \,d\sigma_g \,ds\\
&=\int_{\{a<\ffi<b\}}\langle\nabla(\eta(\ffi))\,|\,\nabla\ffi\rangle_g\,|\nabla\ffi|_g^{p-1} \,dV_g,
\end{align*}
where $dV_{g}$ denotes the volume element on $M$ induced by $g$. Integrating by parts, recalling $\Deg\ffi=0$, and applying again the coarea formula, we get
\begin{align*}
\int_a^b \eta'(s)F_p(s)\,ds&=-\int_{\{a<\ffi<b\}}\eta(\ffi)\big\langle\nabla|\nabla\ffi|_g^{p-1}\,\big|\,\nabla\ffi\big\rangle_{\!\!g}\,dV_g\\
&=-\int_a^b\eta(s)\int_{\{\ffi=s\}} \bigg\langle\nabla|\nabla\ffi|_g^{p-1}\,\bigg|\,\frac{\nabla\ffi}{|\nabla\ffi|_g}\bigg\rangle_{\!\!g}\,d\sigma_g \,ds.
\end{align*}

On the other hand, if $a<\ffi_0<b$, one obtains the same formula by working separately in $(a,\ffi_0)$ and $(\ffi_0,b)$ and then summing the two integrals (recall that $\C=\{\ffi=\ffi_0\}$ has zero volume). It follows that $F_p$ belongs to $W^{1,1}_{\text{loc}}$, with weak derivative satisfying
\begin{align}
F_p'(s)&=\int_{\{\ffi=s\}} \bigg\langle\nabla|\nabla\ffi|_g^{p-1}\,\bigg|\,\frac{\nabla\ffi}{|\nabla\ffi|_g}\bigg\rangle_{\!\!g}\,d\sigma_g.
\end{align}
Recalling~\eqref{eq:Hg_nanaffi}, the claim follows.
\end{proof}

As a consequence of the previous result, since $F_p$ is locally absolutely continuous and its weak derivative is non-positive, it follows that $F_p$ is monotonically decreasing. In particular, we have
\begin{align}
\label{eq:formula1}
\int_{\d M}|\nabla\ffi|_g^p\,d\sigma_g\geq \lim_{s\to\infty}\int_{\{\ffi=s\}}|\nabla\ffi|_g^p\,d\sigma_g.
\end{align}

\subsection{Proof of \Cref{thm:photon_surface} with $\C \neq \emptyset$ and $N_0>\sqrt{-k}$}\label{subsec:critnonemptyN0large}
We are now in the position to exploit the monotonicity of the functions $F_p$ defined in~\eqref{eq:Fp}. First of all, as usual we fix $\mm$, $\qq$ so that $|\nabla\ffi|_g=1$ on $\d M$, leading to \eqref{eq:choiceq} and \eqref{eq:qqphotobiggerthan1}. Recalling formul\ae~\eqref{eq:formula0},~\eqref{eq:formula1}, we deduce
\begin{align}\label{eq:monotonicity_p1}
\frac{|\d M|}{\prao^2}&=\lim_{s\to\infty}\int_{\{\ffi=s\}}|\nabla\ffi|_g\,d\sigma_g,\\\label{eq:monotonicity_p3}
\frac{|\d M|}{\prao^2}&\geq\lim_{s\to\infty}\int_{\{\ffi=s\}}|\nabla\ffi|_g^p\,d\sigma_g
\end{align}
for $p\geq3$. Now we exploit asymptotic flatness. Clearly the expansions presented in \Cref{sub:asymptotics} work in the current framework as well as $\mm>0$ and hence $\mu>0$ by \Cref{lem:sign}. 
In particular, from~\eqref{eq:limit_naffi} and~\eqref{eq:limit_area}, it follows that
\begin{align*}
  \lim_{s\to\infty}\int_{\{\ffi=s\}}|\nabla\ffi|_g\,d\sigma_g&=\frac{\ADM}{\mm}\,|\Sphere^2|.   \\
  \lim_{s\to\infty}\int_{\{\ffi=s\}}|\nabla\ffi|_g^p\,d\sigma_g&=\left(\frac{\mm}{\ADM}\right)^{p-2}|\Sphere^2|
\end{align*}
for $p\geq3$. Substituting this into~\eqref{eq:monotonicity_p1} and~\eqref{eq:monotonicity_p3}, we obtain
\begin{align}\label{eq:monotonicity_consequence_p1}
\frac{1}{\prao^2}|\d M|\,&=\,\frac{\ADM}{\mm}\,|\Sphere^2|,\\
\label{eq:monotonicity_consequence_p3}
\frac{1}{\prao^2}|\d M|\,&\geq\,\left(\frac{\mm}{\ADM}\right)^{p-2}\,|\Sphere^2|
\end{align}
for $p\geq3$. Combining~\eqref{eq:monotonicity_consequence_p1} and~\eqref{eq:monotonicity_consequence_p3} with $\mm>0$ and $\mu>0$, we then immediately deduce $\mm\leq\mu$. As usual, we now proceed to showing that the opposite inequality is also in place, thus proving $\mm=\ADM$, from which the result follows.

Since $\Hg$ is constant on $\partial M$, it follows from \Cref{pro:integral_Hg_superext} that $\Hg\geq 0$ on $\partial M$. Next, notice that
\begin{align}
\prao&=\frac{\qq^2}{\mm + \qq\sqrt{N_0^2+k}}<\frac{\qq^2}{\mm},
\end{align}
as expected, hence in particular $\mm\prao-\qq^2<0$. Recalling that $\HHH=-\HHH_{\nu}$ as $\nu=-{\rm n}$ was established at the end of \Cref{sub:super-extremal_Cnonempty}, it follows as in \Cref{subsub:mmnot0} that \eqref{eq:Hestimate_negativemass}, \eqref{eq:Restimate_negativemass}, \eqref{eq:scalarMeanCurv_negativemass}, and \eqref{eq:normalLapse_negativemass} hold on $\partial M$. In particular, $\HHH_{\nu}>0$ (compare the discussion after \eqref{normalLapse}). This time however, the parameter $\mm$ is positive and this changes the argument. Namely, from \eqref{eq:Hestimate_negativemass} and \eqref{eq:normalLapse_negativemass}, we find 
\begin{align}\label{eq:clowerbound}
\frac{\mm\prao-\qq^2}{N_{0}^2\prao^2}\leq c<0.
\end{align}
Computing $\HHH_{\nu}$ from \eqref{eq:normalLapse_negativemass} and substituting it into \eqref{eq:scalarMeanCurv_negativemass}, using also \eqref{eq:relation_N_psi}, we get
\begin{align}\notag
\frac{c^2 N_0^2\prao^6}{2}\operatorname{R}^{\d M}&=c^2 N_0^2 \prao^2 \qq^2+(2c+1)(\mm\prao-\qq^2)^2\\\label{eq:scalbelowlastcase}
&=-k c^2 \prao^2\qq^2+(c+1)^2(\mm\prao-\qq^2)^2.
\end{align}
To proceed to estimate $\operatorname{R}^{\partial M}$ from below by $\frac{2}{\rho_{0}^{2}}$ as before, we now appeal to the assumption that $N_0^2 \geq |1-\Psi_0^2|$ holds if $N_0^2 > (1-|\Psi_0|)^2$, noting that $N_{0}^{2}>(1-|\Psi_0|)^2$ follows from $\Psi_{0}>0$ and $0<\mm<\qq$ via \eqref{eq:m_over_q}. Using first $N_{0}^{2}\geq 1-\Psi_{0}^{2}$ and again \eqref{eq:m_over_q}, we find that $\Psi_{0}\geq\frac{\mm}{\qq}$. From this, \eqref{eq:rho+}, \eqref{eq: Nk_geq_0}, we find $\Psi_{0}=\frac{\qq}{\rho_{0}}$ 
and thus, using \eqref{eq:m_over_q} to express $\frac{\mm}{\qq}$ in terms of $N_{0}$ and $\Psi_{0}$, we compute
\begin{align}
1-\frac{\mm}{\rho_{0}}&=1-\frac{\mm}{\qq}\Psi_{0}=1-\frac{\Psi_{0}^{2}+1-N_{0}^{2}}{2}=\frac{1+N_{0}^{2}-\Psi_{0}^{2}}{2}\geq0,
\end{align}
where we have used the assumption $N_{0}^{2}\geq\Psi_{0}^{2}-1$ in the last inequality. This tells us that $\rho_{0}\geq\mm$ and hence 
\begin{align}\label{eq:clowerbound+1}
c+1\geq \frac{\rho_{0}-\mm}{N_{0}^{2}\rho_{0}}\geq0
\end{align}
by \eqref{eq:clowerbound+1} and \eqref{eq:relation_N_psi}. Hence dividing \eqref{eq:scalbelowlastcase} by $c^{2}$ and using \eqref{eq:clowerbound}, \eqref{eq:clowerbound+1}, and \eqref{eq:relation_N_psi}, we obtain
\begin{align}
\frac{\prao^2}{2}\operatorname{R}^{\d M}&\geq -\frac{k\qq^2}{N_{0}^{2} \prao^2}+\frac{(\rho_{0}-\mm)^{2}}{N_{0}^{2}\rho_{0}^{2}}=1
\end{align}
as desired. From here on, one concludes as usual by integrating and using the Gauss--Bonnet formula as in \Cref{sub:proof_photon_standard}. This concludes the proof of \Cref{thm:photon_surface}.\qed

\section{Discussion}\label{sec:discussion}
\subsection{On the technical condition in \Cref{thm:photon_surface}}\label{sec:condition}
\Cref{thm:photon_surface} assumes the technical condition that $N_0^2 \geq |1-\Psi_0^2|$ holds if $N_0^2 > (1-|\Psi_0|)^2$. We have used this condition in \Cref{subsec:critnonemptyN0large} to ensure that $\rho_{0}\geq\mm$ and thus $\operatorname{R}^{\partial M}\geq\frac{2}{\rho_{0}^{2}}$ in the case when $\C\neq\emptyset$ and $N_{0}>\sqrt{-k}$. Let us first point out that this technical condition does not pose any additional restrictions in any of the other cases we have discussed: first of all, $N_0^2 > (1-|\Psi_0|)^2$ is equivalent to asking that $\mm<\qq$ by \eqref{eq: idPsiN} and the assumption $\Psi_{0}>0$ and hence only occurs in the super-extremal case (as we have assumed $\qq>0$). Next, $N_0^2 \geq |1-\Psi_0^2|$ poses no restriction when $\mm\leq0$ as can be seen when converting it to $\Psi_{0}\geq\frac{\mm}{\qq}$ and $1\geq\frac{\mm}{\qq}\Psi_{0}$ by \eqref{eq: idPsiN} (and recalling $\Psi_{0}>0$). Thus the only case when the technical condition possibly adds a restriction is when $0<\mm<\qq$. On the other hand, when $N_{0}=\sqrt{-k}$, we have $\Psi_{0}=\frac{\mm}{\qq}$ by \eqref{eq: Nk_geq_0} and thus automatically $N_{0}^{2}=\vert 1-\Psi_{0}^{2}\vert$. Furthermore, when $\C=\emptyset$, we already know that $\rho\geq\frac{\qq^{2}}{\mm}>\mm$. From \eqref{eq: idPsiN} and \eqref{eq:relation_N_psi}, we compute that either $\Psi=\frac{\qq}{\rho}$ or $\Psi=\frac{\mm}{\qq}-\frac{\qq}{\rho}$, with the latter contradicting our asymptotic assumption that $\Psi\to0$ at infinity and the asymptotic assertion that $\rho\to\infty$ at infinity from \Cref{lem:rho-}, using again $\mm>0$. Hence $N_0^2 \geq |1-\Psi_0^2|$ holds automatically when $\qq>\mm>0$ and $\C=\emptyset$.

Now, let us discuss why we need to rely on this technical condition in \Cref{subsec:critnonemptyN0large}. To close the argument, we would like to use the lower bound on $c$ derived in \eqref{eq:clowerbound} (which is effectively a lower bound on $\HHH_{\nu}$) in order to prove that $\operatorname{R}^{\partial M}\geq\frac{2}{\rho_{0}^{2}}$. Appealing to \eqref{eq:scalbelowlastcase}, this requires estimating $f(c)\definedas\frac{(c+1)^{2}}{c^{2}}$ from below by $f(c_{*})$, using that $c<0$ and that $c\geq c_{*}$ for $c_{*}=\frac{\mm\prao-\qq^2}{N_{0}^2\prao^2}$ by \eqref{eq:clowerbound}. Considering that $f'(c)=-\frac{2(c+1)}{c^{3}}$ is strictly decreasing for $c<-1$, has a global minimum at $c=-1$, and is strictly increasing on $(-1,0)$, we can only expect to get the estimate $f(c)\geq f(c_{*})$ when $c_{*}\geq-1$ which is equivalent to $\rho_{0}\geq\mm$ by \eqref{eq:relation_N_psi}. We close this discussion by stating that the above issue is not visible in the exact Reissner--Nordstr\"om case with $q>m>0$ and $\rho_{0}<\frac{\qq^{2}}{\mm^{2}}$ because in that case the estimate  \eqref{eq:Hestimate_negativemass} on the mean curvature, $\HHH_{\nu}\geq\frac{2N_{0}}{\rho_{0}}$ is sharp on every coordinate sphere and hence $c=c_{*}$ for any coordinate sphere as well as $\operatorname{R}^{\partial M}=\frac{2}{\rho_{0}^{2}}$. In other words, because the estimate $c\geq c_{*}$ ist sharp in this case, it does not matter that $f$ is strictly increasing rather than decreasing for $c<1$ (or $\rho_{0}<\mm$).  

\subsection{Comparison of notions of sub-/super-/extremality}\label{sec:extremalcomparison}
In \Cref{sec:intro}, we have compared our uniqueness results to other uniqueness results in the literature. In particular, when doing so for equipotential photon surfaces (and in particular for photon spheres), we used the words ``sub-/super-/extremal'' as if they were universally defined across the various results in the literature; however, this is not a priori the case, as most authors give their own definitions adapted to their strategy of proof. In this section, we will hence make the effort to compare the various notions of sub-/super-/extremality for equipotential photon surfaces, restricting to the setting where all results apply, i.e., assuming that the electrostatic electro-vacuum system $(M,g_{0},N,\Psi)$ under consideration has \emph{connected, non-degenerate} equipotential photon surface boundary $\partial M$ with $\Psi_{0}>0$ and that it is \emph{asymptotic to Reissner--Nordstr\"om}, meaning that, outside some compact set and in suitable coordinates, we have
\begin{align}
\begin{split}\label{def:asRN}
(g_{0})_{ij}&=\left(1+\frac{2\mu}{\vert x\vert}\right)\delta_{ij}+\mathcal{O}_{2}(\vert x\vert^{-2}),\\
N&=1-\frac{\mu}{\vert x\vert}+\mathcal{O}_{2}(\vert x\vert^{-2}),\\
\Psi&=\frac{\kappa}{\vert x\vert}+\mathcal{O}_{2}(\vert x\vert^{-2})
\end{split}
\end{align} 
for some $\mu, \kappa\in\R$ as $\vert x\vert\to\infty$. We will use the notation from our paper also for the other approaches in order to increase legibility. First, recall that we say in this paper that $(M,g_{0},N,\Psi)$ is sub-extremal / extremal / super-extremal if $\mm>\qq$ / $\mm=\qq$ / $\qq>\mm$, respectively, with $\qq>0$ implied by $\Psi_{0}>0$ and where $\frac{\mm}{\qq}$ is fixed by \eqref{eq:m_over_q}. Using the stronger asymptotic decay assumptions  \eqref{def:asRN} and the relationship between $N$ and $\Psi$ established in \eqref{eq: idPsiN}, one immediately sees that $\frac{\mm}{\qq}=\frac{\mu}{\kappa}$, hence our definitions of extremality align with a definition based on the asymptotic mass $\mu$ and charge $\kappa$. 

Let us begin our comparison by studying notions of sub-/super-/extremality for photon spheres. In \cite[Definition 4.1]{YazLaz} by Yazadjiev--Lazov, a photon sphere is called \emph{non-extremal} if it satisfies
\begin{align}\label{eq:nonextremal}
\frac{\HHH_{\nu}^{2}\vert\partial M\vert}{4\pi}\neq1.
\end{align}
It is then asserted that this equivalent to $\mu\neq\kappa$, where we are using here that $\kappa>0$ by our assumption $\Psi_{0}>0$ as shown in \Cref{subsub:mm=0}. Their uniqueness result \cite[Theorem 4.1]{YazLaz} then pertains to non-extremal photon spheres as in \eqref{eq:nonextremal}, with the notion non-extremality aligned with the one used in this paper. Arguing somewhat differently, it is first established in \cite[Lemma 2.6]{CedGalElec} (see also \cite[Theorem 3.1]{GalMiao}) that $\HHH_{\nu}>0$.  In \cite[Definition 2.9]{CedGalElec} by Cederbaum--Galloway, a photon sphere is defined to be sub-extremal / extremal / super-extremal when 
\begin{align}\label{eq:Hr}
\HHH_{\nu} r_{0}>1\ /\ \HHH_{\nu} r_{0}=1\ /\ \HHH_{\nu} r_{0}<1,
\end{align}
respectively, where $r_{0}>0$ denotes the \emph{area radius} defined by $\vert\partial M\vert\asdefined4\pi r_{0}^{2}$. It is then shown in the proof of \cite[Theorem 3.1]{CedGalElec} that these definitions are equivalent to asking that $m_{0}> q_{0}$ / $m_{0}=q_{0}$ / $ q_{0}>m_{0}$, with $m_{0}>0$, $q_{0}>0$ defined as
\begin{align}\label{eq:q0}
q_{0}&=-\frac{\nu(\Psi)\vert_{\partial M}\, r_{0}^{2}}{N_{0}},\\\label{eq:m0}
m_{0}&=\frac{r_{0}}{3}+\frac{2q_{0}^{2}}{3r_{0}}.
\end{align}
respectively, where we are using that $\Psi_{0}>0$ and $\Psi\to0$ at infinity imply $\nu(\Psi)\vert_{\partial M}>0$ by applying the Hopf lemma\textsuperscript{\ref{foot:Hopf}} to the last equation in \eqref{eq:pb}. Note that positive mass is to be expected from the Reissner--Nordstr\"om solutions when a photon sphere exists. Furthermore, in the case of connected $\partial M$, it is shown in \cite[Appendix A]{CedGalElec} (extending the corresponding argument from \cite{YazLaz}) that these definitions are equivalent to those based on the asymptotic quantities $\mu$ and $\kappa$ and hence to our definitions. In particular, recall that $q_{0}=\kappa$ by \eqref{eq:foundkappa}, so consequently $m_{0}=\mu$ by \eqref{eq: idPsiN}. 

Next, in \cite{Jahns} by Jahns, which mostly aims at an extension to higher dimensions (but presented here only for $4$-dimensional spacetimes), a photon sphere is called sub-extremal / extremal / super-extremal if 
\begin{align}
\frac{\HHH_{\nu}^{2}}{\operatorname{R}^{\partial M}}>\frac{1}{2}\ /\ \frac{\HHH_{\nu}^{2}}{\operatorname{R}^{\partial M}}=\frac{1}{2}\ /\ \frac{\HHH_{\nu}^{2}}{\operatorname{R}^{\partial M}}<\frac{1}{2},
\end{align}
respectively. To see that this is well-defined, note that it is asserted in \cite[Proposition 2.4]{CedGalElec} (see also \cite{CederPhoto,YazLaz}) that the photon surface identity \eqref{normalLapse} holds with $c=1$, relying again on the assertion that $\HHH_{\nu}>0$. Thus in particular $\partial M$ has constant positive scalar curvature $\operatorname{R}^{\partial M}>0$ by \eqref{scalarMeanCurv}. Using again the above assertion that $\HHH_{\nu}>0$ and the fact that $\operatorname{R}^{\partial M}=\frac{8\pi}{\vert\partial M\vert}=\frac{2}{r_{0}^{2}}$ by the Gauss--Bonnet formula, this turns out to be equivalent to \eqref{eq:Hr} (but better adjusted to higher dimensions), and thus equivalent to our definitions. This is also addressed in the proof of \cite[Theorem 3]{Jahns} via quasi-local mass and charge quantities as in \eqref{eq:q0}, \eqref{eq:m0}, except that \cite{Jahns} uses a different sign on $q_{0}$ (which does not affect the result as we know that $\Psi$ and $-\Psi$ can be interchanged). Finally, in \cite[Proposition and Definition 5.12]{CedJaVi} by Cederbaum--Jahns--Vi\v{c}\'anek-Mart\'inez, using again that $\operatorname{R}^{\partial M}>0$ for photon spheres, the quasi-local charge $q_{0}$ as in \eqref{eq:q0} is also used, replacing the area radius $r_{0}$ by the scalar curvature radius $\overline{r}_{0}>0$ defined by $\operatorname{R}^{\partial M}=\frac{2}{\overline{r}_{0}^{2}}$ (again, in view of higher dimensions). Yet by the Gauss-Bonnet formula we know that $r_{0}=\overline{r}_{0}$ here. Assuming that $r_{0}>q_{0}$ with $q_{0}>0$ by the above, a quasi-local mass $\overline{m}_{0}>0$ is then defined as the unique solution of
\begin{align}\label{eq:photomcondition}
\frac{\vert\nu(N)\vert_{\partial M}\vert}{N_{0}}\sqrt{1-\frac{2\overline{m}_{0}}{r_{0}}+\frac{q_{0}^{2}}{r_{0}^{2}}}&=\frac{\overline{m}_{0}}{r_{0}^{2}}-\frac{q_{0}^{2}}{r_{0}^{3}}
\end{align}
(in view of addressing equipotential photon surfaces, see below). The reason why $r_{0}>q_{0}$ is assumed is similar to why we need to assume $N_{0}^{2}\geq\vert1-\Psi_{0}^{2}\vert$ or in other words $\rho_{0}\geq\mm$, see \Cref{sec:condition} and \cite[Remark 5.13]{CedJaVi}. As $\HHH_{\nu}>0$, it follows from the photon surface constraint \eqref{normalLapse} with $c=1$ that $\nu(N)\vert_{\partial M}>0$. Plugging in \eqref{eq:m0} into \eqref{eq:photomcondition} shows that $\overline{m}_{0}=m_{0}$. The parameters $q_{0}$ and $\overline{m}_{0}=m_{0}$ are then used to define sub-/super-/extremality as usual, so the definitions coincide with ours -- provided that $r_{0}>q_{0}$. As \cite{CedJaVi} only addresses uniqueness in the sub-extremal case where $r_{0}>m_{0}+\sqrt{m_{0}^{2}-q_{0}^{2}}>q_{0}$, this does not cause any restrictions.

Let us now move on to general equipotential photon surfaces, where we need to compare our notions of sub-/super-/extremality with those used by Cederbaum--Jahns--Vi\v{c}\'anek-Mart\'inez in \cite{CedJaVi}. As we have already discussed above, we will assume that the equipotential photon surface under consideration is \emph{not} a photon sphere. This will allow us to appeal to the refined analysis for equipotential photon surfaces performed in \cite[Section 5.2]{CedJaVi}. First, note that  \cite{CedJaVi} only discusses sub-/super-/extremality under the additional assumption $\HHH_{\nu}\nu(N)>0$, hence we will also assume this in our discussion. By \cite[Lemma 2.6]{CedGalElec} (see also \cite[Theorem 3.1]{GalMiao}), we know that $\HHH_{\nu}>0$, so this is equivalent to assuming $\nu(N)>0$ or in other words $N_{0}\leq1$ and $\C=\emptyset$ in our approach. Hence, this only excludes some super-extremal cases (of Reissner--Nordstr\"om and according to our definitions) and as \cite{CedJaVi} only treats uniqueness in the sub-extremal case, this does not cause any restrictions. Now, \cite[Theorem and Definition 5.16]{CedJaVi} establishes that a (one-sided) tubular neighborhood $U$ of $\partial M$ in $(M,g_{0},N,\Psi)$ must be isometric to a suitable piece of the spatial factor of a generalized Reissner--Nordstr\"om spacetime of mass $\widetilde{m}$, charge $\widetilde{q}$, parameter $\lambda\in\{\pm1,0\}$, and Einstein manifold base $(\Sigma_{*},\sigma_{*})$ with scalar curvature $\operatorname{R}^{\Sigma_{*}}=2\lambda$. Using $\Psi_{0}>0$, this means that
\begin{align}
g_{0}&=\frac{1}{N^{2}_{\lambda,\widetilde{m},\widetilde{q}}}dr^{2}+r^{2}\sigma_{*},\\
N&=\alpha N_{\lambda,\widetilde{m},\widetilde{q}},\\
\Psi&=\alpha \Psi_{\lambda,\widetilde{m},\widetilde{q}}+\beta
\end{align} 
on $U$ up to isometry, for some constants $\alpha>0$, $\beta\in\R$, where 
\begin{align}
N_{\lambda,\widetilde{m},\widetilde{q}}&=\sqrt{\lambda-\frac{2\widetilde{m}}{r}+\frac{\widetilde{q}^{2}}{r^{2}}},\\
\Psi_{\lambda,\widetilde{m},\widetilde{q}}&=\frac{\widetilde{q}}{r}.
\end{align}
By \cite[Lemma 5.20]{CedJaVi}, $\HHH_{\nu}\nu(N)>0$ excludes $\lambda=0,-1$. On the other hand, using $\lambda=1$, the uniformization theorem informs us that $(\Sigma_{*},\sigma_{*})$ is a round sphere, so that in particular the area radius $r_{0}$ coincides with the scalar curvature radius, moreover, \cite[Proposition and Definition 5.12]{CedJaVi} applies. Using this and observing that  \eqref{eq:q0} and \eqref{eq:photomcondition} are invariant under the above rescaling by $\alpha>0$ and $\beta\in\R$, one finds $\widetilde{q}=q_{0}$ and $\widetilde{m}=\overline{m}_{0}=m_{0}$. When $\nu(N)>0$, \cite[Corollary 5.17]{CedJaVi} establishes that sub-/super-/extremality is then defined in \cite{CedJaVi} as usual via $\widetilde{m}=m_{0}$ and $\widetilde{q}=q_{0}$ showing coinciding notions of sub-/super-/extremality, recalling \eqref{eq: idPsiN} and $\frac{\mm}{\qq}=\frac{\mu}{\kappa}$ by the above. 

We conclude that our notions of sub-/super-/extremality coincide with those existing in the literature (whenever those apply).


\bibliographystyle{amsplain}
\bibliography{sample.bib}

\providecommand{\bysame}{\leavevmode\hbox to3em{\hrulefill}\thinspace}
\providecommand{\MR}{\relax\ifhmode\unskip\space\fi MR }
\providecommand{\MRhref}[2]{%
  \href{http://www.ams.org/mathscinet-getitem?mr=#1}{#2}
}
\providecommand{\href}[2]{#2}
\begin{thebibliography}{10}

\bibitem{AgoBorMaz}
Virginia Agostiniani, Stefano Borghini, and Lorenzo Mazzieri, \emph{On the
  serrin problem for ring-shaped domains},
  \href{https://arxiv.org/abs/2109.11255}{https://arxiv.org/abs/2109.11255},
  2021.

\bibitem{Ago_Fog_Maz-1}
Virginia Agostiniani, Mattia Fogagnolo, and Lorenzo Mazzieri, \emph{Sharp
  geometric inequalities for closed hypersurfaces in manifolds with nonnegative
  {R}icci curvature}, Inventiones Mathematicae \textbf{222} (2020), no.~3,
  1033--1101.

\bibitem{Ago_Fog_Maz-2}
\bysame, \emph{Minkowski inequalities via nonlinear potential theory}, Archive
  for Rational Mechanics and Analysis \textbf{244} (2022), no.~1, 51--85.

\bibitem{Oronzio-Penrose}
Virginia Agostiniani, Carlo Mantegazza, Lorenzo Mazzieri, and Francesca
  Oronzio, \emph{Riemannian penrose inequality via nonlinear potential theory},
  arXiv preprint arXiv:2205.11642 (2022).

\bibitem{AgoMazz_first}
Virginia Agostiniani and Lorenzo Mazzieri, \emph{Riemannian aspects of
  potential theory}, Journal de Math\'{e}matiques Pures et Appliqu\'{e}es.
  Neuvi\`eme S\'{e}rie \textbf{104} (2015), no.~3, 561--586.

\bibitem{Mazz}
\bysame, \emph{On the geometry of the level sets of bounded static potentials},
  Communications in Mathematical Physics \textbf{355} (2017), no.~1, 261--301.

\bibitem{MazzAgo}
\bysame, \emph{Monotonicity formulas in potential theory}, Calculus of
  Variations and Partial Differential Equations \textbf{59} (2019), no.~1.

\bibitem{Oronzio-PMT}
Virginia Agostiniani, Lorenzo Mazzieri, and Francesca Oronzio, \emph{A green's
  function proof of the positive mass theorem}, arXiv preprint arXiv:2108.08402
  (2021).

\bibitem{Ambrozio}
Lucas Ambrozio, \emph{On static three-manifolds with positive scalar
  curvature}, Journal of Differential Geometry \textbf{107} (2017), no.~1,
  1--45.

\bibitem{Beig_Simon}
Robert Beig and Walter Simon, \emph{On the uniqueness of static perfect-fluid
  solutions in general relativity}, Communications in Mathematical Physics
  \textbf{144} (1992), no.~2, 373--390.

\bibitem{Ben_Fog_Maz}
Luca Benatti, Mattia Fogagnolo, and Lorenzo Mazzieri, \emph{{M}inkowski
  {I}nequality on complete {R}iemannian manifolds with nonnegative {R}icci
  curvature}, 2021, \href{https://arxiv.org/abs/
  2101.06063}{https://arxiv.org/abs/2101.06063}.

\bibitem{Besse}
Arthur~L. Besse, \emph{Einstein manifolds}, Classics in Mathematics,
  Springer-Verlag, Berlin, 2008, Reprint of the 1987 edition.

\bibitem{BorPhD}
Stefano Borghini, \emph{On the characterization of static spacetimes with
  positive cosmological constant}, 2018, PhD thesis, Scuola Normale Superiore.

\bibitem{Bor_staticnegative}
\bysame, \emph{Static black hole uniqueness for nonpositive masses}, Nonlinear
  Analysis \textbf{220} (2022), 112843.

\bibitem{BorChrMaz}
Stefano Borghini, Piotr~T. Chru{\'s}ciel, and Lorenzo Mazzieri, \emph{On the
  uniqueness of {S}chwarzschild-de {S}itter spacetime}, Archive for Rational
  Mechanics and Analysis \textbf{247} (2023), no.~2, 22.

\bibitem{BorMasMaz}
Stefano Borghini, Giovanni Mascellani, and Lorenzo Mazzieri, \emph{Some sphere
  theorems in linear potential theory}, Transactions of the American
  Mathematical Society \textbf{371} (2019), no.~11, 7757--7790.

\bibitem{BorMaz-I}
Stefano Borghini and Lorenzo Mazzieri, \emph{On the mass of static metrics with
  positive cosmological constant: I}, Classical and Quantum Gravity \textbf{35}
  (2018), no.~12, 125001.

\bibitem{BorMaz-collection}
\bysame, \emph{Monotonicity formulas for static metrics with non-zero
  cosmological constant}, Contemporary Research in Elliptic PDEs and Related
  Topics (2019), 129--202.

\bibitem{BorMazII}
\bysame, \emph{On the mass of static metrics with positive cosmological
  constant: {II}}, Communications in Mathematical Physics \textbf{377} (2020),
  no.~3, 2079--2158.

\bibitem{ndimunique}
Carla Cederbaum, \emph{Rigidity properties of the {S}chwarz\-schild manifold in
  all dimensions}, in preparation.

\bibitem{CDiss}
Carla Cederbaum, \emph{{The Newtonian Limit of Geometrostatics}}, Ph.D. thesis,
  FU Berlin, 2012,
  \href{http://arxiv.org/abs/1201.5433v1}{arXiv:\linebreak[1]1201.5433v1}.

\bibitem{CederPhoto}
Carla Cederbaum, \emph{{Uniqueness of photon spheres in static vacuum
  asymptotically flat spacetimes}}, Complex Analysis \& Dynamical Systems VI,
  Contemp. Math, vol. 667, AMS, 2015, pp.~86--99.

\bibitem{CedCoFer}
Carla Cederbaum, Albachiara Cogo, and Axel Fehrenbach, \emph{{U}niqueness of
  {E}quipotential {P}hoton {S}urfaces in {V}acuum}, forthcoming.

\bibitem{CCLP}
Carla Cederbaum, Albachiara Cogo, Benedito Leandro, and Jo{\~a}o~Paulo dos
  Santos, \emph{Uniqueness of static vacuum asymptotically flat black holes and
  equipotential photon surfaces in $n+1$ dimensions \`a la {R}obinson},
  forthcoming.

\bibitem{CedGalElec}
Carla Cederbaum and Gregory~J. Galloway, \emph{Uniqueness of photon spheres in
  electro-vacuum spacetimes}, Classical Quantum Gravity \textbf{33} (2016),
  no.~7, 075006, 16.

\bibitem{cedergal}
\bysame, \emph{Uniqueness of photon spheres via positive mass rigidity},
  Communications in Analysis and Geometry \textbf{25} (2017), no.~2, 303--320.

\bibitem{CedGalSurface}
\bysame, \emph{Photon surfaces with equipotential time slices}, Journal of
  Mathematical Physics \textbf{62} (2021), no.~3, Paper No. 032504, 22.

\bibitem{CGM}
Carla Cederbaum, Melanie Graf, and Jan Metzger, \emph{Initial data sets that do
  not satisfy the {R}egge--{T}eitelboim conditions}, in preparation.

\bibitem{CedJaVi}
Carla Cederbaum, Sophia Jahns, and Olivia Vi\v{c}\'anek-Mart\'inez, \emph{On
  equipotential photon surfaces in (electro-)static spacetimes of arbitrary
  dimension},
  \href{https://arxiv.org/abs/2311.17509}{https://arxiv.org/abs/2311.17509},
  2023.

\bibitem{Chrusciel}
Piotr~T. Chru\'{s}ciel, \emph{Towards a classification of static electrovacuum
  spacetimes containing an asymptotically flat spacelike hypersurface with
  compact interior}, Classical and Quantum Gravity \textbf{16} (1999), no.~3,
  689--704.

\bibitem{LivRev}
Piotr~T. Chru\'sciel, Jo{\~a}o~Lopes Costa, and Markus Heusler,
  \emph{Stationary {B}lack {H}oles: {U}niqueness and {B}eyond}, Living Reviews
  in Relativity \textbf{15} (2012), no.~7.

\bibitem{Chrusciel-KillingHorizon}
Piotr~T. Chru{\'s}ciel, Jo{\~a}o~Lopes Costa, and Markus Heusler,
  \emph{Stationary black holes: uniqueness and beyond}, Living Rev. Relativ.
  \textbf{15} (2012), 73 (English), Id/No 2012-7.

\bibitem{Chrusciel_Simon}
Piotr~T. Chru\'sciel and Walter Simon, \emph{Towards the classification of
  static vacuum spacetimes with negative cosmological constant}, Journal of
  Mathematical Physics \textbf{42} (2001), no.~4, 1779--1817.

\bibitem{CT}
Piotr~T. Chru\'{s}ciel and Paul Tod, \emph{The classification of static
  electro-vacuum space-times containing an asymptotically flat spacelike
  hypersurface with compact interior}, Communications in Mathematical Physics
  \textbf{271} (2007), no.~3, 577--589.

\bibitem{CVE}
Claudel Clarissa-Marie, K.S. Virbhadra, and G.F.R. Ellis, \emph{The geometry of
  photon surfaces}, Journal of Mathematical Physics \textbf{42} (2001), no.~2,
  818.

\bibitem{cogo}
Albachiara Cogo, \emph{Uniqueness of photon spheres and equipotential photon
  surfaces via potential theory}, 2020, MSc thesis, University of T\"ubingen
  and University of Trento.

\bibitem{Fog_Maz_Pin}
Mattia Fogagnolo, Lorenzo Mazzieri, and Andrea Pinamonti, \emph{Geometric
  aspects of {$p$}-capacitary potentials}, Annales de l'Institut Henri
  Poincar\'{e} C. Analyse Non Lin\'{e}aire \textbf{36} (2019), no.~4,
  1151--1179.

\bibitem{GalMiao}
Gregory~J. Galloway and Pengzi Miao, \emph{Variational and rigidity properties
  of static potentials}, Communications in Analysis and Geometry \textbf{25}
  (2017), no.~1, 163--183.

\bibitem{GibbonsWarnick}
Gary~W. Gibbons and Claude~M. Warnick, \emph{Aspherical photon and anti-photon
  surfaces}, Physics Letters B \textbf{763} (2016), 169--173.

\bibitem{Heus1}
Markus Heusler, \emph{On the uniqueness of the {R}eissner--{N}ordstr\"om
  solution with electric and magnetic charge}, Classical and Quantum Gravity
  \textbf{11} (1994), no.~3, L49--L53.

\bibitem{Heusler_Blackholes}
\bysame, \emph{Black hole uniqueness theorems}, Cambridge Lecture Notes in
  Physics, vol.~6, Cambridge University Press, Cambridge, 1996.

\bibitem{He}
\bysame, \emph{Stationary black holes: Uniqueness and beyond}, Cambridge
  Lecture Notes in Physics, Cambridge University Press, 1996.

\bibitem{Heus2}
\bysame, \emph{On the {U}niqueness of the {P}apapetrou--{M}ajumdar {M}etric},
  Classical and Quantum Gravity \textbf{14} (1997), no.~7, L129--L134.

\bibitem{Israel}
Werner Israel, \emph{{E}vent {H}orizons in {S}tatic {V}acuum {S}pace-{T}imes},
  Physical Review \textbf{164} (1967), no.~5, 1776--1779.

\bibitem{IsrEl}
\bysame, \emph{Event horizons in static electrovac space-times}, Communications
  in Mathematical Physics \textbf{8} (1968), no.~3, 245--260.

\bibitem{Jahns}
Sophia Jahns, \emph{Photon sphere uniqueness in higher-dimensional
  electrovacuum spacetimes}, Classical Quantum Gravity \textbf{36} (2019),
  no.~23, 235019, 24.

\bibitem{KM}
Daniel Kennefick and Niall~\'O Murchadha, \emph{{Weakly decaying asymptotically
  flat static and stationary solutions to the Einstein equations}}, Class.
  Quantum Grav. \textbf{12} (1995), no.~1, 149.

\bibitem{KW}
Marcus Khuri and Eric Woolgar, \emph{Nonexistence of {D}egenerate {H}orizons in
  {S}tatic {V}acua and {B}lack {H}ole {U}niqueness}, Physics Letters B
  \textbf{777} (2018), 235--239.

\bibitem{Koga:2020gqd}
Yasutaka Koga, \emph{{Photon surfaces as pure tension shells: {U}niqueness of
  thin shell wormholes}}, Physical Review D \textbf{101} (2020), no.~10,
  104022.

\bibitem{Krantz_Parks}
Steven~G. Krantz and Harold~R. Parks, \emph{A primer of real analytic
  functions}, second ed., Birkh\"{a}user Advanced Texts: Basler Lehrb\"{u}cher,
  Birkh\"{a}user Boston, Inc., Boston, MA, 2002.

\bibitem{bubble}
Hari~K. Kunduri and James Lucietti, \emph{No static bubbling spacetimes in
  higher dimensional {E}instein{\textendash}{M}axwell theory}, Classical and
  Quantum Gravity \textbf{35} (2018), no.~5, 054003.

\bibitem{Lee_Neves}
Dan~A. Lee and André~A. Neves, \emph{The {P}enrose inequality for
  asymptotically locally hyperbolic spaces with nonpositive mass},
  Communications in Mathematical Physics \textbf{339} (2015), no.~2, 327--352.

\bibitem{Lojasiewicz}
Stanislaw Lojasiewicz, \emph{Introduction to complex analytic geometry},
  Birkh\"{a}user Verlag, Basel, 1991, Translated from the Polish by Maciej
  Klimek.

\bibitem{Lucietti}
James Lucietti, \emph{All {H}igher-{D}imensional {M}ajumdar-{P}apapetrou
  {B}lack {H}oles}, Annales Henri Poincar\'e \textbf{22} (2021), 2437--2450.

\bibitem{MS}
Marc Mars and Walter Simon, \emph{{On uniqueness of static
  {E}instein-{M}axwell-dilaton black holes}}, Advances in Theoretical and
  Mathematical Physics \textbf{6} (2002), no.~2, 279--305.

\bibitem{MuA}
Abdul Kasem~Muhammad Masood-ul Alam, \emph{{Uniqueness proof of static charged
  black holes revisited}}, Classical and Quantum Gravity \textbf{9} (1992),
  L53--L55.

\bibitem{Milnor}
John~W. Milnor, \emph{Morse theory}, Annals of Mathematics Studies, No. 51,
  Princeton University Press, Princeton, N.J., 1963, Based on lecture notes by
  M. Spivak and R. Wells.

\bibitem{MoncriefIsemberg}
Vincent Moncrief and James Isenberg, \emph{Symmetries of cosmological {Cauchy}
  horizons with non-closed orbits}, Commun. Math. Phys. \textbf{374} (2020),
  no.~1, 145--186 (English).

\bibitem{Vol}
Volker Perlick, \emph{On totally umbilic submanifolds of semi-{R}iemannian
  manifolds}, Nonlinear Analysis-theory Methods \& Applications \textbf{63}
  (2005), e511--e518.

\bibitem{Raulot}
Simon Raulot, \emph{A spinorial proof of the rigidity of the {R}iemannian
  {S}chwarzschild manifold}, Classical and Quantum Gravity \textbf{38} (2021),
  085015.

\bibitem{Reiris}
Mart\'{\i}n Reiris, \emph{On static solutions of the {E}instein-scalar field
  equations}, General Relativity and Gravitation \textbf{49} (2017), no.~3,
  Paper No. 46, 15.

\bibitem{RobinsonReview}
David~C. Robinson, \emph{Four decades of black hole uniqueness theorems}, The
  Kerr spacetime: Rotating black holes in General Relativity (SM~Scott
  DL~Wiltshire, M.~Visser, ed.), Cambridge University Press, 2009,
  pp.~115--143.

\bibitem{Rogatko}
Marek Rogatko, \emph{Uniqueness of photon sphere for
  {E}instein-{M}axwell-dilaton black holes with arbitrary coupling constant},
  Physical Review D \textbf{93} (2016), 064003.

\bibitem{Ruback}
Peter Ruback, \emph{A new uniqueness theorem for charged black holes},
  Classical and Quantum Gravity \textbf{5} (1988), no.~10, L155--L159.

\bibitem{Shoom}
Andrey~A. Shoom, \emph{Metamorphoses of a photon sphere}, Physical Review D
  \textbf{96} (2017), no.~8, 084056, 15.

\bibitem{SimEV}
Walter Simon, \emph{{A} simple proof of the generalized electrostatic {I}srael
  theorem}, General Relativity and Gravitation \textbf{17} (1985), no.~8,
  761--768.

\bibitem{Tod}
Paul Tod, \emph{Analyticity of strictly static and strictly stationary,
  inheriting and non-inheriting {E}instein-{M}axwell solutions}, General
  Relativity and Gravitation \textbf{39} (2007), no.~7, 1031--1042.

\bibitem{Tomi2}
Yoshimune Tomikawa, Tetsuya Shiromizu, and Keisuke Izumi, \emph{On the
  uniqueness of the static black hole with conformal scalar hair}, PTEP.
  Progress of Theoretical and Experimental Physics (2017), no.~3, 033E03, 8.

\bibitem{Tomi}
\bysame, \emph{On uniqueness of static spacetimes with non-trivial conformal
  scalar field}, Classical and Quantum Gravity \textbf{34} (2017), no.~15,
  155004, 8.

\bibitem{PhysRevD.104.124016}
Naoki Tsukamoto, \emph{Gravitational lensing by a photon sphere in a
  reissner-nordstr\"om naked singularity spacetime in strong deflection
  limits}, Phys. Rev. D \textbf{104} (2021), 124016.

\bibitem{VE1}
Kumar~Shwetketu Virbhadra and George F.~R. Ellis, \emph{{Schwarzschild black
  hole lensing}}, Physical Review D \textbf{62} (2000), 084003.

\bibitem{VE2}
\bysame, \emph{{Gravitational lensing by naked singularities}}, Physical Review
  D \textbf{65} (2002), 103004.

\bibitem{Wald}
Robert~M. Wald, \emph{General relativity}, 1984 (English).

\bibitem{Wang}
Xiao~D. Wang, \emph{On the uniqueness of the {A}d{S} spacetime}, Acta
  Mathematica Sinica (English Series) \textbf{21} (2005), no.~4, 917--922.

\bibitem{Yazadjiev}
Stoytcho Yazadjiev, \emph{Uniqueness of the static spacetimes with a photon
  sphere in {E}instein-scalar field theory}, Physical Review D - Particles,
  Fields, Gravitation, and Cosmology \textbf{91} (2015), no.~12, 123013, 7.

\bibitem{YazLaz}
Stoytcho~S. Yazadjiev and Boian Lazov, \emph{Uniqueness of the static
  {E}instein–-{M}axwell spacetimes with a photon sphere}, Classical and
  Quantum Gravity \textbf{32} (2015), no.~16, 165021.

\bibitem{YazaLazov2}
\bysame, \emph{Classification of the static and asymptotically flat
  {E}instein-{M}axwell-dilaton spacetimes with a photon sphere}, Physical
  Review D \textbf{93} (2016), no.~8, 083002, 11.

\bibitem{Yoshino}
Hirotaka Yoshino, \emph{Uniqueness of static photon surfaces: {P}erturbative
  approach}, Physical Review D \textbf{95} (2017), 044047.

\bibitem{Yoshino2}
\bysame, \emph{Distorted static photon surfaces in perturbed
  {R}eissner-{N}ordstr\"om spacetimes},
  \href{https://arxiv.org/abs/2309.14318}{https://arxiv.org/abs/2309.14318},
  2023.

\bibitem{MRS}
Henning~M\"uller zum Hagen, David~C. Robinson, and Hans~J. Seifert,
  \emph{{B}lack {H}oles in static electrovac space-times}, General Relativity
  and Gravitation \textbf{5} (1974), no.~1, 51--72.

\end{thebibliography}

\end{document}